\documentclass[a4paper,11pt]{article}

\usepackage{dsfont}
\usepackage{authblk}
\usepackage{yfonts}

\usepackage[utf8]{inputenc}
\usepackage{graphicx}        
\usepackage[dvipsnames]{xcolor}
\usepackage{amsfonts}
\usepackage{eufrak}
\usepackage{amsmath}
\usepackage{amssymb}
\usepackage{amsthm}
\reversemarginpar

\usepackage{bm}
\newcommand\normrho[1]{\| #1 \|}
\newcommand\norm[1]{\left|#1\right|}
\newcommand\normu[1]{| #1 |_{\mathbb{R}^{p+q}}}
 
\def\mkill{s}

\def\fptcan{\fpt^{\Psi}}
\def\sptcan{\spt^{\Psi}}

\def\AsHone{$\mathrm{H}_1$}
\def\AsSone{$\mathrm{S}_1$}

\def\Qpol{\Polgen}

\def\qfonetheta{q_1^{\Psi}}
\def\qftwotheta{q_2^{\Psi}}

\def\Y{\mathcal Y}
\def\nur{\mu}
\def\V{\mathcal V}

\def\arho{\alpha_u}
\def\aw{\alpha_w}
\def\zerorho{0_p}
\def\Bp{B^p}
\def\rad{a_0}
\def\z{w}

\def\normx{\normrho{x}}
\def\normr{\norm{x}}

\def\absx{\mbox{abs}(x)}
\def\costwophi{\frac{x^2-y^2}{\normx^2}}
\def\sintwophi{\frac{2xy}{\normx^2}}
\def\partialr{\hat n}
\def\ngamma{\hat\varrho}
\def\nzeta{\hat z}
\def\oaxial{\iota}
\def\starsphere{\star_{\mathbb{S}^2}}
\def\tz{w^u}
\def\bU{V}
\def\H{k}
\def\balpha{\alpha}
\def\bbeta{\beta}
\def\boaxial{\bm{\overline\oaxial}}
\def\baxial{\bm{\overline\axial}}
\def\talpha{\tilde\alpha}

\def\flemma{\gamma}
\def\tlemma{\flemma}
\def\tX{\tilde\upsilon}
\def\tXX{\tilde\zeta}
\def\trho{\eta}
\def\qtheta{q}
\def\Phizero{\Gamma}
\def\Zzero{Z_\Phizero}
\def\Zphi{Z_{\Phi}}

\newcommand\tenstt[1]{\mathring #1}
\newcommand\vect[1]{{#1}_{\parallel}}

\def\vort{\nu}

\def\xc{x}
\def\yc{y}
\def\zc{z}

\def\x{{x_1}}

\newcommand{\trace}[1]{\textswab{#1}}

\def\traceA{\trace{A}}
\def\C{\traceA_1}
\def\A{\traceA_2}
\def\AT{A_2}
\def\B{\traceA_3}
\def\BT{A_3}
\def\DD{\traceA_{\rho u}}
\def\E{\traceA_{\axial u}}
\def\QQ{\traceA_{uv}}
\def\QQz{\traceA_{zz}}

\def\Polgen{\mathcal{P}}
\def\Polz{\widetilde{\mathcal{P}}}

\def\BB{U^3}

\newcommand{\ot}[1]{\mathfrak{#1}}
\def\otp{{\ot{c}}}
\def\otm{{\ot{a}}}
\def\otn{{\ot{b}}}

\def\kvi{\mathfrak{i}}
\def\kvj{\mathfrak{j}}
\def\kvk{\mathfrak{l}}

\def\paramp{\mathfrak{B}}
\def\parampp{\mathfrak{C}}

\def\fin{\hspace*{\fill} \rule{2.5mm}{2.5mm}\\ \vspace{0mm}} 

\def \ge {g_\pertp}
\def \gp {K_1}
\def \gpp {K_2}

\def\mmm{M}

\def\lie{\mathcal{L}}

\def\pertp{\varepsilon}
\def\gfam{\hat g}
\def\fpt{\gp}
\def\spt{\gpp}
\def\Kper{\fpt}
\def\Kperper{\spt}

\def\axial{\eta}
\def\stat{\xi}
\def \Wtwo {{\mathcal{W}}}

\def\defi{:=}

\def\D{\overline{D}{}}
\def\dsph{\overline{\d}}
\def\gsph{{g_{\mathbb{S}^2}}}

\def\centresph{\mathcal{C}}

\def\axis{\mathcal{A}}
\def\UinM{\mathcal{U}}
\def\Uaxis{\mathcal{U}_{\axis}}

\def\d{d}

\def\opertbase{\omega}
\def\opert{{\omega^{(1)}}}

\def\foh{h^{(1)}}
\def\fom{m^{(1)}}
\def\fok{k^{(1)}}
\def\fof{f^{(1)}}
\def\sow{\omega^{(2)}}
\def\sowf{\mathcal{W}} 

\def \soh{h^{(2)}}
\def \som{m^{(2)}}
\def \sok{k^{(2)}}
\def \sof{f^{(2)}}

\def\X{{\mathfrak{X}}}
\def\sper{V_1}
\def\sperper{V_2}
\def\K{K}
\def\s{s}
\def\Al{{\mathds{A}}}
\def\lam{\lambda}

\def\la{\langle}
\def\ra{\rangle}

\def\Q{\mathcal{R}}
\def\UR{\chi}
\def\VR{\upsilon}
\def\fy{{Y^3}}
\def\guepsilon{g^{\sharp}_{\pertp}{}}
\def\gepsilon{g_{\pertp}}

\newtheorem{theorem}{Theorem}[section]
\newtheorem{proposition}[theorem]{Proposition}
\newtheorem{lemma}[theorem]{Lemma}
\newtheorem{corollary}[theorem]{Corollary}
\newtheorem{corolemm}[theorem]{Corollary}
\theoremstyle{definition}
\newtheorem{definition}[theorem]{Definition}
\newtheorem{remarkdef}[theorem]{Remark}
\newtheorem{remarklem}[theorem]{Remark}
\newtheorem{remarkpro}[theorem]{Remark}
\newtheorem{remark}[theorem]{Remark}
\numberwithin{equation}{section}

\begin{document}
\title{Gauge fixing and regularity of axially symmetric and axistationary second order perturbations around spherical backgrounds
}
\author[1]{Marc Mars}
\author[2]{Borja Reina}
\author[2]{Ra\"ul Vera}
\affil[1]{Department of Fundamental Physics and\protect\\
Institute of Fundamental Physics and Mathematics, University of Salamanca}
\affil[2]{Department of Theoretical Physics and History of Science,\protect\\ University of the Basque Country UPV/EHU}
\date{}
\maketitle

\begin{abstract}
Perturbation theory in geometric theories of gravitation is a gauge theory of symmetric tensors defined on a Lorentzian manifold (the background spacetime).
The gauge freedom
makes uniqueness problems in perturbation theory particularly hard as one needs to understand in depth the
process of gauge fixing before attempting any uniqueness proof.
This is the first paper of a series of two aimed at deriving an existence and uniqueness result for 
rigidly rotating stars  to second order in perturbation theory in General Relativity.
A necessary step is to show the existence of a suitable choice
of gauge and to understand the differentiability and regularity properties of the resulting gauge
tensors in some ``canonical form'', particularly at the centre of the star.
With a wider range of applications in mind, in this paper we analyse the fixing and regularity problem in a more general setting.
In particular we tackle the problem 
of the Hodge-type decomposition into scalar, vector and tensor components
on spheres of symmetric and axially symmetric tensors
with finite differentiability down to the origin,
exploiting a strategy
in which the loss of differentiability is as low as possible.
Our primary interest, and main result, is to show that stationary and axially symmetric second order
perturbations around static and
spherically symmetric background configurations can indeed be rendered
in the usual ``canonical form'' used in the literature while loosing \emph{only} one
degree of differentiability and keeping all relevant quantities
bounded near the origin.

\end{abstract}

\tableofcontents

\section{Introduction}

Perturbation theory in metric theories of gravity is one of the fundamental tools to tackle many  realistic problems in relativistic astrophysics, ranging from slowly rotating stars to the emission of gravitational waves from
binary systems  in certain limits. Perturbation theory is, in essence, a theory of symmetric tensors defined on a Lorentzian manifold (the background spacetime). From a structural point of view its main particularity is that the theory is not only covariant, but also gauge invariant. If the perturbation theory is
developed to order $k$,  the number of symmetric tensors $K_N$ is also $k$
and the gauge freedom involves $k$ vector fields (see \cite{Bruni_et_al_1997} for
the explicit gauge transformation law at every level $k$).

The gauge freedom is at the same time a feature and a nuisance of the theory. Among its positive consequences,  the gauge freedom can often  be exploited to simplify the problem under consideration (in much the same way as in electromagnetism). On the other hand, the gauge freedom is always there, so any solution of a problem immediately gives rise to the whole class of gauge related solutions that are, a priori, equally valid. An immediate consequence is that uniqueness problems in perturbation theory become much harder, since one needs to understand in depth the process of gauge fixing before attempting any uniqueness proof.

We encountered this difficulty in full when we started the project of proving
a rigorous existence and uniqueness result for slowly rotating stars to second order in perturbation theory. As already mentioned, one of the necessary steps was to show the existence of a suitable choice of gauge and to understand the
differentiability and regularity properties of the resulting gauge fixed tensors, particularly at the centre of the star. Despite the vast literature available on choice of gauges, specifically in our setup of perturbations around spherically symmetric backgrounds, this problem had never been addressed before rigorously. It turned out that the problem is considerable harder than one could have expected a priori. In this paper we report on our results on this subject. Although our primary motivation for this work still lies on the existence and uniqueness problem for slowly rotating stars, the existence of gauges that we analyse in this paper have a much wider range of applicability and are interesting on their own,
independently of the original application we have in mind.  This, combined with the length and level of complication we have encountered, justifies presenting the results in a separate paper.

Our specific interest is to understand the problem of gauge fixing and the properties of the resulting ``canonical form'' for stationary and axially symmetric perturbations around static and spherically symmetric background configurations with a regular centre.  We restrict the perturbations to the so-called orthogonally transitive case and we go to second order in perturbation theory. It is important to emphasize that, for the sake of generality and particularly to apply the results in our subsequent work on slowly rotating stars, we need to
work with finite differentiability, and in fact we want to keep our differentiability requirements as low as possible. This requirement is one of the main sources of complication in our arguments.

We tackle the problem in two separate steps, each of which requires fewer assumptions on the background. The first step is concerned with orthogonally transitive stationary and axially symmetric perturbations. Here the background need not admit any extra symmetry, i.e. the results apply for general backgrounds admitting a stationary and axially symmetric orthogonally transitive action.
This step is not particularly complicated and we deal with it in Section \ref{sec:OT_perturbations}.
We first analise the case of general backgrounds admitting an orthogonally transitive Abelian group action of any dimension and perturbations up to second order that ``inherit'' these background symmetries
(the precise definition of this notion is given  in Definition \ref{def:inherit}).
However, it is only in the case of orthogonally transitive stationary and axially symmetric perturbations that we can ensure the gauge transformation keeps the differentiability
of the perturbation tensors \emph{also} on the axis. The main result for this step is Proposition \ref{Block} where \emph{block-}canonical forms for the first and second order perturbations tensors are given.  

The second step is considerable harder. Here the background is assumed to be spherically symmetric (but not necessarily static).
Perturbations around spherically symmetric backgrounds have been studied extensively in the literature
 and in many different areas
(see e.g. \cite{LRR_Kokkotas_Schmidt1999,brien2004,TR-Berti-Cardoso-Starinets2009,Brizuela2010,
llibre-RotRelStars2013} and references therein).
A very common choice of gauge is to assume that the  angular-angular part of the perturbation tensor is proportional to the standard metric on the sphere. This is for instance one of the defining properties
of the Regge-Wheeler (RW) gauge \cite{Regge-Wheeler}, but it is shared by many other gauge fixing procedures. By far the argument  most widely used to justify that such a choice of gauge is possible is to decompose the perturbation tensor into scalar, vector and tensor spherical harmonics. Then, for each mode it is easy to construct a gauge vector that transforms the perturbation tensor into the desired form. Despite its simplicity, this argument falls short to provide an existence proof of the gauge vector because that would require showing that the collection of  gauge vectors at each mode corresponds to the mode decomposition of a gauge vector. In other words, one must show that the mode series converges. This is not a simple problem.

The second approach is based on using the Hodge-type scalar-vector-tensor (SVT) decomposition
of symmetric tensors on the sphere \cite{york1973}.
The approach of replacing the spherical harmonic
mode decomposition by SVT decompositions
to study perturbations around spherical background
has been used in the literature (see e.g. \cite{Dotti2014,ishibashi-kodama2011}
and \cite{MarsMenaVera2007},
where the full set of perturbations are expressed in terms of functions on the sphere).
Its use to show existence of a suitable gauge vector in four spacetime dimensions and for first order perturbation tensors can be summarized as follows.  For spacetime dimension two, the
SVT decomposition is applied on two-dimensional spheres and it is a well-known fact
that in the two-dimensional sphere the only traceless and transverse symmetric tensor
is the zero tensor. Thus, the SVT decomposition takes a simple form that
involves only a scalar and a vector field
(the explicit form appears in \eqref{eq:spherical_o} below). This, combined
with the gauge transformation law, makes it immediate to show that a gauge vector exists such that  the angular-angular part of the first order perturbation tensor can be made proportional to the standard metric on the sphere. This approach however does not cover all our needs either, even at the first order level. The main difficulty lies at the centre, i.e. at the point(s) where the spheres
defined as the surfaces of transitivity of the spherical action on the spacetime degenerate to points. The SVT decomposition is well-understood on each sphere, but we need to deal with a two-parameter family of spheres that {\it degenerate to a point}. This prevents us from using directly the standard results on SVT decomposition to show existence of the appropriate gauge vector. It should be emphasized, however, that the SVT decomposition is a very important guiding principle for our approach to the problem.

Another important source of complication is our need to use finite differentiability. As we shall see along the text, it is a fact that rendering the perturbation tensor into a canonical form typically lowers the differentiability. Given that we want to keep the differentiability requirements as low as we can, it becomes necessary to find a good strategy where the loss of differentiability is as low as possible.   This prevents us even from adapting directly the standard methods on Hodge decomposition on the sphere in the domain away from the origin. 
Indeed, these methods (see e.g. \cite{IshibashiWald2004}) work by obtaining second order elliptic equations
for each one of the components arising in the Hodge decomposition. To derive these equations, one needs to take two derivatives of the original tensor.
These derivatives are then regained by standard elliptic regularity (working e.g. in H\"older spaces). However, when dealing with a two parameter family of problems as in our case, the loss of two derivatives in the coordinates that label
the spheres cannot be regained by elliptic regularity.

The method we follow
(see Theorem \ref{res:decomp_theorem}) consists in writing directly a system of coupled first order PDE on each sphere.
While we are not aware of any general theorem that gives existence, we can exploit the axial symmetry of the perturbations to
transform the system of PDE into a decoupled system of ODE.
This strategy allows us to achieve a loss of only one derivative (away from the origin). Although we have no proof that this loss is optimal, we do have strong indication that it cannot be improved in general. The method that we follow 
introduces an important complication at the axis of symmetry where the ODEs become singular. In fact, most of the technical work in this paper is devoted to understanding the existence and regularity at the axis of the
solutions of these ODE as well as to understand the regularity with respect to
transversal directions away from the two-sphere and, very particularly, the
behaviour of the solutions near the origin. We devote Appendix \ref{app:iota} to
study all these issues.

It turns out that the behaviour near the origin is complicated. Our main result in this respect is that all the relevant
quantities stay bounded near the origin. However, we do not show that the perturbation tensor is even
continuous at the origin (let alone differentiable). Again we have no proof that our result is optimal, but we strongly suspect that it is not possible to write the perturbation tensor in canonical form and not to lose a great deal of regularity at the origin.

This is in fact one of the points we want to stress in this paper. In the physics literature it is a rather common believe that, as long as the number of restrictions matches the number of free functions in the gauge transformation, 
the process of restricting the gauge and writing the resulting tensor in some predetermined form comes at essentially no cost. The analysis in this paper shows very clearly that these issues are very delicate and that exploiting the gauge in order to transform the perturbation tensors into some useful form may easily spoil some other desired properties (such as continuity or
differentiability at certain places). Only by knowing precisely  how much deterioration is generated, can one decide whether using the canonical form is convenient (or even possible) for the specific problem under consideration.

\subsection{Main result}
The main result, Theorem \ref{theo:main}, can be stated roughly as follows.
Consider a static and spherically symmetric spacetime $(\mmm,g)$
with $g$ of class $C^{n+1}$ with $n\geq 2$, with timelike integrable
KVF $\stat$, and single out a generator of an axial symmetry $\axial$, so that
\[
g   = - e^{\nu(r)} dt^2 + e^{\lambda(r)} dr^2 + \Q^2(r) \left ( d\theta^2
+ \sin^2 \theta d \phi^2 \right ),\qquad \stat=\partial_t,\quad \axial=\partial_\phi.
\]
Now, given any stationary ($\stat$) and axially ($\axial$) symmetric (and orthogonally transitive)
$C^{n+1}$ perturbation to second order around  $(\mmm,g)$ there exists a gauge transformation that yield
first and second order perturbation tensors $\fptcan$ and $\sptcan$
that are $C^{n-1}$ and $C^{n-2}$ outside the origin, respectively,
and can be written as
\begin{align*}
\fptcan =& -4 e^{\nu(r)} \foh(r, \theta) dt^2
-2  \opert(r,\theta) \Q^2(r)\sin^2\theta dtd\phi
+ 4 e^{\lambda(r)} \fom(r, \theta) dr^2 \nonumber\\
&+4   \fok(r, \theta)\Q^2(r)(d\theta^2+ \sin ^2 \theta  d\phi^2)
+ 4e^{\lambda(r)}\partial_\theta \fof(r,\theta)\Q(r) dr d\theta ,
\\
\sptcan =& \left(-4 e^{\nu(r)} \soh(r, \theta) + 2{\opert}^2(r, \theta) \Q^2(r) \sin ^2 \theta \right)dt^2\nonumber\\
&-2 \sow(r,\theta) \Q^2(r) \sin^2\theta  dt d\phi
+ 4 e^{\lambda(r)} \som(r, \theta) dr^2 \nonumber\\
&+4  \sok(r, \theta) \Q^2(r) (d\theta^2+\sin ^2 \theta  d\phi^2)
+ 4e^{\lambda(r)}\partial_\theta \sof(r,\theta)\Q(r)dr d\theta
\end{align*}
outside the axis of symmetry.
Moreover, the result provides full control of the differentiability and boundedness
properties of the functions involved.
Let us stress again the fact that this result
does not ensure
the continuity of either tensor $\fptcan$ or $\sptcan$ at the origin.
We can prove, however, that the gauge vectors extend continuously to zero at the origin.

The gauge freedom involved in the above forms is found and discussed in Section \ref{sec:gauges}.

\subsection{Plan of the paper}
The structure of the paper is as follows. 
Section \ref{sec:pert_scheme} is devoted to produce the necessary definitions
on perturbations that inherit some of the symmetries present in the background.
It serves us also to fix the differentiability  of the perturbation
scheme and the perturbation tensors.
In Section \ref{sec:axial} we analyse the structure
that symmetric 2-covariant tensors
invariant under the axial symmetry must have in a convenient (partly Cartesian)
class of charts in the presence of the axis. The results are presented in Lemma \ref{res:axial},
which generalises the well known results on the form of the metric
in axially symmetric spaces, see e.g. \cite{Carot2000}.
In parallel, Section \ref{sec:OT_perturbations} deals with
orthogonally transitive perturbation schemes, that is,
perturbations that inherit the two-dimensional group of isometries
acting orthogonally transitively admitted by the background (but arbitrary otherwise),
to second order. The result for sationary and axisymmetric orthogonally
transitive spacetimes is given in Proposition \ref{Block}.

Next we retake the results from Section \ref{sec:axial}
and particularise to axial perturbations around spherically symmetric backgrounds.
In particular, we prove in Theorem \ref{res:decomp_theorem}
the existence of the 
decomposition on the sphere of symmetric axially symmetric tensors (of finite differentiability),
down to the behaviour of the decomposition at the origin.
That result is then (partially) used to prove Proposition \ref{prop:K_spher},
which states the existence of a gauge vector that renders the first order
perturbation tensor in some convenient form, while keeping control of the differentiability
properties and behaviour at the origin of the relevant quantities.
The analogous, but much more involved result, for second order is presented in Proposition \ref{prop:K2_spher}.

We finally combine in Section \ref{sec:main} all those results
to build the proof of the main results of this
paper, in the form of Proposition \ref{prop:pre_main}
leading to Theorem \ref{theo:main}.

Let us stress that our work here is purely geometric, we do not make
use of any field equations.  For the same reason, we do not make any
consideration as to the physical meaning of the perturbation.

We have tried to write down this work as self-contained as possible, leaving the
more technical work for the Appendices.
The control of the differentiability (specially on the axis)
and boundedness near the
origin of the relevant components of the perturbation tensors
requires several results on
radially symmetric functions which we state and prove in 
Appendix \ref{app:diff_origin}. Although these results should be essentially known, they
are not easily found in the literature in the form we need.
Finally, the building block in showing existence of gauges
is Lemma \ref{res:S_c4_lemma}. Establishing this result requires some rather long technical work which is left to Appendix \ref{app:iota}.

\subsection{Notation}
Given the various setups considered in this work,
we have been compelled to introduce a substantial amount of notation.
We will fix most of the notation along the way, fundamentally at the start
of the sections in which the relevant frameworks are introduced,
in particular in Section \ref{sec:axis_on_sph} and Appendix \ref{app:iota}.
Nevertheless, we fix here some basic notation that will be used from the start.

A $C^{n+1}$ spacetime $(M,g)$
is a $k$-dimensional ($k\geq 2$) orientable $C^{n+2}$ 
manifold $M$ endowed with a time-oriented Lorentzian metric $g$ of class
$C^{n+1}$. We assume $n \geq 2$ unless otherwise stated. Scalar products of two vector fields $X$, $Y$
with the
metric $g$ will be denoted both by
$g(X,Y)$ and $\la X,Y\ra$. We say that a geometric object is ``smooth''
when it has maximum differentiability allowed by the background.
A function $f$ defined on an open dense subset $\mathcal{U}'$
of some neighbourhood $\mathcal{U}\in M$
is said to be $C^m(\mathcal{U})$ if it can extended to all $\mathcal{U}$ with this property.

We will use the usual square bracket notation $[m]$ for the integer part of $m \in \mathbb{R}$.
We also use Landau's  big-$O$ and little-$o$ notation with its standard meaning.

\section{Definition of perturbation scheme and symmetry preserving perturbations}
\label{sec:pert_scheme}
The construction of a spacetime perturbation relies on a one-parameter family of $C^{n+1}$ ($n \geq 2$) spacetimes
$(\mmm_\pertp, \gfam_\pertp)$, where $\pertp$ takes values
in an open interval $I_0 \subset \mathbb{R}$  containing zero,
from where we single out the background
$(\mmm,g) \defi (\mmm_0, \gfam_0)$, diffeomorphically identified through a gauge $\psi_\pertp$
so that
\begin{equation}
	\psi_\pertp : \mmm \rightarrow \mmm_\pertp,\label{diffeom:spacetime}
\end{equation}
and $\psi_0$ is the identity.
The diffeomorphisms $\psi_\pertp$ are assumed to be $C^{n+2}$ for each $\pertp$.
This  allows us to define a family of metrics
$\ge$ of class $C^{n+1}$ on $\mmm$ related to $\gfam_\pertp$ by
$\ge := \psi_\pertp^*(\gfam_\pertp)$. We further assume that this family
of metrics is at least $C^2$ in $\pertp$ (to guarantee we can go to second order)
and that $\pertp$-derivatives do not affect the differentiability
  class. 
Define also the tensor $\guepsilon^{\mu\nu}$ by $\guepsilon\defi 
d\psi_\pertp^{-1}(\gfam^{\sharp}_\pertp)$
where $\gfam^{\sharp}_{\pertp}$ is the contravariant metric associated to
$\gfam_{\pertp}$.
Note that $\guepsilon$ is the contravariant metric associated to $\gepsilon$. 

The first and second order perturbation tensors $\fpt$ and $\spt$ on $(\mmm,g)$ are obtained
from $\ge$ as follows
\begin{equation}
	\fpt := \left.\frac{d \ge}{d \pertp}\right|_{\pertp=0}, \quad \spt := \left.\frac{d^2 \ge}{d \pertp^2}\right|_{\pertp=0}, \label{def:metricperturbations}
\end{equation}
while the derivatives of the contravariant metrics are
\footnote{All Greek indices in this paper are raised and lowered with the background metric  $g$ and its inverse.}
\begin{align}
\left . \frac{d \guepsilon^{\alpha\beta}}{d \pertp} \right |_{\pertp=0}
= - \Kper^{\alpha\beta}, \quad \quad
\left . \frac{d^2 \guepsilon^{\alpha\beta}}{d \pertp^2} \right |_{\pertp=0}
  = - \Kperper^{\alpha\beta} + 2 \Kper^{\alpha\mu} \Kper{}_{\mu}^{\phantom{\mu}\beta}.
  \label{eq:gsharppert}
\end{align}
The identification by $\psi_\pertp$ is highly non-unique and its freedom
can be realized by
taking into consideration an $\pertp$-dependent diffeomorphism 
$\Omega_\pertp:\mmm\to\mmm$
in $\mmm$ before applying $\psi_\pertp$. The new identification is
$\psi^{g}_{\pertp} \defi \psi_\pertp\circ\Omega_\pertp$ and introduces a new family of 
tensors $\ge^{g}=\psi^{g}_\pertp{}^*(\gfam_\pertp)=\Omega^*_\pertp(\ge)$ on $\mmm$ with corresponding first and second order
perturbation tensors $\fpt^{g}$ and $\spt^{g}$. We again assume that
$\Omega_{\pertp}$ are $C^{n+2}$ diffeomorphisms with $C^2$ dependence in
$\pertp$ and that $\pertp$-derivatives do not change the
  differentiability class.
In terms of the first and second order (spacetime)
\textit{gauge vectors} $\sper$ and $\sperper$, defined as follows
\begin{eqnarray}
\sper &:=& \left. \frac{\partial \Omega_\pertp} {\partial \pertp} \right|_{\pertp=0}, \nonumber\\
       \sperper &:=&\left. \frac{\partial {V}_\pertp}{\partial \pertp}\right|_{\pertp=0}, \qquad               {V}_\pertp := \left.\frac{\partial (\Omega_{\pertp + h} \circ \Omega_\pertp^{-1})}{\partial h}\right|_{h=0},
       \label{def:sp_gauge_vectors}
   \end{eqnarray}
the relation between $\fpt^{g}$, $\spt^{g}$ and $\fpt$, $\spt$ is given by \cite{Bruni_et_al_1997,Mars2005}
\begin{align}
  {\fpt^{g}} &={\fpt} + \lie_{\sper} g,\label{gaugeper}\\
  {\spt^{g}} &= {\spt} + \lie_{\sperper} g + 2\lie_{\sper}{\fpt}^g
- \lie_{\sper} \lie_{\sper} g.
\label{gaugeperper_bis}
\end{align}
Since the background manifold is $C^{n+2}$ and the metric is $C^{n+1}$, the
natural differentiability class preserved by these gauge transformations is
as follows. At first order the gauge vector  $\sper$ is $C^{n+1}$ and
the perturbation tensor $\Kper$ is $C^{n}$. At second order $\sperper$ is $C^n$ and $\Kperper$ is $C^{n-1}$.
We will therefore incorporate this assumption into  our definitions:

\vspace{3mm}

A \emph{perturbation} is, by definition, the family $(\mmm_{\pertp},\hat{g}_{\pertp})$
with all the possible identifications $\{\psi_{\pertp}\}$ 
related to each other by a gauge transformation.  The intrinsic 
gauge freedom of a perturbation is a major source of complication, as a specific problem may have very different
forms in different gauges. A selection of a class of gauges leads to 
what we call a 
\emph{perturbation scheme}, namely a
triple $(\mmm_\pertp, \gfam_\pertp,\{\psi_\pertp\})$,
where $\{\psi_\pertp\}$
denotes the selected class of 
gauges.
A \emph{perturbation scheme} will be said  of class $C^{n+1}$ when
  the family $\gfam_\pertp$ is $C^{n+1}$, the perturbation tensors $\Kper$, $\Kperper$ are, respectively,
  $C^{n}$ and $C^{n-1}$ and the gauge vectors $\sper$, $\sperper$ are, respectively, $C^{n+1}$ and $C^{n}$.

In many specific problems one often needs to preserve some of the symmetries of the background along the perturbation. This leads to the notion of 
``inheritance of symmetries'' which we define next.

\begin{definition}
\label{def:inherit}
Let $(\mmm_\pertp, \gfam_\pertp,\{ \psi_\pertp \})$ be a perturbation scheme
whose background
spacetime  $(\mmm,g)$ admits a Killing vector field $\xi$.
The perturbation scheme is said to {\bf inherit the (local) symmetry generated by $\xi$} 
whenever for all $\pertp\in I_0$  and
all $\psi_{\pertp} \in \{ \psi_{\pertp} \}$, the vector field
$\hat{\xi}_\pertp\defi d\psi_\pertp(\xi)$  is a Killing vector of
$(\mmm_\pertp, \gfam_\pertp)$.
\end{definition}
\begin{remarkdef}
Of course, this definition just recovers the usual idea that the perturbation admits a symmetry when so does the family
$\ge$, since
$$\lie_{\xi} \ge=\psi_\pertp^*(\lie_{\hat{\xi}_\pertp}\gfam_\pertp)=0.$$
Note that the notion of ``inheriting a (local) isometry'' depends not only on the  perturbation itself, but also on the perturbation scheme. Indeed, given
a perturbation scheme $(\mmm_\pertp, \gfam_\pertp,\{ \psi_\pertp \})$ one may
construct
other perturbation schemes belonging to the same perturbation where 
the symmetry is not inherited. This is because a fixed vector
field $\xi$ which is Killing  for all
$g_{\pertp}$ will not, in general, be a Killing  of
$\Omega_{\pertp}^{\star} (g_{\pertp})$.
Thus, demanding the existence of a perturbation scheme where a symmetry
is inherited is useful both to restrict geometrically 
the family of perturbations and {\it also} to restrict the class of
allowed gauges. 
\end{remarkdef}
\begin{remarkdef}
If a perturbation scheme inherits a collection of
 (local) isometries of the background that form a subalgebra $\Al_0$,
then each $(\mmm_\pertp, \gfam_\pertp)$ admits the same algebra, because
push-forwards preserve  commutation relations \cite{Schouten1954}. 
For the same reason, $\Al_0$ must leave invariant the family of tensors $\ge$ on
$\mmm$.
\end{remarkdef}

Besides inheriting (local) isometries, 
in many cases some other geometrical aspects concerning the orbit of the (local)
group are also required
to be preserved. In the following two
sections we consider axially symmetric perturbations
and ``orthogonal transitive'' perturbations  independently.
Our final aim will be the construction of a convenient perturbation scheme
for stationary and axially symmetric perturbations that inherit also
the geometric property of being ``orthogonal transitive.''

\section{Axially symmetric perturbations}
\label{sec:axial}

In this section we recall the concept of axial symmetry and introduce the definition of axially symmetric perturbations as a particular instance of  the general notion of ``symmetry inheritance''. Our main result of the section is Lemma \ref{res:axial} where we explore the consequences of axial symmetry on the structure of symmetric two-covariant tensors. This result will play a relevant role in subsequent sections, where the background is taken to be spherically symmetric.

The definition of axial symmetry is standard (see e.g. \cite{carter_axial,marcsenoaxial,alanaxialcomment}):

\begin{definition}
  \label{def:axial}
	A spacetime $(\mmm,g)$ 
	is axially symmetric whenever there is an effective realization of the
	one-dimensional torus $T$ into $\mmm$ that is an isometry and such that the set of fixed points
	is non-empty.
\end{definition}

We denote by $\axial$ the Killing vector field defined by this
realization assuming that the torus $T$ has been parametrized with the standard  $2 \pi$-periodicity angle.
First consequences from the definition are that the set of fixed points, where $\axial=0$
and which we call the axis $\axis$, is a codimension two
Lorentzian and time-oriented surface \cite{carter_axial,marcsenoaxial}.
Furthermore, $\axis$ is autoparallel and for any point $p\in \axis$ there is a neighbourhood
of $p$ such that $\trho^2\defi \la \axial,\axial\ra$ is non-negative and zero only at points on
the axis (a priori this property may fail sufficiently far away from the axis). Moreover
\begin{equation}
\lim_{\trho^2\to 0}\frac{\la \nabla\trho^2,\nabla\trho^2\ra}{4\trho^2}=1.
\label{regular_axis}
\end{equation}
This is the so-called \emph{regular axis property}, from where the usual \emph{elementary flatness}
around the axis can be inferred.

We can particularize Definition \ref{def:inherit} to the case of axial symmetry and introduce the notion of axially symmetric perturbation scheme. To be explicit:

\begin{definition}
  \label{def:ASP}
  A perturbation scheme $(\mmm_\pertp,\gfam_\pertp,\{\psi_\pertp\})$
  of an axially symmetric background spacetime $(\mmm,g)$ with axial vector $\axial$
  is an \textbf{axially symmetric perturbation scheme} if
  it inherits the axial symmetry in the sense of  Definition \ref{def:inherit}.
  \end{definition}

In this setup
$\hat \axial_\pertp=d\psi_\pertp(\axial)$ is a Killing vector for each
$(\mmm_\pertp,\gfam_\pertp,\{\psi_\pertp\})$,
with axis $\hat\axis{}_{\pertp}$ defined by the points at which
$\hat \axial_\pertp=0$,
i.e. $\hat\axis{}_{\pertp}=\{p_\pertp \in \mmm_\pertp;\hat \axial_\pertp|_{p_\pertp}=
0\}$.
By the invertibility of $\hat \axial_\pertp=d\psi_\pertp(\axial)$
the points $p_\pertp$ that satisfy $\hat \axial_\pertp|_{p_\pertp}=0$
are those $p_\pertp=\psi_\pertp(p)$ such that $\axial|_p=0$,
and therefore $\psi_\pertp$ simply maps the axis at the background $\axis$
to their corresponding $\hat\axis{}_{\pertp}$.

Since the Killing equations $\lie_\axial g_\pertp=0$ hold for all $\pertp$,
we necessarily have, up to second order,
\begin{equation}
\lie_\axial \Kper=0,\qquad
\lie_\axial \Kperper=0.
\label{eq:axial_Ks}
\end{equation}

The regular axis property (\ref{regular_axis}) holds at each $\pertp$-component,
that is, the function
\[
\hat{\varLambda}_\pertp\defi \frac{\hat g_\pertp^{\alpha\beta}\partial_\alpha(\hat\trho_\pertp^2) \partial_\beta(\hat\trho_\pertp^2)}{4\hat\trho_\pertp^2},
\]
where $\hat \trho_\pertp^2\defi \gfam_\pertp (\hat\axial_\pertp,\hat\axial_\pertp)$,
must be 1 at $\hat\axis{}_{\pertp}$. Therefore, by construction, the pullback
$\varLambda_\pertp\defi\psi^*_\pertp(\hat\varLambda_\pertp)=g_\pertp^{\sharp\alpha\beta}\partial_\alpha(\trho_\pertp^2) \partial_\beta(\trho_\pertp^2)/4\trho_\pertp^2$, where
$\trho^2_\pertp\defi\psi^*_\pertp(\hat\trho_\pertp^2)=g_\pertp(\axial,\axial)$,
must attain 1 at $\axis$, i.e.
\[
\lim_{\trho^2\to 0}\varLambda_\pertp=1.
\]
Since $\lim_{\trho^2\to 0}\varLambda_0=1$ (regular background configuration),
the \emph{regular axis property} on the perturbation scheme translates, to second order,  onto
the fact that
\[
\varLambda^{(1)}\defi\left.\frac{d\varLambda_\pertp}{d\pertp}\right|_{\pertp=0},\qquad
\varLambda^{(2)}\defi\left.\frac{d^2\varLambda_\pertp}{d\pertp^2}\right|_{\pertp=0},
\]
satisfy
\begin{equation}
\lim_{\trho\to 0}\varLambda^{(1)}=0,\qquad
\lim_{\trho\to 0}\varLambda^{(2)}=0.
\label{regular_axis_e}
\end{equation}
Since
\[
\left.\frac{d\trho_\pertp^2}{d\pertp}\right|_{\pertp=0}=\Kper(\axial,\axial),\qquad
\left.\frac{d^2\trho_\pertp^2}{d\pertp^2}\right|_{\pertp=0}=\Kperper(\axial,\axial),
\]
and recalling (\ref{eq:gsharppert}) a straightforward calculation shows
\begin{equation}
\varLambda^{(1)}=\frac{1}{4\trho^2}\left(-\Kper(\nabla\trho^2,\nabla\trho^2)+2\la \nabla(\Kper(\axial,\axial)),\nabla\trho^2\ra-\frac{1}{\trho^2}\Kper(\axial,\axial)\la \nabla\trho^2,\nabla\trho^2\ra\right)
\label{eq:varlambda1}
\end{equation}
and
\begin{align}
\varLambda^{(2)}=&\frac{1}{4\trho^2}\left\{-\Kperper(\nabla\trho^2,\nabla\trho^2)
+2\la \nabla(\Kperper(\axial,\axial)),\nabla\trho^2\ra-\frac{1}{\trho^2}\Kperper(\axial,\axial)\la \nabla\trho^2,\nabla\trho^2\ra\right.\nonumber\\
&\left.+ 2(\Kper\cdot\Kper)(\nabla\trho^2,\nabla\trho^2)-4\Kper(\nabla(\Kper(\axial,\axial)),\nabla\trho^2)
+2\la \nabla(\Kper(\axial,\axial)),\nabla(\Kper(\axial,\axial))\ra\right\}\nonumber\\
&\left.-\frac{2}{\trho^2} \varLambda^{(1)}\Kper(\axial,\axial)\right..
\label{eq:varlambda2}
\end{align}

Summarizing we have shown the following result.

\begin{lemma}
  \label{state:axial_pert}
  Consider an axially symmetric perturbation scheme $(\mmm_\pertp,\gfam_\pertp,\{\psi_\pertp\})$
  for an axially symmetric background spacetime $(\mmm,g)$ with axial Killing vector $\axial$.
  Denote $\trho^2\defi\la\axial,\axial\ra$. Then,
  the first and second order perturbation tensors  $\Kper$ and $\Kperper$ satisfy
  (\ref{eq:axial_Ks}), the axis of symmetry of the perturbation coincides with the background
  axis of symmetry $\axis$ and the quantities $\varLambda^{(1)}$ and $\varLambda^{(2)}$,
  given by (\ref{eq:varlambda1}) and (\ref{eq:varlambda2}) respectively, vanish there.
  \fin 
\end{lemma}

Since in this paper we will be concerned with perturbations up to second order it makes sense to relax the definition of axially symmetric perturbation and impose conditions only up to this order. This leads to the following definition.

\begin{definition}
  \label{def:ASP_2}
  A perturbation scheme $(\mmm_\pertp,\gfam_\pertp,\{\psi_\pertp\})$
  of an axially symmetric background spacetime $(\mmm,g)$
  is a \textbf{second order axially symmetric perturbation} if
  it satisfies 
  the outcome of Lemma \ref{state:axial_pert}.
\end{definition}

The defining property of an axially symmetric perturbation scheme is that
  the axial Killing vector of the background is mapped to an axial
  Killing vector of $(\mmm_{\pertp},\gfam_{\pertp})$. Except in very special
  circumstances the spacetime $(\mmm_{\pertp},\gfam_{\pertp})$, $\pertp \neq 0$
  will admit only one axial Killing vector $\hat\axial_\pertp$, so
  $d\psi_\pertp(\axial)$ is forced to be this unique axial Killing field. In exceptional circumstances, where additional axial symmetries are present,
    the axially symmetric
    perturbation scheme becomes automatically larger. This additional freedom,
    however, is of a trivial nature and can be removed by any a priori identification of an axial Killing at every $\pertp \neq 0$. In the next lemma we identify restrictions on the perturbation vectors that arise from either the uniqueness
    of the axial symmetry of $(\mmm_{\pertp}, \gfam_{\pertp})$, $\pertp \neq 0$ or, in the exceptional cases, of having identified one axial Killing at each $\pertp$. This result will be useful at the end of the paper where
  uniqueness issues are discussed.
\begin{lemma}
\label{lemma:s_axial_commute}
In the setup of Lemma \ref{state:axial_pert}, if the class $\{\psi_\pertp\}$ is such that the axial Killing vector $\hat\axial_\pertp=d\psi_\pertp(\axial)$
  of $(\mmm_{\pertp},\gfam_{\pertp})$ is independent of $\psi_\pertp\in\{\psi_\pertp\}$, then the gauge vectors $\sper$ and $\sperper$ defined by a change of gauge within the class satify
  \begin{align*}
    [ \sper,\axial] =0, \qquad [\sperper,\axial] =0.
  \end{align*}
\end{lemma}
\begin{proof}
    Let $\psi_\pertp$
   and $\psi^g_\pertp$ belong to the class $\{ \psi_{\pertp} \}$. Recalling
   the definition $\Omega_\pertp=\psi_\pertp^{-1}\circ\psi_\pertp^g$, the assumption of the lemma implies (in fact, is equivalent to)  both
   $\bm \axial=\Omega^{-1}_\pertp{}^*(\bm \axial)$
   and $\bm \axial=\Omega_\pertp{}^*(\bm \axial)$,
   for all $\pertp \in I_0$.
   From the expresion of $V_{\pertp}$ in \eqref{def:sp_gauge_vectors} and the definition of Lie derivative we get $\lie_{V_{\pertp}}\bm\axial=0$. Evaluating at $\pertp =0$ yields $\lie_{V_{\pertp=0}} \bm{\axial} = \lie_{\sper} \bm{\axial} =0$.
   Taking the derivative at $\pertp=0$ and using
   (see e.g. the proof of Lemma 1 in \cite{Mars2005})
\[
\frac{d}{d\pertp}\lie_{V_\pertp}\bm\axial=\lie_{\frac{dV_\pertp}{d\pertp}}\bm\axial,
\]
the definition of $\sperper$ in \eqref{def:sp_gauge_vectors} gives
$\lie_{\sperper}\bm \axial=0$.\fin
\end{proof}

In order to explore the consequences of Definition \ref{def:ASP_2}  we need to
study the restrictions imposed by equation (\ref{eq:axial_Ks}). The
following lemma applies to arbitrary symmetric tensors invariant under
an axial Killing and may have independent interest. For related
results in the particular case when $K$ is the spacetime (background)
metric see \cite{Carot2000}.

\begin{lemma} 
\label{res:axial}
Let $(\mmm,g)$ be a $k$-dimensional ($k \geq 2$) 
axially symmetric spacetime with axial Killing 
$\axial$  and axis $\axis$. 
Let $K$ be a symmetric 2-covariant tensor satisfying $\lie_{\axial}K=0$
and $C^m$ ($m \geq 1$) on a neighbourhood
$\Uaxis\subset M$ of a portion of  $\axis$. By restricting
$\Uaxis$ 
if necessary 
we take $\Uaxis$ invariant under $\eta$ and admitting
global coordinates  $\{x,y,\tz\}$, where $u=3,\ldots,k$ if $k\geq 3$ or $\tz=\varnothing$ otherwise,
such that 
$\axial=x\partial_y-y\partial_x$. Let  $U \subset 
\mathbb{R}^k$ be the set where $\{x,y,\tz\}$ take values.
Then, using the notation $\normx \defi\sqrt{x^2+y^2}$ and
$K_{xx}\defi K(\partial_x,\partial_x)$, etc, it follows that, on $U \setminus \{ \normx =0\}$:
\begin{align}
&2K_{xx}=\C(\normx,\tz)+\A(\normx,\tz) \costwophi-\B(\normx,\tz) \sintwophi,\label{Kxx}\\
&2K_{yy}=\C(\normx,\tz)-\A(\normx,\tz) \costwophi+\B(\normx,\tz) \sintwophi,\label{Kyy}\\
&2K_{xy}=\A(\normx,\tz) \sintwophi+\B(\normx,\tz) \costwophi,\label{Kxy}\\
&K_{xu}=\frac{1}{\normx^2} \left ( \DD (\normx,\tz)x+\E(\normx,\tz)y 
\right ), \label{Kzxy1} \\
& K_{yu}= \frac{1}{\normx^2} \left ( -\E(\normx,\tz)x+\DD(\normx,\tz)y
\right ),\label{Kzxy2}\\
&K_{uv}=\QQ(\normx^2,\tz),\label{Kzz} 
\end{align}
where, for each $\paramp \in \{ 1,2,3,\rho u,\axial u,uv\}$,  the functions
$\traceA_{\paramp}(\rho,\tz)$ 
are defined on the domain $U_{\eta} := \{ \rho \in \mathbb{R}_{\geq0}, \tz \in \mathbb{R}^{k-2}
(\rho,0,\tz ) \in U\} \subset \mathbb{R}_{\geq0} \times \mathbb{R}^{k-2}$.
 Moreover,
\begin{align}
\traceA_{\paramp}(\rho,\tz) =\Polgen_{\paramp} (\rho^2, \tz) +\Phi^{(m)}_{\paramp}
(\rho,\tz),
\label{decomp}
\end{align}
where each  $\Polgen_{\paramp} (\rho^2,\tz)$ is a polynomial of degree
$[\frac{m}{2}]$ in $\rho^2$, $\Phi^{(m)}_\paramp$ is
$o(\rho^m)$ and $\Polgen_{\paramp}$, $\Phi^{(m)}_{\paramp}$ are
 $C^m$ on $U_{\eta}$. Furthermore $\Polgen_{2}, \Polgen_{3}, \Polgen_{\rho u}, \Polgen_{\axial u}$
vanish at $\rho=0$ (for all $\tz$).
\end{lemma}

\begin{remarklem}
\label{res:axial_remark}
  Several expressions above
  look undetermined at $\normx=0$. The factors $(x^2-y^2)/\normx^2$ and $xy/\normx^2$
  are bounded but have no limit as $\normx\to 0$.
  However, this is only apparent because $\A$, $\B$, $\DD$, $\E$ vanish as $\normx\to 0$.
    Therefore, 
    the expressions are valid in the whole  $U$.
\end{remarklem}

\begin{proof}
The equation $\lie_\axial K=0$ takes the following explicit form in components
\begin{align}
& \axial(K_{xx})=-2K_{xy}, \qquad \axial(K_{yy})=2K_{xy},\qquad \axial(K_{xy})=K_{xx}-K_{yy},
\label{eq:Kxy} \\
& \axial(K_{xu})+K_{yu}=0,\qquad
\axial(K_{yu})-K_{xu}=0, \label{xz}
\\
& \axial(K_{uv}) = 0. \label{zz}
\end{align}
$K$ being at least $C^1$ and $\eta$ vanishing on the axis these equations 
readily imply that $K_{xx}- K_{yy}$, $K_{xy}$, $K_{xu}$, $K_{yu}$ 
all vanish at $\normx=0$. Moreover, the following consequences are
easily obtained
\[
\axial(K_{xx}+K_{yy})=0, \quad \quad
\axial (x K_{x u} + y K_{y u} )=0, \quad \quad
\axial (y K_{x u}-x K_{y u}) =0, \quad \quad
\axial(K_{u v}) = 0,
\]
so $K_{xx}+K_{yy}, x K_{x u} + y K_{y u}$, 
$y K_{x u} -x K_{y u}$,  $K_{uv}$ are radially symmetric
in the variables $\{x,y\}$
and $C^m$ in $U$. By Lemma \ref{origin} in Appendix \ref{app:diff_origin}, there exist
$\C(\rho,\tz)$, $\DD(\rho,\tz)$, $\E(\rho,\tz)$ and $\QQ(\rho,\tz)$ defined on $U_{\eta}$
and satisfying (\ref{decomp}) such that
\begin{align*}
&K_{xx}+K_{yy}=\C(\normx,\tz), &K_{u v} = \QQ(\normx,\tz),\\
&y K_{x u} -x K_{y u} = \E(\normx,\tz), &x K_{x u} + y K_{y u} = \DD(\normx,\tz) .
\end{align*}
Now, $\DD(\rho,\tz)$ and $\E(\rho,\tz)$ vanish at $\rho=0$ (for all
$\tz$), so
the corresponding polynomials $\Polgen_{\rho u}$, $\Polgen_{\eta u}$ have the same property.
From the definitions it is clear that (\ref{Kzxy1}), 
(\ref{Kzxy2}) and 
(\ref{Kzz}) hold. 
Next, define
\[
\AT\defi \frac{x^2-y^2}{\normx^2}(K_{xx}-K_{yy})+\frac{2xy}{\normx^2}2K_{xy},\qquad
\BT \defi -\frac{2xy}{\normx^2}(K_{xx}-K_{yy})+\frac{x^2-y^2}{\normx^2}2K_{xy}
\]
in $U \setminus \{ \normx =0\}$.
The vanishing of $K_{xx} - K_{yy}$ and
$K_{xy}$ at the axis  show that these functions extend continuously
to zero at $\normx=0$.
Moreover, since
\[
\axial \left (\frac{x^2-y^2}{\normx^2} \right )=-2\frac{2xy}{\normx^2},\qquad
\axial \left (\frac{2xy}{\normx^2} \right )=2\frac{x^2-y^2}{\normx^2},
\]
we have
$\axial(\AT)=0$ and $\axial(\BT)=0$, and therefore $\AT,\BT$ are radially
symmetric in $\{x,y\}$. Define the corresponding traces 
$\A,\B:U_\axial\to \mathbb{R}$
by $\AT=\A(\normx,\tz)$ and $\BT=\B(\normx,\tz)$. Obviously
these functions  satisfy $\A(0,\tz)=\B(0,\tz)=0$.
Directly from the definitions 
$f\defi K_{xx}-K_{yy}$ 
takes the form $f=\A(\normx,\tz)\frac{x^2-y^2}{\normx^2}
-\B(\normx,\tz)\frac{2xy}{\normx^2}$ in $U \setminus \{ \normx =0\}$. 
So far we thus have (\ref{Kxx}) and (\ref{Kyy}),
and the first equation in (\ref{eq:Kxy}) leads to (\ref{Kxy}).
It only remains to show that $\A$ and $\B$ admit the decomposition
(\ref{decomp}).
Define the functions
\begin{align*}
&s_2(x,\tz)\defi f(x,y=0,\tz)=\A(\absx,\tz),\\
&s_3(x,\tz)\defi -\frac{1}{\sqrt{2}}f(x,y=x,\tz)=\B(\absx,\tz).
\end{align*}
First of all, the fact that $f$ is $C^m$  (in $U$)
implies  $s_2(x,\tz)$ and $s_3(x,\tz)$ are $C^m$ functions of their arguments,
particularly in the domain $x \geq 0$.
Setting $\rho=\absx$ we have that
$\A(\rho,\tz)=s_2(\rho,\tz)$  and $\B(\rho,\tz)=s_3(\rho,\tz)$,
and by Lemma  \ref{origin} particularised for $p=1$, $\A,\B$  are $C^m$
functions in $U_\axial$ with the structure given in (\ref{decomp}).\fin
\end{proof}

\section{``Orthogonally transitive'' perturbations}
\label{sec:OT_perturbations}
Recall that an Abelian $G_2$ (local) group of isometries on a spacetime $(M,g)$
is {\bf orthogonally transitive}   if, except possibly on a subset with
empty interior, the
orbits of the local isometry are two-dimensional, non-null and their orthogonal spaces form an
integrable distribution (equivalently, almost everywhere in $M$ there exists a foliation by
immersed non-null surfaces orthogonal to the orbits everywhere). 
By definition we say that this property is inherited by a perturbation scheme 
if the background admits an Abelian orthogonally transitive 
$G_2$ local action which is inherited on each $(\mmm_\pertp, \gfam_\pertp)$
and the corresponding (local) orbits
are orthogonally transitive.

\subsection{Orthogonally transitive actions}
\label{sec:OTactions}
The results of this subsection are essentially known.
We include them for completeness since they will be needed later.

Consider a $C^{n+1} (n \geq 0)$ pseudo-Riemannian manifold $(M,g)$ of dimension $k$
and arbitrary signature. Let $\{\xi_\kvi \}$ $\kvi,\kvj,\kvk, \cdots =1,\cdots, \mkill <   k$ 
be
a collection of (linearly independent) Killing vector fields that form an algebra, i.e.
such that
\begin{align*}
[\xi_\kvi, \xi_\kvj ] = C^{\kvk}_{\kvi\kvj} \xi_\kvk.
\end{align*}
For all $p \in M$ define
$\Pi|_p = \mbox{span} ( \xi_1 |_p, \cdots \xi_{\mkill} |_p )$. This is a vector
subspace of $T_p M$ of dimension at most $\mkill$. Define
$M' : = \{ p \in M; \mbox{dim}(\Pi|_p) = \mkill\}$. Linear independence of the Killing vectors implies
$M'$ is dense in $M$.
We make the following two assumptions:
\begin{itemize}
\item[(i)] For all $p \in M'$, the metric $g$ restricted to $\Pi|_p$ is non-degenerate.
\item[(ii)] For all $p \in M'$, the $g$-orthogonal complement $\Pi^{\perp}|_p$
defines an integrable distribution of $M'$.
\end{itemize}
Note that by (i) we have $T_p M = \Pi |_p \oplus \Pi^{\perp} |_p$
for all $p \in M'$.
Since the distribution $\{ \Pi^{\perp} |_p\}$ is
 defined as the set of vectors in the kernel of all\footnote{We use boldface to denote the metrically
  related one-form associated to a vector field.}
 $\bm{\xi}_\kvi :=
g (\xi_\kvi ,\cdot)$, the Fr\"obenius theorem states that (ii) can
be written equivalently in the following two ways
\begin{align*}
\mbox{(ii')} \quad \quad &  d \bm{\xi}_{\kvi_1} \wedge \bm{\xi}_{\kvi_2}
\wedge \cdots \wedge \bm{\xi}_{\kvi_{\mkill+1}} = 0 \quad \quad
\quad \forall \kvi_1, \cdots \kvi_{\mkill+1} = 1, \cdots \mkill, \\
\mbox{(ii'')} \quad \quad & d\bm{\xi}_{\kvi} =
\sum_{\kvj=1}^{\mkill} \bm{\xi}_\kvj \wedge \bm{\Sigma}_{\kvi\kvj} 
\end{align*}
where $\bm{\Sigma}_{\kvi\kvj}$ are smooth one-forms on $M'$.

We note the following facts. If $V,W$ are vector fields everywhere orthogonal
to $\{ \xi_\kvi \}$ (or, in other words, taking values in $\Pi^{\perp}|_p$, for all
$p \in M$) then $[\xi_\kvi, V]$ and $[V,W]$ have the same property. The proof
is by direct computation ($\alpha,\beta, \cdots$ are indices of $M$ and
$\nabla$ is the Levi-Civita derivative of $g$):
\begin{align*}
\xi_\kvi{}_{\alpha} [ \xi_\kvj, V]^{\alpha} & =
\xi_\kvi{}_{\alpha}
\left( \xi_\kvj^{\beta} \nabla_{\beta} V^{\alpha}
- V^{\beta} \nabla_{\beta} \xi_\kvj^{\alpha} \right)
= - V_{\alpha} \left( 
\xi_\kvj^{\beta} \nabla_{\beta} \xi_\kvi{}^{\alpha}
- \xi_\kvi^{\beta} \nabla_{\beta} \xi_\kvj^{\alpha} \right) \\
& = - V_{\alpha} [ \xi_\kvj,\xi_\kvi]^{\alpha} =0 
\end{align*}
and
\begin{align*}
\xi_\kvi{}_{\alpha} [ V, W] ^{\alpha} & =
\xi_\kvi{}_{\alpha}
\left( V^{\beta} \nabla_{\beta} W^{\alpha}
- W^{\beta} \nabla_{\beta} V^{\alpha} \right)
= V^{\beta} W^{\alpha}
 \left( - \nabla_{\beta} \xi_\kvi{}_{\alpha}
 + \nabla_{\alpha} \xi_\kvi{}_{\beta} \right) \\
&  = d \bm{\xi}_\kvi (W,V) = 
 \Big ( \sum_{\kvj=1}^{\mkill} \bm{\xi}_\kvj \wedge \bm{\Sigma}_{\kvi\kvj}
 \Big )(W,V) =0.
 \end{align*}
 Define $\gamma_{\kvi\kvj} := g(\xi_\kvi,\xi_\kvj) $. The next lemma introduces a set of closed one-forms that will then allow us to introduce suitable local coordinates. 
 \begin{lemma}
\label{OTlemma}
 In the setup above, assume further that the Lie
 algebra  $\{\xi_\kvi \}$ is Abelian. Then there exists smooth closed one-forms
 $\bm{\zeta}^\kvi$ on $M'$ such that
 \begin{equation}
 \bm{\xi}_\kvi = \gamma_{\kvi\kvj} \bm{\zeta}^\kvj.
\label{c1}
\end{equation}
 \end{lemma}

\begin{proof}
By condition  (i), $\gamma_{\kvi\kvj}$ is invertible
 on $M'$. Let $\gamma^{\kvi\kvj}$ be the inverse (i.e.
 such that $\gamma^{\kvi\kvk} \gamma_{\kvk\kvj} = \delta^\kvi_\kvj$) and define
 $\bm{\zeta}^\kvi := \gamma^{\kvi\kvj} \bm{\xi}_\kvj$. These are $C^{n+1}$ one-forms
 on $M'$ obviously satisfying (\ref{c1}).
It remains to show that
 $d \bm{\zeta}^\kvi =0$.
Observe first that
 $\bm{\zeta}^\kvi (\xi_\kvj ) = \delta^\kvi_\kvj$.
We use the definition of exterior differential
 \[
 d \bm{\zeta}^\kvi (X,Y)  = X (\bm{\zeta}^\kvi (Y))
  - Y (\bm{\zeta}^\kvi (X)) - \bm{\zeta}^\kvi ([X,Y]) 
\]
 where $X,Y$ are arbitrary vector fields in $M'$. To prove
 $d \bm{\zeta}^\kvi=0$
is suffices to show that the right-hand side of this expression vanishes
in the following three cases, (a) $X= \xi_\kvi, Y = \xi_\kvj$,
(b) $X = \xi_\kvi, Y = W$ and (c) $X=V, Y= W$,
where $V,W$ are everywhere orthogonal to  $\{ \xi_\kvj \}$. Now
\begin{align*}
(a) \quad \quad &  d \bm{\zeta}^\kvi (\xi_\kvj,\xi_\kvk)  = \xi_\kvj ( \delta^\kvi_\kvk )
 - \xi_\kvk ( \delta^\kvi_\kvj)  - \bm{\zeta}^\kvi ([\xi_\kvj,\xi_\kvk])  =0 \\
(b) \quad \quad &  d \bm{\zeta}^\kvi (\xi_\kvj, V)  = 
 - V ( \delta^\kvi_\kvj)  - \bm{\zeta}^\kvi ([\xi_\kvj,V])  =0 \\
(c) \quad \quad &  d \bm{\zeta}^\kvi (V, W)  = 
- \bm{\zeta}^\kvi ([V,W])  =0,
\end{align*}
 where in (a) we used the assumption that algebra is Abelian and
 in (b), (c) the fact that  $[\xi_\kvj,V]$ and $[V,W]$, along with $V$, $W$
are all orthogonal to $\xi_{\kvj}$.\fin
\end{proof}

\begin{corolemm}
\label{OTcoro}
If $M'$ is simply connected, then there exist smooth functions
$z^i$ on $M'$ such that
\begin{equation}
\xi_\kvi (z^\kvj ) = \delta^\kvj_\kvi, \quad  \quad
\bm{\xi}_\kvi = \gamma_{\kvi\kvj} d z^\kvj . 
\label{prop}
\end{equation}
\end{corolemm}

\begin{proof}
By simply connectedness, closed is equivalent to exact. Define
$z^\kvj$ by $\bm{\zeta}^\kvj = d z^\kvj$ and (\ref{prop}) are immediate.\fin
\end{proof}

It is also a well-known  fact that under assumptions (i) and (ii) there exist local coordinates near any point in $M'$ in which the metric  separates in blocks. The argument is standard, but we include it for completeness.
  Fix $p \in M'$ and let $\Sigma_p$ be integrable manifold of the distrution
  $\Pi^{\perp}$ containing $p$. Select any local coordinate system $\{ x^\otm \}$
  on $\Sigma_p$ near $p$ and extend the functions $x^{\otm}$ as constants along the integral manifolds of the distribution $\Pi$. By construction $\xi^{\kvi}(x^{\otm})=0$.
  It is immediate that $\{ x^\otm, z^\kvi\}$ defines  a local coordinate system
  in a neighbourhood of $p$. By the first of \eqref{prop} we have  $\xi^\kvi = \partial_{z^\kvi}$, so the metric in these coordinates only depends on $\{ x^\otm \}$. By the second in \eqref{prop} the cross components $g_{\kvi\otm}$ of the metric vanish and $g_{\kvi\kvj} = \gamma_{\kvi\kvj}$.  The block diagonal structure follows  
\[
g = \gamma_{\kvi\kvj} (x^\otp) dx^\kvi dx^\kvj + h_{\otm\otn}(x^\otp) dx^\otm dx^\otn.
\]

\subsection{``Orthogonally transitive'' perturbation scheme}

Let $(\mmm_\pertp, \gfam_\pertp,\{ \psi_\pertp\})$ be a  perturbation scheme that
inherits an orthogonally transitive 
Abelian $G_{\mkill}$ (local) group of isometries generated by
$\mbox{span} \{ \xi_\kvi \}$, $\kvi,\kvj =1, \cdots, \mkill$, ($1 \leq \mkill < \mbox{dim} (\mmm)$)
 on $(\mmm,g)$. The next proposition shows that 
the first and second order
metric perturbation tensors can be taken as block diagonal in $\mmm'$.
Later we show that in the stationary and axisymmetric case
the transformation does not lower the differentiability
of the perturbation tensors in the whole $\mmm$, that is, including the axis.

\begin{proposition}
\label{Block_pre}
Let $(\mmm_\pertp, \gfam_\pertp,\{ \psi_\pertp\})$ be a  perturbation scheme of class $C^{n+1}$ $(n \geq 2)$
that inherits an Abelian $G_{\mkill}$, $1 \leq \mkill < \mbox{dim} (\mmm)$
(local) isometry group that acts orthogonally transitively. 
Let $\mmm'$ be the open and dense subset  where the orbits of the 
Abelian (local) isometry group are non-null and have dimension $\mkill$ and assume that
$\mmm'$ is simply connected. Then, for 
each $\ge = \psi_{\pertp} (\gfam_{\pertp})$, $\psi_{\pertp} \in \{ \psi_{\pertp} \}$
there exist first and second order gauge vectors $\sper$ and $\sperper$ such that
the corresponding transformed
$\fpt^{g}$ and $\spt^g$ are of class $C^n(\mmm')$ and $C^{n-1}(\mmm')$ respectively and
block diagonal, i.e. take form
\begin{align*}
\fpt dx^\alpha dx^{\beta} & = {\fpt}_{\kvi\kvj}(x^\otp) dx^\kvi dx^\kvj + {\fpt}_{\otm\otn}(x^\otp)
dx^\otm dx^\otn, \quad \quad \\
\spt dx^\alpha dx^{\beta} & = {\spt}_{\kvi\kvj}(x^\otp) dx^\kvi dx^\kvj + {\spt}_{\otm\otn}(x^\otp)
dx^\otm dx^\otn
\end{align*}
in any 
local coordinate system $\{ x^{\alpha} \} = \{ x^\kvi, x^\otm\}$ in $\mmm'$ adapted to $\xi_\kvi$
and to the orthogonal transitivity of the background, i.e. where
\begin{equation}
g = \gamma_{\kvi\kvj} (x^\otp) dx^\kvi dx^\kvj + h_{\otm\otn}(x^\otp) dx^\otm dx^\otn,
\quad \quad \quad
\xi_\kvi = \partial_{x^\kvi}.
\label{form_pre}
\end{equation}
\end{proposition}

\begin{proof}
By  Lemma \ref{OTlemma} and Corollary \ref{OTcoro},
there exists smooth  scalar functions $z^\kvi_{\pertp} : \mmm' \longrightarrow
\mathbb{R}$ satisfying $\xi_\kvi (z_{\pertp}^\kvj )  = \delta^\kvj_\kvi$
such that the 
one-form fields
$\bm{\xi}^{\pertp}_\kvi := \ge(\xi_\kvi,\cdot)$ 
take the form
\begin{equation*}
\bm{\xi}^{\pertp}_\kvi = \gamma_{\kvi\kvj}^{\pertp} d z^\kvj_\pertp, \quad \quad
\quad \gamma_{\kvi\kvj}^{\pertp} := \ge(\xi_\kvi, \xi_\kvj).
\end{equation*}
Let $\{x^{\alpha}\}$ be any local coordinate where the background metric
takes the form  (\ref{form_pre}). The condition $\xi_\kvi (z^\kvj_{\pertp}) = 
\delta^\kvj_\kvi$ implies the existence of scalar functions $u^\kvj_{\pertp} (x^\otm)$ 
such that $z^\kvj_\pertp = x^\kvj + u^\kvj_{\pertp}(x^\otm)$. The form of $\bm{\xi}^{\pertp}_\kvi$ 
forces that the metric $\ge$ in these coordinates takes the form
\begin{equation*}
\ge = 
 \gamma^{\pertp}_{\kvi\kvj} ( dx^\kvi + \partial_\otm u_{\pertp}^\kvi dx^\otm) (dx^\kvj + 
\partial_\otn u_{\pertp}^\kvj dx^\otn )  + h^{\pertp}_{\otm\otn} dx^\otm dx^\otn.
\end{equation*}
Since $\xi_\kvi=\partial_{x^\kvi}$ are Killing vectors of this metric, 
$\gamma^{\pertp}_{\kvi\kvj} (x^\otm)$ and $h^{\pertp}_{\otm\otn}(x^\otp)$. Moreover, the equality $g =
g_{\pertp=0}$ implies
\begin{equation*}
\gamma^{\pertp}_{\kvi\kvj} |_{\pertp=0} = \gamma_{\kvi\kvj},
\quad \quad
h^{\pertp}_{\kvi\kvj} |_{\pertp=0} = h_{\kvi\kvj},
\quad \quad
\left. u^\kvi_\pertp \right |_{\pertp =0} = 0.
\end{equation*}
Computing the first and second order perturbation tensors, one finds after a
simple calculation
\begin{align*}
\fpt  
= &{\fpt}_{\kvi\kvj} dx^\kvi dx^\kvj + 2 \gamma_{\kvi\kvj} \partial_\otm U^\kvj dx^\kvi dx^\otm
+ {\fpt}_{\otm\otn} dx^\otm dx^\otn, \\
\spt  = &
{\spt}_{\kvi\kvj} dx^\kvi dx^\kvj + 4 {\fpt}_{\kvi\kvj} \partial_\otm U^\kvj dx^\kvi dx^\otm
+ 2 \gamma_{\kvi\kvj} \partial_\otm U^\kvi \partial_\otn U^\kvj dx^\otm dx^\otn
+ 2 \gamma_{\kvi\kvj} \partial_\otm W^\kvi dx^\kvi dx^\otm  \\
& + {\spt}_{\otm\otn} dx^\otm dx^\otn,
\end{align*}
where 
\begin{align*}
{\fpt}_{\kvi\kvj} & := 
\left. \frac{d \gamma^{\pertp}_{\kvi\kvj}}{d \pertp} 
\right|_{\pertp=0}, 
\quad \quad
{\fpt}_{\otm\otn} := 
\left. \frac{d h^{\pertp}_{\otm\otn}}{d \pertp} 
 \right|_{\pertp=0}, \quad \quad
U^\kvi := \left. \frac{ d u^\kvi_{\pertp}}{d \pertp}  \right|_{\pertp=0}, \\
{\spt}_{\kvi\kvj} & := 
\left. \frac{d^2 \gamma^{\pertp}_{\kvi\kvj}}{d \pertp^2} 
\right |_{\pertp=0}, 
\quad \quad
{\spt}_{\otm\otn} := 
\left. \frac{d^2 h^{\pertp}_{\otm\otn}}{d \pertp^2} 
\right|_{\pertp=0}, \quad \quad
W^\kvi := \left. \frac{ d^2 u^\kvi_{\pertp}}{d \pertp^2}  \right|_{\pertp=0}.
\end{align*}
We consider the gauge vectors $\sper = -  U^\kvi(x^\otm) \partial_\kvi$ and $ \sperper=
- W^\kvi(x^\otm) \partial_\kvi$. It is immediate that
\begin{equation*}
\lie_{\sper} g = - 2 \gamma_{\kvi\kvj} \partial_\otm U^\kvj dx^\kvi dx^\otm
\end{equation*}
and hence $\fpt^g = \fpt + \lie_{\sper} g  
= {\fpt}_{\kvi\kvj} dx^\kvi dx^\kvj + {\fpt}_{\otm\otn} dx^\otm dx^\otn$ is block diagonal, as claimed.
For the second order perturbation we use the form (\ref{gaugeperper_bis}).
A simple calculation gives
\begin{equation*}
\lie_{\sper} \fpt^g = - 2 {\fpt^{g}}_{\kvi\kvj} \partial_\otm U^\kvj dx^\kvi dx^\otm, \quad
\quad \quad
\lie_{\sper} \lie_{\sper} g =
2 \gamma_{\kvi\kvj} \partial_\otm U^\kvi \partial_\otn U^\kvj dx^\otm dx^\otn
\end{equation*}
and we conclude that $\spt = 
{\spt}_{\kvi\kvj} dx^\kvi dx^\kvj 
+ {\spt}_{\otm\otn} dx^\otm dx^\otn$, i.e. block diagonal. We have performed the computation
to change the gauge in local coordinates, but it is
clear that both $\sper$ and $\sperper$ are
globally defined
on $\mmm'$ and smooth because they are intrinsically defined  by
\begin{equation}
\label{eq:build_gauge_vectors_OT}
\sper
=  - \left. \frac{d z^\kvi_{\pertp}}{d \pertp} \right|_{\pertp=0} \xi_\kvi
\, =- \left. \frac{d u^\kvi_{\pertp}}{d \pertp} \right|_{\pertp=0} \xi_\kvi,
\quad \quad 
\sperper
= - \left. \frac{d^2 z^\kvi_{\pertp}}{d \pertp^2} \right|_{\pertp=0} \xi_\kvi
\, =- \left. \frac{d^2 u^\kvi_{\pertp}}{d \pertp^2} \right|_{\pertp=0} \xi_\kvi.
\end{equation}\fin
\end{proof}

\begin{remarkpro}
  This proposition includes as a particular case the situation when
  the perturbation scheme
 inherits a static (local) isometry, i.e. an orthogonally transitive
 one-parameter (local) isometry group with timelike orbits.
\end{remarkpro}

At the level of generality of Proposition \ref{Block_pre} we can only establish the differentiablity of  $\sper$ and $\sperper$
    on $\mmm'$, i.e. away from points where the distribution
    $\{\Pi_p\}$ degenerates. The reason lies in the non-invertibility
    of  $\gamma_{\kvi\kvj}$ at the  complement
    of $\mmm'$. Extending the differentiability of the gauge vectors to the whole of $\mmm$ is a delicate issue which, however, must be addressed in our problem. Indeed, an important step of our argument will be applying Lemma \ref{res:axial} on pertubation tensors written in block diagonal form,
    for which it is crucial that these tensors
    are differentiable everywhere, including the axis. To accomplish this we will exploit the fact that the metrics $\gfam_{\pertp}$ are stationary and axisymmetric. Using the results in Lemma \ref{res:axial} we will be able to analyse the behaviour of the one-forms $\bm\zeta_{\pertp}^\kvi$ defined in Lemma  \ref{OTlemma} near the  axis. This will be sufficient, via integration of $
    d z^{\kvi}_{\pertp}= \bm\zeta_{\pertp}^{\kvi}$, to show via
    \eqref{eq:build_gauge_vectors_OT}, that the gauge vectors
    $\sper$ and $\sperper$ are differentiable everywhere.
    
We need several well-known facts (see e.g. \cite{Exact_solutions})
  about stationary and axisymmetric  group actions. Recall that this is a spacetime isometry
generated
by an axial Killing vector $\axial$, see Section \ref{sec:axial},
assumed to be spacelike everywhere outside the axis,
and a timelike Killing vector $\stat$. The group is necessarily Abelian and
  the points where the group orbits are two-dimensional and non-null
  (i.e. $\mmm'$) are determined by $\det\gamma=\la\stat,\stat\ra \la\axial,\axial\ra-\la\stat,\axial\ra^2< 0$. Given that $\trho^2=\la\axial,\axial\ra$ is non-negative and zero only at points of the axis $\axis$, $\det\gamma=0$ also defines $\axis$.
The construction of a stationary and axially symmetric perturbation scheme
implies all the above for each $\pertp$, in particular for  $\det\gamma^\pertp$. 
Lemma \ref{state:axial_pert} ensures that the axis of the perturbation,
i.e. the axis at each $\pertp$, coincides with $\axis$,
and therefore we have that $\det\gamma^\pertp$ vanishes on $\axis$ and
only there.

\begin{proposition}
\label{Block}
Let $(\mmm_\pertp, \gfam_\pertp,\{ \psi_\pertp\})$ be a  perturbation scheme of class $C^{n+1}$ $(n \geq 2)$
that inherits a stationary and axisymmetric
 isometry group that acts orthogonally transitively. 
Assume that $\mmm'=\mmm\setminus \axis$ is simply connected.
Then, for 
each $\ge = \psi_{\pertp} (\gfam_{\pertp})$, $\psi_{\pertp} \in \{ \psi_{\pertp} \}$
there exist first and second order gauge vectors $\sper$ and $\sperper$ such that
the corresponding transformed
$\fpt^{g}$ and $\spt^g$ are of class $C^n(\mmm)$ and $C^{n-1}(\mmm)$ respectively and
block diagonal, i.e. take form
\begin{align*}
\fpt dx^\alpha dx^{\beta} & = {\fpt}_{\kvi\kvj}(x^\otp) dx^\kvi dx^\kvj + {\fpt}_{\otm\otn}(x^\otp)
dx^\otm dx^\otn, \quad \quad \\
\spt dx^\alpha dx^{\beta} & = {\spt}_{\kvi\kvj}(x^\otp) dx^\kvi dx^\kvj + {\spt}_{\otm\otn}(x^\otp)
dx^\otm dx^\otn
\end{align*}
in any 
local coordinate system $\{ x^{\alpha} \} = \{ x^\kvi, x^\otm\}$ in $\mmm'$ adapted to $\xi_\kvi$
and to the orthogonal transitivity of the background, i.e. where
\begin{equation}
g = \gamma_{\kvi\kvj} (x^\otp) dx^\kvi dx^\kvj + h_{\otm\otn}(x^\otp) dx^\otm dx^\otn,
\quad \quad \quad
\xi_\kvi = \partial_{x^\kvi}.
\label{form}
\end{equation}
\end{proposition}

\begin{proof} We are in the setting of
Proposition \ref{Block_pre}
with  
  $\mkill=2$, $\xi_1=\stat$ 
  and $\xi_2=\axial$. We define $\bm\xi^{\pertp}_{\kvi}$ as in this proposition
  and  $\bm{\zeta}^{\kvi}_{\pertp}
  := \gamma^{\pertp}{}^{\kvi\kvj}  \bm\xi^{\pertp}_{\kvj}$
  on $\mmm \setminus \axis$ (cf. Lemma \ref{OTlemma}). The core of the proof is understanding in detail the
  behaviour of the closed one-forms $\bm{\zeta}^{\kvi}_{\pertp}$
  as we approach $\axis$, for which we will use the results in Lemma \ref{res:axial}.  Introducing coordinates $\{x,y,z,t\}$ in a  neighbourhood $\Uaxis\subset\mmm$
  of the axis such that $\axial=x\partial_y-y\partial_x$
  and $\xi= \partial_t$, we may apply Lemma \ref{res:axial}
  with
  $K=g_{\pertp}$, $m=n+1$ and $w^u = \{ z,t\}$ ($u=3,4$). As in this lemma, 
  we call $U\subset \mathbb{R}^4$
  the set where $\{x,y,z,t\}$ take values.

We shall use expressions \eqref{Kxx}-\eqref{Kzz} with
   $K\to g_\pertp$ and $\traceA\to \traceA^\pertp$.
Since $g_\pertp$ is stationary,  $\traceA^\pertp_\paramp$
do not depend on $t$, i.e. they are functions
$\traceA^\pertp_\paramp(\normx,z)$, where  $\normx^2\defi x^2+y^2$.
Their traces satisfy \eqref{decomp} and are defined on the domain
$D_{\eta} := \{ \rho \in \mathbb{R}_{\geq 0}, z \in \mathbb{R} ;
(\rho,0,z,t ) \in U\} \subset \mathbb{R}_{\geq 0} \times \mathbb{R}$.
Lemma \ref{origin} in Appendix \ref{app:diff_origin} ensures that 
$\traceA^\pertp_\paramp(\normx,z)$ are $C^{n+1}(\Uaxis)$. From on on, we
write $\traceA^\pertp_\paramp$ when we refer to
$\traceA^\pertp_\paramp(\normx,z)$, i.e. as functions on
$\Uaxis$.

\vspace{2mm}
The analysis of $\bm\zeta^\kvi_\pertp$ relies on the following:

\vspace{2mm}

 {\bf Claim:} For $\parampp \in \{2,3,\rho t, \rho z, \axial t,\axial z\}$, the functions
  $\widetilde\traceA^\pertp_\parampp(\rho,z)\defi \frac{1}{\rho}\traceA^\pertp_\parampp(\rho,z)$
  and
  $\frac{1}{\rho} \widetilde\traceA^\pertp_{\parampp}\widetilde\traceA^\pertp_{\parampp'}(\rho,z)$ admit the  expansion (we drop a label $\pertp$ in the right-hand
  side for simplicity)
\begin{align}
  \widetilde\traceA^\pertp_\parampp(\rho,z)
& =\rho \sum_{k=0}^{[\frac{n-1}{2}]} \rho^{2k} P_{\parampp \, k} (z)
 +\tilde\Phi^{(n)}_{\parampp}(\rho,z), \label{eq:tilde_gothic_A} \\
  \frac{1}{\rho} \widetilde\traceA^\pertp_{\parampp}\widetilde\traceA^\pertp_{\parampp'} (\rho,z)
& = \rho \sum^{[\frac{n-1}{2}]}_{k=0} \rho^{2k} P_{\parampp\parampp' k}(z) + \tilde\Phi^{(n)}_{\parampp\parampp'}(\rho,z), \label{eq:product_tildes}
  \end{align}
  where $P_{\parampp \, k} (z)$, $P_{\parampp \, \parampp' \, k} (z)$ are $C^{n+1}$ functions of $z$ and  $\tilde\Phi^{(n)}_{\parampp}(\rho,z), \tilde\Phi^{(n)}_{\parampp\, \parampp'}(\rho,z)$ are $C^n$ and $o(\rho^n)$ with respect to $\rho$ and $C^{n+1}$ with respect to $z$. 

\vspace{2mm}

{\it Proof of the claim:}  Lemma \ref{res:axial} with $m=n+1$ establishes that 
  $\traceA^\pertp_{\parampp}(\rho,z)$ vanishes at the axis and admits an expansion
\[
\traceA^\pertp_{\parampp} (\rho,z)= \rho^2
\sum_{k=0}^{[\frac{n-1}{2}]} \rho^{2k} P_{\parampp \, k} (z)
+ \Phi^{(n+1)}_{\parampp}(\rho,z),
\]
where $P_{\parampp \, k} (z)$ are $C^{n+1}$ functions of $z$ and
$\Phi^{(n+1)}_{\parampp}$ are $C^{n+1}(D_\axial)$ and $o(\rho^{n+1})$. Setting
$\tilde\Phi^{(n)}_{\parampp}(\rho,z)
\defi\frac{1}{\rho} \Phi^{(n+1)}_{\parampp}(\rho,z)$, the expansion \eqref{eq:tilde_gothic_A} follows. The function $\tilde\Phi^{(n)}_{\parampp}(\rho,z)$ 
is $o(\rho^n)$ and $C^n$
with respect to $\rho$
by item \emph{(iii)} of
Lemma \ref{res:lemma_nkf} applied to the one-dimensional case.
Moreover, this function is $C^{n+1}$ in $z$,
since, being $o(\rho^n)$, all partial derivatives with respect to $z$
extend to $\rho=0$, where they vanish.
For the product function, it is immediate from the definition that 
\begin{align*}
& \frac{1}{\rho} \widetilde\traceA^\pertp_{\parampp}\widetilde\traceA^\pertp_{\parampp'}(\rho,z)=
  \rho \sum^{[\frac{n-1}{2}]}_{k=0} \rho^{2k} P_{\parampp\parampp' k}(z) \\
& + \rho \sum^{2[\frac{n-1}{2}]}_{k= [\frac{n+1}{2}]} \rho^{2k} P_{\parampp\parampp' k}(z)     
+\sum_{k=0}^{[\frac{n-1}{2}]} \rho^{2k}
\left(P_{\parampp \, k} (z) \tilde\Phi_{\parampp'}^{(n)}+P_{\parampp' \, k} (z)\tilde\Phi_\parampp^{(n)}\right)  + \frac{1}{\rho} \tilde\Phi_\parampp^{(n)}(\rho,z) \tilde\Phi_{\parampp'}^{(n)}(\rho,z). \nonumber
\end{align*}
The second line defines $\tilde\Phi^{(n)}_{\parampp\,\parampp'}(\rho,z)$ and the expansion \eqref{eq:product_tildes} follows. The property that
$\tilde\Phi^{(n)}_{\parampp\,\parampp'}(\rho,z)$ is $C^{n+1}$ in $z$ is clear, as it holds for each term. Concerning the property
of being $C^n$ in $\rho$ and $o(\rho^n)$, this is immediate for the first two terms. For the last term, it follows from its product structure, as shown in Corollary \ref{corollary:ufff} of Appendix \ref{app:iota}.

\vspace{2mm}

Combining this claim with Lemma \ref{origin}, we also conclude that
the functions
$\widetilde\traceA^\pertp_{\parampp}\widetilde\traceA^\pertp_{\parampp'}(\normx,z)$
are $C^{n}(\Uaxis)$.
We are now ready to compute $\bm\zeta_\pertp^\kvi$. 

Recall that $\gamma^\pertp_{\kvi\kvj}=g_\pertp(\xi_\kvi,\xi_\kvj)$. Using the explicit expressions for $g_{\pertp}$ given in Lemma \ref{res:axial}, a straightforward calculation gives
\[
  \det\gamma^\pertp= \frac{1}{2} \normx^2 \left ( 
  (\C^\pertp-\A^\pertp)\traceA^\pertp_{tt}-2\widetilde\traceA^\pertp_{\axial t}\widetilde\traceA^\pertp_{\axial t}\right ) := \normx^2 D_{\gamma^\pertp}.
\]
From the above, $D_{\gamma^\pertp}$ 
is $C^{n}(\Uaxis)$ and its restriction
to any value of $\normx$ is $C^{n+1}$ in $z$.
Moreover, the value of $D_{\gamma^\pertp}$ on the axis
is $D_{\gamma^\pertp}|_{\normx=0}(z)=\C^\pertp(0,z)\traceA_{tt}^\pertp(0,z)$.
This is not zero
because
$\traceA^\pertp_{tt}=g_\pertp(\stat,\stat)< 0$, c.f. \eqref{Kzz},
while \eqref{Kxx} restricted to $x=y=0$
provides $2g_\pertp(\partial_x,\partial_x)|_{\normx=0}=\C(0,z)$,
so that $\C(0,z)>0$. Since $\det\gamma^\pertp$ vanishes only at the axis, it follows that
$D_{\gamma^\pertp}$ is nowhere zero on $\Uaxis$.
Therefore the inverse $D^{-1}_{\gamma^\pertp}\defi 1/D_{\gamma^\pertp}$
is also $C^{n}(\Uaxis)$, and
by Lemma \ref{origin}, the corresponding trace function $D^{-1}_{\gamma^\pertp}(\rho,z)$
must have the form
\begin{equation}
\label{eq:inverse_determinant}
D^{-1}_{\gamma^\pertp}(\rho,z)=\sum_{k=0}^{[\frac{n}{2}]}\rho^{2k}P_{\gamma^\pertp k}(z)
+\tilde\Phi^{(n)}_{\gamma^\pertp}(\rho,z),
\end{equation}
where
$P_{\gamma^\pertp 0}(z)$ is nowhere zero,
$P_{\gamma^\pertp k}(z)$ are $C^{n+1}$ in $z$, 
and $\tilde\Phi^{(n)}_{\gamma^\pertp}(\rho,z)$ are $C^{n+1}$ in $z$,
$C^n$ in $\rho$ and $o(\rho^n)$. From
  the  definition of $\bm\zeta^i_{\pertp}$ (see Lemma \ref{OTlemma}), a
  direct computation gives
\begin{align*}
&\bm\zeta^1_\pertp=dt+ \frac{1}{\normx}\trace{Z}^1_\pertp{}_\rho (xdx+ydy)+ \trace{Z}^1_\pertp{}_z  dz,\\
&\bm\zeta^2_\pertp=d\arctan\frac{y}{x}
+\frac{1}{\normx}\trace{Z}^2_\pertp{}_\rho (xdx+ydy)+\trace{Z}^2_\pertp{}_z  dz,
\end{align*}
where
\begin{align*}
&\trace{Z}^1_\pertp{}_\rho(\normx,z)\defi  \frac{1}{2D_{\gamma^\pertp}}
\left((\C^\pertp-\A^\pertp)\widetilde\traceA^\pertp_{\rho t}+\B^\pertp \widetilde\traceA^\pertp_{\axial t}\right),
&& \trace{Z}^1_\pertp{}_z(\normx,z)\defi \frac{1}{2D_{\gamma^\pertp}}
\left((\C^\pertp-\A^\pertp)\traceA^\pertp_{tz}-2\widetilde\traceA^\pertp_{\axial t}\widetilde\traceA^\pertp_{\axial z}\right), \\
&\trace{Z}^2_\pertp{}_\rho(\normx,z)\defi \frac{1}{2D_{\gamma^\pertp}}\left(\widetilde\B^\pertp\traceA^\pertp_{tt}
+2\frac{1}{\normx}\widetilde\traceA^\pertp_{\axial t}\widetilde\traceA^\pertp_{\rho t}\right),
&& \trace{Z}^2_\pertp{}_z(\normx,z)\defi  \frac{1}{2D_{\gamma^\pertp}}\frac{1}{\normx}
\left(\traceA^\pertp_{tz}\widetilde\traceA^\pertp_{\axial t}-\traceA^\pertp_{tt}\widetilde\traceA^\pertp_{\axial z}\right).
\end{align*}
We can now integrate
$dz^\kvi_\pertp=\bm\zeta^\kvi_\pertp$.
Using $\normx d \normx = x dx + y dy$, 
  the functions  $z^\kvi_\pertp$ take the form (the integrability
  conditions are ensured by Corollary \ref{OTcoro})
  \begin{align*}
&z_\pertp^1=t+\int_0^{\normx} \trace{Z}_\pertp^1{}_\rho(s,z) ds
+ \int \trace{Z}_\pertp^1{}_z(0,z) dz := x^1 + u^1_{\pertp}, \\
&z_\pertp^2=\arctan\frac{y}{x}+\int_0^{\normx} \trace{Z}_\pertp^2{}_\rho(s,z) ds
+ \int \trace{Z}_\pertp^2{}_z(0,z) dz := x^2 + u^2_{\pertp},
  \end{align*}
  where we have set $x^1=t$ and $x^2 =\arctan(y/x) = \varphi$.
    By \eqref{eq:build_gauge_vectors_OT}, the gauge vectors
    are constructed from  $u^\kvi_\pertp$, so we only need to care about the
    integral terms and show that they define
    $C^{n+1}(\Uaxis)$ functions.

    From Lemma \ref{res:axial}    and the claim above it follows that
    $\trace{Z}_\pertp^\kvi{}_z(0,z)$ are $C^{n+1}$ functions of $z$, so the
    two integrals over $z$  are $C^{n+2}$ functions of $z$. Concerning  the
    $\rho$-integrals, \eqref{decomp}, \eqref{eq:inverse_determinant}
    and the claim imply that
        \begin{equation*}
\trace{Z}_\pertp^\kvi{}_\rho(\rho,z)=\rho \sum^{[\frac{n-1}{2}]}_{k=0}\rho^{2k}P_\pertp^\kvi{}_{k}(z)+\tilde\Phi^\kvi_\pertp{}^{(n)}(\rho,z),
\end{equation*}
where $P_\pertp^\kvi{}_{k}(z)$ and $\tilde\Phi^\kvi_\pertp{}^{(n)}$
are $C^{n+1}$ with respect to $z$
and $\tilde\Phi^\kvi_\pertp{}^{(n)}$
is $C^{n}$ in $\rho$  and $o(\rho^n)$.
As a result, 
\[
  Z^\kvi_\pertp(\rho,z)\defi
  \int_0^{\rho} \trace{Z}_\pertp^\kvi{}_\rho(s,z) ds=
\sum^{[\frac{n-1}{2}]}_{k=0}P_\pertp^\kvi{}_{k}(z)\frac{1}{2k+1}\rho^{2(k+1)}+
\int_0^{\rho} \tilde\Phi^\kvi_\pertp{}^{(n)}(s,z)ds
\]
is $C^{n+1}$ in the domain $D_\axial$.
The first term is a polynomial in $\rho^2$, while the last integral
  is   $o(\rho^{n+1})$
  (e.g. by the mean value theorem).
Lemma \ref{origin} ensures that 
$Z^\kvi_\pertp(\normx,z)$ is $C^{n+1}(\Uaxis)$. Applying the
  definifion \eqref{eq:build_gauge_vectors_OT} we conclude that the
  gauge vectors $\sper$ and $\sperper$ are $C^{n+1}(\Uaxis)$.  \fin
\end{proof}

Given a pertubation scheme that inherits a stationary and
  axisymmetric orthogonally transitive
group action, we shall always consider the related 
perturbation scheme whose existence is proved in this proposition. The
corresponding  metric perturbation tensors will be denoted 
$\fpt$ and $\spt$ (without the $g$ superindex).

\begin{remarkpro}
\label{Remark:diag}
When the codimension of the Abelian (local) group action is two, the metric
$h^{\pertp}(x^\otp)_{\otm\otn} dx^\otm dx^\otn$ is two-dimensional. Given any point 
$p \in \mmm'$ and  a suficiently small neighbourhood $U_p$ of $p$ 
there exist  coordinates  $\{ x^{\kvi}, x^\otm\}$ adapted to the Killings
where the background metric is not only block diagonal, but also
with $h_{\otm\otn}$ diagonal.  By restricting $U_p$ if necessary, it
is also true that there exists a change of coordinates 
$\{ x^\kvi, x^\otm_{\pertp}(x^\otn)\} $ where $h^{\pertp}_{\otm\otn}$ is also 
 diagonal for all $\pertp$. This map can be considered as a local diffeomorphism
and hence as a local gauge transformation. In the transformed gauge one has, by construction, that
${\fpt^g}_{\otm\otn}$ and ${\spt^g}_{\otm\otn}$ are both diagonal.
It should be emphasized however, that in general this local gauge transformation does not exist globally, so
it is unjustified to assume that ${\fpt}_{\otm\otn}$ and ${\spt}_{\otm\otn}$  are diagonal
in  some atlas. Generally speaking, this
limitation is present  even in the most favourable situation
when a global coordinate system $\{x^\alpha\}$ 
adapted to the
Killings exist on $\mmm'$ in which the background metric if both block diagonal
and with diagonal $h_{\otm\otn}$.
\end{remarkpro}

\section{Axially symmetric perturbations on spherically symmetric backgrounds}
\label{sec:axis_on_sph}
From now on we restrict ourselves to the case in which
the spacetime dimension is four and the background spacetime
is spherically symmetric. Before introducing some
convenient definitions and notation on spherically symmetric spacetimes
we start by fixing some notation on the unit sphere $(\mathbb{S}^2,\gsph)$.

\subsection{Spherically symmetric backgrounds}
\label{sec:spher_symm}
Let $Y^a$ ($a=1,2,3$) 
be the spherical harmonics with $\ell=1$ on the sphere. More specifically,
$Y^a$ is defined 
as the restriction of the Cartesian
coordinates 
of $\mathbb{R}^3$ to the unit sphere, i.e.
\[
Y^1=\sin\theta\cos\phi,\qquad
Y^2= \sin\theta\sin\phi,\qquad
Y^3=\cos\theta.
\]
We will refer to $\{x^A\}=\{\theta,\phi\}$ as the standard angular spherical coordinates,
in which the spherical unit metric reads $\gsph=d\theta^2+\sin^2\theta d\phi^2$.
Let us denote by $\dsph$ the exterior differential on $\mathbb{S}^2$,
and by
$\D_A$ the covariant derivative on $(\mathbb{S}^2,\gsph)$.
The spherical harmonics $Y^a$  satisfy
$\D_A \D_B Y^a = - Y^a\gsph{}_{AB}$ and the six dimensional algebra
of conformal Killing vectors on $\mathbb{S}^2$ is spanned by
$\{\D_A Y^a (=\dsph Y^a)\} $ (proper conformal Killings) and 
$\{ \epsilon_{AB} \D{}^B Y^a (=-\starsphere\dsph Y^a)\} $ (Killing vectors) 
where $\epsilon_{AB}$ is the volume form of $(\mathbb{S}^2,\gsph)$. 
In standard spherical coordinates
we choose the orientation so that the Hodge dual $\starsphere$
acts as $\starsphere d\phi=-1/\sin\theta d\theta$ and
$\starsphere d\theta=\sin\theta d\phi$.

\begin{definition}
  \label{def:spher}
A (four-dimensional) spacetime $(\mmm,g)$ is \textbf{spherically symmetric}
if it admits an $SO(3)$ group of isometries
acting transitively on spacelike surfaces (which may degenerate to points).
Denote by $\centresph\subset \mmm$ the set of fixed points of the group action (which may be empty).
\end{definition}

As it is known \cite{Szenthe04}, every connected component of  $\centresph$ is the image of a timelike geodesic
which is a closed set in $\mmm$. The set
$\mmm\setminus \centresph$, the principal part of $\mmm$,
is dense in $\mmm$.
Standard results \cite{OT-Schmidt1967} show that the surfaces of transitivity of the group (orbits)
$SO(3)$ generating the spherical symmetry in $(\mmm\setminus\centresph,g)$
are spheres $ \mathbb{S}^2(\subset \mathbb{R}^3)\to S_r\subset\mmm$,
admit a family
of orthogonal (and thus timelike) surfaces $S_\perp$ (the integrable distribution of assumption (ii) in
subsection \ref{sec:OTactions}), and $\mmm\setminus\centresph$ is
diffeomorphic to a warped product $S_\perp \times_{\Q_\perp} \mathbb{S}^2$.
Therefore, on $\mmm\setminus\centresph$ there exist coordinates
$\{x^I, x^A\}$ 
for which  $g$
takes the form
\begin{equation}
  g   = g_{\perp}{(x^I)}_{JK}dx^Jdx^K +\Q_\perp{}^2 \gsph_{AB}dx^A dx^B,
\label{metric_OT}
\end{equation}
where $g_\perp$ is Lorentzian and $\Q_\perp$ is a positive
function $\Q_\perp:S_\perp\to \mathbb{R}$.
The function $\Q_\perp$ is the restriction to $S_\perp$ of a
function $\Q:\mmm\setminus\centresph\to\mathbb{R}$ invariant under the
$SO(3)$ group and which extends to $\centresph$ as  $\Q(\centresph)=0$.
The function $\Q$ measures the area $4\pi \Q^2$ of
the orbits of the $SO(3)$ group, which are thus spheres $(S_r,\Q^2\gsph)$
endowed with the metric $\Q^2\gsph$.

Any vector field $V$ on $\mmm\setminus\centresph$
can  be decomposed as $V=V_\bot+\vect{V}$,
where $\vect{V}=V^A\partial_A$ is tangent to $S_r$ and $V_\bot$ is tangent 
to its $g$-orthogonal complement, $S_\bot$.
For tangent vectors $X=\vect{X}$
we will use $\bm{\overline X} := \gsph(\vect{X}, \cdot)$
to distinguish $\bm{\overline X}=\overline X_Adx^A$ from $\bm X=X_A dx^A=g(X, \cdot)$,
so that $\bm{X} = \Q^2 \bm{\overline X}$.

For later use, it is convenient to express the Lie derivative of any 2-covariant
symmetric tensor $T$ along any vector of the form $\vect{s}=s^A(x^B)\partial_A$.
In coordinates $\{x^I,x^A\}$ this is
\begin{align}
\lie_{\vect{s}} T=s^A\partial_A T_{IJ} dx^I dx^J
+2\left(s^B\partial_B T_{IA}+T_{IB}\partial_As^B  \right) dx^I dx^A+
(\lie_{\vect{s}}\tenstt{T})_{AB}dx^A dx^B,
\label{eq:Lie_s_T}
\end{align}
where we denote by $\tenstt{T}$ the ``full tangent  part'' to $S_r$ of $T$, i.e.
$\tenstt{T}=T_{AB} dx^Adx^B$.
The application of this expression to $T=g$, and taking into account that $\tenstt{g}=\Q^2\gsph$,
yields, for any pair of vectors $X,Y$, the following
equality on $\mmm\setminus\centresph$
\begin{equation}
  \lie_{\vect{s}} g(X,Y)=\Q^2\lie_{\vect{s}}\gsph(\vect{X},\vect{Y}). 
\label{eq:lie_oaxial}
\end{equation}

Let us now single out one axial Killing vector $\axial$
of the $so(3)$ algebra, which after
a  convenient rotation of the spherical coordinates can be set to be $\axial=\partial_\phi$.
We denote the corresponding axis by $\axis$.
Observe that $\centresph\subset \axis$.
Using the above, the axial Killing vector $\axial$, since $\axial=\vect{\axial}$,
defines an axial Killing vector on the unit sphere,
$\axial=\axial^A\partial_A$,
and we take the labels $a$ on $Y^a$
so that the rotation generated by $\eta$ has axis along $x^3$.
By doing that we are saying that the maps $\mathbb{S}^2(\subset \mathbb{R}^3) \rightarrow S_r\subset \mmm$
are such that $x^3=\pm 1$ are mapped onto $\axis\cap S_r$.
We thus have, by construction,
$\overline\axial_A :=  \gsph{}_{AB}\axial^B = \epsilon_{AB} \D^B Y^3$,
this is $\baxial=-\starsphere \dsph Y^3$.
We denote by $\vect{\oaxial}$ the conformal
Killing vector on the sphere given by
\[
  \boaxial\defi \starsphere \baxial=\dsph Y^3=\d Y^3,
\]
where in the last equality we have used the fact that for any function $f(x^A)$
we have $\dsph f=\d f$.

We finally define the vector $\oaxial$
on $(M,g)$ tangent to each sphere and satisfying $\oaxial=\vect{\oaxial}$.
Clearly $\la\oaxial,\oaxial\ra=\la\axial,\axial\ra$ and $\la\oaxial,\axial\ra=0$ by construction.
The vector $\oaxial$ is thus defined on $\mmm\setminus\centresph$
and has the same differentiability as $\axial$ there. However, since $\axial=0$ on $\axis$ (including $\centresph\subset \axis$), $\oaxial$ extends  continuously to $\centresph$ by setting
$\oaxial|_\centresph=0$.
In spherical coordinates $x^A=\{\theta,\phi\}$ we just have $\axial=\partial_\phi$ and
$\oaxial=-\sin\theta\partial_\theta$. Notice also that
$
\bm\oaxial=\Q^2\d Y^3.
$

Consider a spherically symmetric $C^{n+2}$ spacetime $(\mmm,g)$ with $C^{n+1}$ metric $(n \geq 2)$.
Using the construction above  we take a spherically symmetric neighbourhood
$\UinM\subset \mmm$ that may contain a connected component $\centresph_0$ of $\centresph$.
We make the following assumption on the existence of
Cartesian coordinates with suitable differentiability.

\vspace{3mm}

\noindent
\textbf{Assumption {\AsSone}:}
We assume the existence of a $C^{n+2}$ coordinate chart
that maps $\UinM$ onto
$U = U^3\times I $
where $I$ is an open interval and $U^3 \subset \mathbb{R}^3$ is a radially symmetric domain. Moreover, $\centresph_0 \cap \UinM $ (if non empty) is mapped to\footnote{We denote the origin of $\mathbb{R}^p$ by $0_p$, see Appendix \ref{app:iota}.} $ \{0_3\}\times I \subset U$
and using Cartesian coordinates 
$x^i=\{x,y,z\}$ for $\mathbb{R}^3$ and  $\{t\}$ along $I$, the metric $g$ takes the form
\begin{equation}
g=-e^{\nu} dt^2+ 2 \nur dt (x_idx^i)+\VR (x_i dx^i)^2+ 
\UR\delta_{ij}dx^i dx^j,
\label{eq:metric:cartesian}
\end{equation}
where $\nu$, $\nur$, $\VR$  and $\UR$ are $C^{n+1}$ functions of $\{x,y,z,t\}$ and radially symmetric in $\{x,y,z\}$.

\vspace{3mm}

Without loss of generality, we demand that the corresponding spherical coordinates $\{r,\theta,\phi\}$
defined by
\begin{equation}
x=r Y^1=r\sin\theta\cos\phi,\qquad 
y=r Y^2=r\sin\theta\sin\phi,\qquad 
z=r Y^3=r\cos\theta
\label{def:cartesian_coord}
\end{equation}
are such that $\{\theta,\phi\}$ correspond to the above $\mathbb{S}^2\to S_r$ construction.

As usual we identify geometric objects on $\UinM$ with their representation
in this chart. We introduce
$\normr\defi\sqrt{x^2+y^2+z^2}$  (which corresponds to the spherical
coordinate $r$ above) and observe that $\centresph_0 \cap \UinM$ corresponds to
the set of points with vanishing $\normr$.
A consequence of assumption {\AsSone} is that $\UR$ does not vanish on $\UinM$
(as $g$ would degenerate where $\UR=0$),
and the function $\Q$  takes the form $\Q^2=\UR\normr^2$.
Therefore, for any function $f:\UinM\to \mathbb{R}$, 
$\Q^2f\in O(\normr^{l})$ is equivalent to $\UR f\in O(\normr^{l-2})$, and hence
to $f\in O(\normr^{l-2})$. Analogously, if $\Q^2f\in o(\normr^{l})$ then $f\in o(\normr^{l-2})$.

In this chart  
the axial Killing vector $\axial$ reads $\axial=x\partial_y-y\partial_x$
and its square norm
$\trho^2\defi\la\axial,\axial\ra$ is $\trho^2=(\Q^2/\normr^2)( x^2+y^2)=\UR \normx^2$.
The (piece of) axis of symmetry $\axis\cap\UinM$ is thus located at $\axis\cap\UinM=\{p\in\UinM; x(p)=y(p)=0\}$.

In the following (see also Appendix \ref{app:iota})
  we extend the little-o notation in terms of a limit on the axis $\axis$:
For any positive function $g$ defined on  $\UinM \setminus \axis$ we set
(observe we are using boldface in ${\bf o}$)
$$f\in {\bf o}(g)\iff \lim_{\normx\to 0}f g^{-1}=0.$$
Note that for any function $f\in{\bf o}(\normx^l)$ we have
that $f/(\sqrt{\trho^2})^l$ vanishes on all $a\in\axis$.

We define $\partialr=\frac{1}{\normr} x^i\partial_i$ as the radial vector (normal to each $S_r$)
outside the origin, which in spherical coordinates reads $\partialr=\partial_r$.
It is convenient to introduce
the following smooth 
vector field defined on
$\UinM$:
\begin{equation}
\ngamma\defi x\partial_x+y\partial_y 
\label{eq:ngamma}
\end{equation}
and simply use the shorthand $\nzeta\defi\partial_z$.
In terms of these objects, the following expressions hold
on $\UinM\setminus\centresph_0$
\begin{align}
&\oaxial=
-\frac{1}{\normr}(z\ngamma-\normx^2\nzeta), 
\label{eq:oaxial}\\
&\partialr=\frac{1}{\normr}\left(\ngamma+z \nzeta\right).\label{eq:partialr_2}
\end{align}
In addition to being  $C^{n+1}$ on $U\setminus \centresph_0$, it is
clear that $\partialr$ is bounded near the origin and $\oaxial$ extends continuously to $\centresph_0$ as $\oaxial|_{\centresph_0}=0$. An immediate consequence of
(\ref{eq:oaxial})-(\ref{eq:partialr_2}) is 
\begin{equation}
\normr\ngamma=\normx^2\partialr-z\oaxial
\quad \quad \mbox{ on } \quad \UinM\setminus\centresph_0.
\label{eq:ngamma2}
\end{equation}
When convenient  we will use $\{\tz\}$ to refer to $\{z,t\}$.

\subsection{Axially symmetric perturbations on spherically symmetric backgrounds}
In this subsection we prepare the stage that will allow us to prove
  Propositions \ref{prop:K_spher} and \ref{prop:K2_spher} below.
  The aim of the propositions is to look (at first and second order respectively)
for a change of gauge that
takes the perturbation tensors $\Kper$ and $\Kperper$,
in the form given in Proposition \ref{Block},
to a form in which their angular part
is proportional to the metric on the sphere, while keeping the block structure.
This is part of the very well known
Regge-Wheeler gauge. As we shall see, the gauge vectors $\sper$ and $\sperper$
that take us from $\Kper$ and $\Kperper$ to the gauged $\Kper^g$ and $\Kperper^g$
with the desired properties satisfy suitable differential equations with
inhomogeneous terms. Since one of our aims is to understand
whether or not the regularity properties (and differentiability) of $\Kper$ and $\Kperper$
are kept in $\Kper^g$ and $\Kperper^g$, it becomes necessary to study
the properties of those inhomogeneous terms, as well as the effect they have on the
regularity of the gauge vectors via the differential equations that are satisfied.

We start by establishing several facts of
  the original perturbation tensors that will play a role in determining the
  properties of the inhomogeneous terms.
These first results that follow are stated for general symmetric
  tensors invariant under the axial symmetry,
  which will be denoted by $K$ and taken to be of class $C^m$.
  Later on, these results will be applied to
  $\Kper$ and $\Kperper$, of class $C^n$ and $C^{n-1}$ respectively, by letting $m$
  take the values $n$ and $n-1$ correspondingly.

\begin{lemma}
\label{res:q-q+}
Let $K$ be a symmetric 2-covariant $C^m$ 
tensor in $\UinM$ for $m\geq 1$  satisfying $\lie_{\axial}K=0$.
Then,
\begin{enumerate}
\item the functions $K(\partialr,\oaxial)$
  and $K(\partialr,\partialr)$
  are $C^m(\UinM\setminus\centresph_0)$ and radially symmetric in $\{x,y\}$.
  Moreover, $K(\partialr,\oaxial) \in {\bf o}(\normx)$
and $K(\partialr,\partialr)$ is bounded near $\centresph_0$.
 \item the functions  $K(\axial,\axial)/\normx^2$ and 
   $K_{\oaxial\oaxial}^{(m)} := \frac{\normr^2}{\normx^2} K(\oaxial,\oaxial)$
   are $C^m(\UinM)$ and radially symmetric in $\{x,y\}$. Moreover,
   $K_{\oaxial\oaxial}^{(m)}  \in  O(\normr^2)$.   
\item the function $K(\axial,\partial_{u})$, as well as  $q_-$, $q_+$, $q_\times$, $\qtheta$ defined
  by
\begin{align}
&q_-\defi \frac{\Q^2}{2\trho^2}\left\{ K(\oaxial,\oaxial)-K(\axial,\axial)\right\},\qquad
q_+\defi \frac{\Q^2}{2\trho^2}\left\{ K(\oaxial,\oaxial)+K(\axial,\axial)\right\},
\label{def:q-q+}\\
&q_\times\defi \normr \frac{\Q^2}{\trho^2 }K(\oaxial,\axial) ,
\qquad \qquad \qquad \qtheta \defi \normr^2K(\partialr,\oaxial)
\label{def:qx}
\end{align}
are all $C^m(\UinM)$, radially symmetric in $\{x,y\}$ 
and have the following structure
\begin{align}
  &K(\axial,\partial_{u})=\normx^2 \sum^{[m/2]-1}_{k=0} \normx^{2k} P^u_{k}(\tz)+\Phi^{u(m)}, \label{eq:Katz}\\
&q_-=\sum_{k=1}^{[m/2]+1}\normx^{2k}P^-_{0k} +z\Phi^{-(m)}_{10}+\normx^2\Phi^{-(m)}_{01},
\label{eq:q-}\\
&q_+=z^2 P^+_{20}+\sum_{k=1}^{[m/2]+1}\normx^{2k}P^+_{0k}+z\Phi^{+(m)}_{10}+\normx^2\Phi^{+(m)}_{01},
\label{eq:q+}\\
&q_\times= z^2\normx^2  P^\times_{21}+\sum_{k=2}^{[m/2]+1}\normx^{2k}P^\times_{0k}
+ z^2 \Phi^{\times (m)}_{20}+\normx^2 \Phi^{\times (m)}_{01},
\label{eq:qx}\\
 & \qtheta= 
 z\normx^2 P^*_{11}+\sum_{k=2}^{[m/2]+1}\normx^{2k}P^*_{0k}  +z^2\Phi^{*(m)}_{20}+\normx^2\Phi^{*(m)}_{01}
   ,\label{eq:r2Kon}
\end{align}
where all $P$ functions are $C^m$ functions of $\{z,t\}$ and $\Phi^{(m)}\in C^m(U)$ are
radially symmetric in $\{x,y\}$ 
and ${\bf o}(\normx^m)$.\end{enumerate}
In particular, $\qtheta$ and $q_{-}$ are ${\bf o}(\normx)$,
$\qtheta\in O(\normr^3)$ and $q_-\in O(\normr^2)$.
Moreover, if $m\geq 2$ we also have $\partialr(q_-)\in {\bf o}(\normx)$.
 \end{lemma}
\begin{proof}
The neighbourhood $\UinM$
(and its corresponding $U$) and $K$ fit the assumptions of
Lemma \ref{res:axial}. From the definitions in (\ref{eq:ngamma}) it is straightforward to obtain
(with obvious notation when the subindices $u,v$ refer to $w^u=\{z,t\}$,
and $\C$ with no arguments stands for
$\C(\normx,\tz)$ on $\UinM$, etc)
\begin{align}
  &2K(\axial,\axial)=\normx^2(\C-\A),\qquad
    2K(\ngamma,\ngamma)=\normx^2(\C+\A),\qquad
    K(\ngamma,\nzeta)=\traceA_{\rho z},\label{eq:Kaa_Krr}\\
    &K(\nzeta,\nzeta)= K_{zz} = \QQz,
    \qquad \qquad
    K(\axial,\partial_u)
    =-\E
    \label{eq:Kzz}
\end{align}
on $U$. Combining this with  (\ref{eq:oaxial})-(\ref{eq:partialr_2}) lead to
\begin{align}
  K(\partialr,\oaxial)
=&-\frac{z\normx}{\normr^2}\frac{1}{2}\normx(\C+\A-2 \QQz)+\frac{1}{\normr^2}(\normx^2-z^2)\traceA_{\rho z}, 
\label{eq:Kno_pre} \\
K(\partialr,\partialr)
=&\frac{\normx^2}{\normr^2}\frac{\C+\A}{2}+2\frac{z}{\normr}\frac{\traceA_{\rho z}}{\normr}
+\frac{z^2}{\normr^2} \QQz, \label{eq:Knono}\\
K_{\oaxial\oaxial}^{(m)}  := & \frac{\normr^2}{\normx^2} K(\oaxial,\oaxial)
=\frac{1}{2} z^2(\C+\A)-2z \traceA_{\rho z}+\normx^2 \QQz. 
\label{eq:Koo}
\end{align}
By Lemma \ref{res:axial} the
functions $\traceA_{\paramp}(\rho,\tz)$, $\paramp \in \{1,2,3,\rho u,\axial u,uv\}$ admit an expansion
\begin{align}
\traceA_{\paramp}(\rho,z,t) = \sum_{k=0}^{[m/2]} \rho^{2k} P_{\paramp \, k} (z,t)
+ \Phi^{(m)}_{\paramp}(\rho,z,t),  \label{PoltraceA}
\end{align}
where $P_{\paramp \, k} (z,t)$ are $C^m$ functions of $z,t$ and
$\Phi^{(m)}_{\paramp}(\rho,z,t)$ are $o(\rho^m)$ and $C^m$ functions in the domain
$U_{\eta} := \{ \rho \in \mathbb{R}_{\geq 0}, (z,t) \in \mathbb{R}^{2} ;
(\rho,0,z,t ) \in U\} \subset \mathbb{R}_{\geq 0} \times \mathbb{R}^{2}$.
Moreover $P_{2 \, 0}, P_{3 \, 0}, P_{\rho u \, 0}, P_{\axial u\, 0} $ are identically zero
because the functions $\traceA_{2}, \traceA_{3}, \traceA_{\rho u},\traceA_{\axial u}$ vanish at $\rho=0$.
By Lemma \ref{origin} in  Appendix \ref{app:diff_origin}, all the functions
$\traceA_{\paramp}(\normx,z,t)$ are $C^{m}(\UinM)$ (and are obviously radially symmetric in $\{x,y\}$).

Point \emph{1.} of the lemma follows form (\ref{eq:Kno_pre})
and (\ref{eq:Knono}).
The right-hand sides are clearly $C^{m}(\UinM\setminus \centresph_0)$. Since
$z\normx/\normr^2$ vanishes on $\normx=0$, the first term
in (\ref{eq:Kno_pre}) is ${\bf o}(\normx)$.
Boundedness of $\normx^2/\normr^2$ and $z^2/\normr^2$
plus the fact that $\traceA_{\rho z}\in {\bf o}(\normx)$
(because $P_{\rho z \, 0} =0$) 
implies that the second term is also ${\bf o}(\normx)$.
Therefore $K(\partialr,\oaxial)\in {\bf o}(\normx)$.
Concerning $K(\partialr,\partialr)$, all factors in all terms in the right-hand side of (\ref{eq:Knono})
are clearly bounded.
Taking into account that, by Lemma \ref{res:rhon-rn} in Appendix \ref{app:iota}, 
$\traceA_{\rho z}\in {\bf o}(\normx)$ implies
$\traceA_{\rho z}\in o(\normr)$ 
we conclude that $K(\partialr,\partialr)$ is bounded near $\centresph_0$.

Point \emph{2.} is immediate from the first in (\ref{eq:Kaa_Krr}) and (\ref{eq:Koo}), after using again that $\traceA_{\rho z} \in o(\normr)$.

Concerning point \emph{3.}, we compute $q_-$, $q_+$, $q_{\times}, \qtheta$. Directly  from the definitions (\ref{def:q-q+})-(\ref{def:qx}) one finds, using that $\Q^2/\trho^2= \normr^2/\normx^2$,
\begin{align}
  2 q_- & = z^2\A-2z \traceA_{\rho z}+\normx^2 \QQz-\frac{1}{2}\normx^2(\C-\A),
  \label{expr:q-}\\
  2 q_+ & = z^2\C-2z \traceA_{\rho z}+\normx^2 \QQz+\frac{1}{2}\normx^2(\C-\A),
  \label{expr_q+}\\
  q_{\times} & = -\frac{1}{2}z (\normx^2+z^2) \B - (\normx^2+z^2) \traceA_{\axial z},  \label{expr:qtimes} \\
  \qtheta & = - \frac{1}{2} z\normx^2 (\C+\A-2 \QQz)+(\normx^2-z^2)\traceA_{\rho z}.
  \label{expr:q}
\end{align}
Inserting (\ref{PoltraceA}) in these expression as well as
in the expression for $K(\axial,\partial_{u})$
given in (\ref{eq:Kzz})  yields
(\ref{eq:Katz})-(\ref{eq:r2Kon}) after a simple rearranging of terms.
The right-hand sides of (\ref{eq:Katz})-(\ref{eq:r2Kon}) satisfy the requirements
of Lemma \ref{origin} so $K(\eta,\partial_u), q_-, q_+, q_\times, \qtheta$ are
$C^m(\UinM)$ as claimed

For the remaining properties, the fact that  $q, q_{-}\in {\bf o}(\normx)$,
  $\qtheta\in O(\normr^3)$ and $q_-\in O(\normr^2)$ are immediate from
\eqref{eq:r2Kon} and \eqref{eq:q-}. It remains to show that $\partialr(q_{-})\in{\bf o}(\normx)$
when $m\geq 2$.
Observe this is a stronger result than the one provided by Lemma \ref{res:lemma_nkf} for a general
$C^1$ and ${\bf o}(\normx)$ function. The reason behind its validity
is the special structure of $q_{-}$,
\eqref{eq:q-}, which we rewrite here as 
\[
q_-=\sum_{k=1}^{[m/2]+1}\normx^{2k}P^-_{0k} +\Phi_{q-}^{(m)},
\]
where
  $\Phi_{q-}^{(m)}$ is $C^m(U)$, ${\bf o}(\normx^m)$  and radially symmetric in $\{x,y\}$.
Applying $\partialr$ yields
\[
\partialr(q_-)=\sum_{k=1}^{[m/2]+1}\left( 2k \frac{\normx}{\normr}\normx^{2k-1}P^-_{0k}
+\normx^{2k}\frac{z}{\normr}\partial_z P^-_{0k}\right)+\partialr(\Phi_{q-}^{(m)}).
\]
Boundedness of $\frac{\normx}{\normr}$, $\frac{z}{\normr}$, $P^-_{0k}$ and
$\partial_z P^-_{0k}$ imply that the terms in the summation are ${\bf o}(\normx)$.
Lemma \ref{res:lemma_nkf}, since $m\geq 2$ by assumption,
ensures $\partialr(\Phi_{q-}^{(m)})\in {\bf o}(\normx)$.
The result follows.
\fin
\end{proof}

 Our next lemma is concerned with the solutions of a class of differential equations that will arise
  in the process of changing the gauge and on the differentiability and regularity
  properties of the corresponding gauge vector. Most of the technical work is developed in
Appendix  \ref{app:iota}, to which we refer.

\begin{lemma}
\label{res:S_c4_lemma}
Consider the vector field in $\UinM\setminus \centresph_0$
\[S_{(a)}=\frac{1}{\Q^{2a}}\tlemma\oaxial\]
for $a \in \mathbb{R}$, where $\tlemma:\UinM\to \mathbb{R}$  satisfies the equation 
\begin{equation}
  \oaxial(\tlemma)=
  \sum_{(k,l) \in \V} \normrho{x}^{2k} z^lP_{lk}(z,t)
  + \sum_{(k',l') \in \V'}  \normrho{x}^{2k'} z^{l'} \Phi^{(m)}_{l'k'}(x)=:\mathcal{Q}
 \label{eq:for_alpha_1}
\end{equation}
on $\UinM$. Here, $\V, \V'$ are finite subsets of
$\mathbb{N} \times \mathbb{N}$ satisfying, respectively, 
$\V \subset \{ k \geq 1\} \times \{ l \geq 0 \}$ and
$\V' \subset \{ 2k'+l' \geq 1 \}$, 
$P_{lk}$ are $C^m$ in their arguments and $\Phi^{(m)}_{l'k'}$ are $C^m(U)$ and ${\bf o}(\normx^m)$
for $m\geq 1$.
Then there exists an axially symmetric
solution $\tlemma\in C^m(\UinM\setminus\centresph_0)\cap C^0(\UinM)$, and consequently
the vector $S_{(a)}$ is $C^{m}(\UinM\setminus\centresph_0)$.

Moreover, if $\displaystyle{b:= \min_{\V} \{2k+l\}\leq \min_{\V'}\{2k'+l'\}+m}=:c+m$ then
$\tlemma\in O(\normr^{b})$
and $\partialr(\tlemma)\in O(\normr^{b-1})$,
whereas if $b> c+m$  then $\tlemma\in o(\normr^{c+m})$
and $\partialr(\tlemma)\in o(\normr^{c+m-1})$.
On top of that, if $c+m\geq b \geq 2a$ or $b> c+m\geq 2a-1$
then $S_{(a)}$ can be extended continuously to $\centresph_0$ by setting $S_{(a)}|_{\centresph_0}=0$.
\end{lemma}
\begin{proof}
As shown in 
Corollary \ref{res:main_coro} in Appendix \ref{app:iota} there exists an
axially symmetric  solution $\tlemma$ of (\ref{eq:for_alpha_1}) which is $C^m(\UinM\setminus \centresph_0)$
and extends continuously to $\centresph_0$, where it vanishes. By the same corollary, if
$\displaystyle{b\leq c+m}$ then  $\tlemma\in O(\normr^b)$
and $\partialr(\tlemma)\in O(\normr^{b-1})$,
and if
$\displaystyle{b> c+m}$ then  $\tlemma\in o(\normr^{c+m})$
and $\partialr(\tlemma)\in o(\normr^{c+m-1})$.
Clearly $S_{(a)}=\tlemma/\Q^{2a}\oaxial$ for any $a$ is a $C^{m}$ vector field in
$\UinM\setminus\centresph_0$. To obtain its behaviour at $\centresph_0$
we simply analyse the application of $S_{(a)}$ to the Cartesian coordinate functions
using \eqref{eq:oaxial} and \eqref{eq:ngamma}, which provide
\[
S_{(a)}(x)=-\UR^{-a}\frac{1}{\normr^{1+2a}}\tlemma xz,\quad S_{(a)}(y)=-\UR^{-a}\frac{1}{\normr^{1+2a}}\tlemma yz,\quad S_{(a)}(z)=\UR^{-a}\frac{1}{\normr^{1+2a}}\tlemma\normx^2
\]
after using  $\Q^2=\UR\normr^2$. Recall that $\UR\in C^{n+1}(U)$
and does not vanish anywhere in $U$.
If 
$\tlemma\in O(\normr^{b})$ (case $c+m \geq b$)
the three components are $O(\normr^{b+1-2a})$ and hence their limit at
  $\centresph_0$ vanish under the hypothesis $b\geq 2a$.
If $\tlemma\in o(\normr^{c+m})$ (case $b > c+m$) the three components are $o(\normr^{c+m+1-2a})$,
which thus vanish as $\normr\to 0$ if $c+m\geq 2a-1$. In both cases we conclude that  $S_{(a)}$
can be extended continuously to $\centresph_0$ as $S_{(a)}|_{\centresph_0}=0$.\fin
\end{proof}

\subsubsection{Decomposition on spheres of symmetric axially symmetric tensors}
Given the above results we are ready to prove an intermediate but important result
that is the core of the existence of the gauge we look for at first order.
We present it as an independent result
on the decomposition on spherically symmetric spaces
of symmetric axially symmetric tensors into scalar, vector and tensor components.
The importance of this result lies on the fact that it determines not only the existence
(known) but also the differentiability of the
decomposition and the behaviour of such decomposition around the origin.
We use the notation and definitions from subsection \ref{sec:spher_symm}
regarding spherically symmetric spaces.
\begin{theorem}
\label{res:decomp_theorem}
  Let $m \geq 1$ and $(\mmm,g)$ be a spherically symmetric $C^{m+2}$ background with a $C^{m+1}$ metric
satisfying assumption {\AsSone}.
Let $K$ be a symmetric 2-covariant $C^m$ tensor on $\UinM\subset \mmm$ satisfying $\lie_\axial K=0$.
There exists a vector $\bU$ tangent to the spheres $S_r$
which is $C^m$  on $\UinM\setminus\centresph_0$ and extends continuously to $\centresph_0$, where it vanishes,
and a function $\H\in C^m(\UinM\setminus \centresph_0)\cap C^0(\UinM)$ and $O(\normr^2)$
such that $\tenstt{K}$, namely  the tangent-tangent part to the spheres of $K$, decomposes as
\begin{equation}
\tenstt{K}_{AB}=\D_A \bU_B + \D_B \bU_A + \H \gsph_{AB}.
\label{eq:spherical_o}
\end{equation} 
The vector $\bU$ is given explicitly by $\bU=\balpha\oaxial+\bbeta\axial$ with
$C^m(\UinM\setminus\centresph_0)$ functions $\balpha$ and $\bbeta$, both bounded near $\centresph_0$.
\end{theorem}
\begin{remark}
\label{res:decomp_remark}
The proof of the theorem also shows that  the tensor
\begin{equation*}
N\defi K-\lie_{\balpha\oaxial} g
\end{equation*}
satisfies $N(\oaxial,\oaxial)=N(\axial,\axial)$
and $N(\oaxial,\axial)=K(\oaxial,\axial)$.
\end{remark}

\begin{proof}
We start by showing that (\ref{eq:spherical_o}) is equivalent to
\begin{equation}
H\defi K-\lie_\bU g
\label{eq:H}
\end{equation}
satisfying
\begin{equation}
H(\oaxial,\oaxial)=H(\axial,\axial),\quad H(\oaxial,\axial)=0.
\label{eq:spherical}
\end{equation}
Indeed, the contraction of \eqref{eq:H} with any pair of vectors $X=\vect{X} $,
$Y=\vect{Y} $, i.e. tangent to the spheres,
\[
H(X,Y)= K(X,Y)-\lie_\bU g(X,Y),
\]
is equivalent, term by term and on each sphere, to
\[
\tenstt{H}(\vect{X},\vect{Y})
=\tenstt{K}(\vect{X},\vect{Y})-\Q^2 \lie_{\vect{\bU}} \gsph(\vect{X},\vect{Y})
\]
after using \eqref{eq:lie_oaxial} and the fact that $\bU=\vect{\bU}$.
In index notation, this is in turn equivalent to
\[
\tenstt{H}_{AB} 
=\tenstt{K}_{AB}-\Q^2 \lie_{\vect{\bU}} \gsph_{AB}=\tenstt{K}_{AB}-\Q^2 (\D_A \overline \bU_B+ \D_B \overline \bU_A)=
\tenstt{K}_{AB}-\D_A \bU_B- \D_B \bU_A,
\]
and the equivalence between \eqref{eq:spherical_o}
and \eqref{eq:spherical} follows.

Consider $\bU=\balpha\oaxial+\bbeta\axial$.
We start with the following identity for any 2-covariant tensor $T$ 
\begin{equation}
\lie_{\bU}T=\balpha\lie_\oaxial T+ \bbeta\lie_\axial T
+\d\balpha \otimes T(\oaxial,\cdot)+T(\cdot,\oaxial)\otimes \d\balpha
+\d\bbeta \otimes T(\axial,\cdot)+T(\cdot,\axial)\otimes \d\bbeta,
\label{eq:M_S1_to_oaxial_b}
\end{equation}
which applied to $g$ renders
\begin{equation}
\lie_{\bU} g= \balpha\lie_{\oaxial} g +
\d\balpha\otimes\bm\oaxial +\bm\oaxial\otimes\d\balpha
+\d\bbeta\otimes\bm\axial +\bm\axial\otimes\d\bbeta
\label{eq:LieS1_g_b}
\end{equation}
after using $\axial$ is Killing.
Therefore
$H$ takes the form
\[
H =
K -\balpha\lie_{\oaxial}g-\d\balpha\otimes\bm\oaxial
-\bm\oaxial\otimes\d\balpha-\d\bbeta\otimes\bm\axial -\bm\axial\otimes\d\bbeta.
\]
We use now that $\oaxial$ 
is the conformal Killing vector on the sphere with $\bm\oaxial= \Q^2\d Y^3$,
so that using (\ref{eq:lie_oaxial}),
\begin{equation}
  \lie_{\oaxial}g(X,Y) =\Q^2\lie_{\vect{\oaxial}}\gsph(\vect{X},\vect{Y})=-2\Q^2\fy\gsph(\vect{X},\vect{Y}) =-2\fy\la\vect{X},\vect{Y}\ra
  \label{expr:lie_oxaial_g}
  \end{equation}
(recall $\tenstt{g}=\Q^2\gsph$ and $g=g_\perp+\tenstt{g}$),
and therefore
\begin{align}
H(X,Y)&= K(X,Y) +2\fy \balpha \la\vect{X},\vect{Y}\ra
-X(\balpha) \bm\oaxial(\vect{Y}) -\bm\oaxial(\vect{X}) Y(\balpha)
-X(\bbeta) \bm\axial(\vect{Y}) -\bm\axial(\vect{X}) Y(\bbeta).
\label{eq:K1gXY_b}
\end{align}
Hence, in particular,
\begin{align}
  &H(\axial,\axial)=K(\axial,\axial)+2\fy\balpha\la\axial,\axial\ra
    -2\axial(\bbeta)\bm\axial(\axial),\label{eq:K1gxx_b}\\
  &H(\oaxial,\oaxial)=K(\oaxial,\oaxial)+2\fy\balpha \la\oaxial,\oaxial\ra-2\oaxial(\balpha)
    \bm\oaxial(\oaxial),\label{eq:K1gii_b}\\
  &H(\oaxial,\axial)=K(\oaxial,\axial)-\boaxial(\oaxial) \axial(\balpha)
    -\oaxial(\bbeta)\baxial(\axial).\label{eq:K1gix_b}
\end{align}
The equations in \eqref{eq:spherical} are therefore equivalent to
\begin{align}
  \label{eq:for_alpha_prop_b}
  0&=K(\axial,\axial)-K(\oaxial,\oaxial)
  +2\left(\oaxial(\balpha)-\axial(\bbeta)\right)\trho^2, \\
  \label{eq:for_beta_prop_b}
  0& =K(\oaxial,\axial)- \left(\axial(\balpha)+\oaxial(\bbeta)\right) \trho^2,
\end{align}
respectively,
after using $\bm\oaxial(\oaxial)=\bm\axial(\axial)=\la\axial,\axial\ra=:\trho^2$.
These equations clearly imply
$\axial(\oaxial(\balpha))+\axial(\axial(\bbeta))=0$ and $\axial(\oaxial(\bbeta))+\axial(\axial(\balpha))=0$.
Since $[\axial,\oaxial]=0$ it suffices to consider $\balpha$ and $\bbeta$ such that
$\axial(\balpha)=0$ and $\axial(\bbeta)=0$,
and
the equations become
\begin{align}
  0 & =K(\axial,\axial)-K(\oaxial,\oaxial)
  +2\oaxial(\balpha)\trho^2,
\label{eq:for_alpha_prop_bb} \\
  \label{eq:for_beta_prop_bb}
  0 & =K(\oaxial,\axial)- \oaxial(\bbeta)\trho^2.
  \end{align}
 We thus have two separate ODEs, one for $\balpha$ and one for $\bbeta$.
We deal first with equation (\ref{eq:for_alpha_prop_bb}), which
can be cast as
\begin{equation}
\oaxial(\Q^2\balpha)=q_{-}
\label{eq:alpha_first_b}
\end{equation}
with $q_-$ given by \eqref{def:q-q+}.
Lemma \ref{res:q-q+}
ensures that $q_{-}$ satisfies the requirements of
the right-hand side $\mathcal{Q}$ of equation (\ref{eq:for_alpha_1})
of Lemma \ref{res:S_c4_lemma} with  $b=\displaystyle{\min_{\V}\{2k+l\}}=2$
and  $c=\displaystyle{\min_{\V'}\{2k'+l'\}}=1$.
Note that Lemma \ref{res:q-q+} also ensures $q_{-}\in {\bf o}(\normx)$ and $O(\normr^2)$.
By setting $\tlemma=\Q^2\balpha$ and $\mathcal{Q}= q_{-}$,
Lemma \ref{res:S_c4_lemma} thus establishes there exists an axially symmetric
solution $\Q^2 \balpha\in C^m(\UinM\setminus \centresph_0)\cap C^0(\UinM)$ of \eqref{eq:alpha_first_b},
and consequently the vector $\bU_\oaxial\defi\balpha\oaxial=\Q^{-2}\tlemma\oaxial$ is
$C^{m}(\UinM\setminus\centresph_0)$.
Moreover, since we have $b(=2)\leq c+m$ given that $m\geq 1$, $\Q^2 \balpha$ is also $O(\normr^2)$,
and since $b\geq 2a$ for $a=1$, the vector $\bU_\oaxial=\Q^{-2}\tlemma\oaxial$ extends continuously to $\centresph_0$, where it vanishes.
Note that $\Q^2 \balpha\in O(\normr^2)$
implies $\balpha\in C^m(\UinM\setminus \centresph_0)$ is bounded near $\centresph_0$, as claimed.
The use of \eqref{eq:K1gix_b} with $\bbeta=0$ (and $\axial(\balpha)=0$)
together with the above proves, in particular, Remark \ref{res:decomp_remark}.

We next analyse in an analogous manner equation (\ref{eq:for_beta_prop_bb}), which we write
as
\begin{equation}
\oaxial(\normr \Q^2 \bbeta)=q_{\times}
\label{eq:beta_first_b}
\end{equation}
with $q_\times$ given by \eqref{def:qx}.
Lemma \ref{res:q-q+}  
ensures now that $q_{\times}$ satisfies the requirements of $\mathcal{Q}$ in
Lemma \ref{res:S_c4_lemma}
with $b=4$ and  $c=2$.
By setting $\tlemma=\normr \Q^2\beta$ and $\mathcal{Q}=q_{\times}$
Lemma \ref{res:S_c4_lemma} thence ensures that
$\normr \Q^2 \bbeta \in C^m(\UinM\setminus \centresph_0)\cap C^0(\UinM)$
and $O(\normr^4)$ if $m\geq 2$ (because then $b\leq c+m$) and $o(\normr^3)$ if $m=1$
(because $b>c+m$).
This implies, in any case, that $\bbeta \in C^m(\UinM\setminus \centresph_0)$ can be continuously extended
to $\centresph_0$, where it vanishes. As a result, the vector $\bU_\axial\defi \bbeta\axial$
is $C^m(\UinM\setminus \centresph_0)$ and can be continuously extended to $\centresph_0$ as zero.

Clearly, the vector $\bU=\bU_\oaxial+\bU_\axial$ satisfies both conditions in
\eqref{eq:spherical} and thence the outcome of the theorem.

Finally, a straightforward calculation shows, since $\axial(\balpha)=0$, that
$\D^A \bU_A=\D^A \bU_{\oaxial}{}_A=\Q^2(\oaxial(\balpha)- 2 \balpha Y^3)$,
from where it is direct to arrive at
\begin{equation}
\H=\frac{1}{2}\mbox{tr}_{\mathbb{S}^2}\tenstt{K}-\D^A \bU_{\oaxial}{}_A
=\frac{\Q^2}{\trho^2} K(\axial,\axial)+2\Q^2\balpha Y^3
\label{eq:h}
\end{equation}
after using \eqref{eq:for_alpha_prop_bb}.
Since $\Q^2/\trho^2=\normr^2/\normx^2$, we can use point \emph{2.} of Lemma \ref{res:q-q+}
to conclude that $\frac{\Q^2}{\trho^2}K(\axial,\axial)$ is $C^m(\UinM)$ and $O(\normr^2)$,
while we have from the above that
$\Q^2 \balpha\in C^m(\UinM\setminus \centresph_0)\cap C^0(\UinM)$ and $O(\normr^2)$.
Therefore $\H\in C^m(\UinM\setminus \centresph_0)\cap C^0(\UinM)$ and $O(\normr^2)$.\fin
\end{proof}

\subsubsection{Choice of gauge at first order}
We are ready to show that given any axially symmetric perturbation
there exists gauge vectors that render the full angular part of the perturbations
in some convenient manner. At this point we could use the previous theorem in full in order to
achieve a perturbation tensor that is proportional to the unit sphere metric (at first order).
However, for our purposes we will only need the partial result given by its
Remark \ref{res:decomp_remark}, leaving aside the $\{\theta,\phi\}$ crossed term,
which makes things simpler.
That is because later we will focus on perturbations that inherit
an orthogonally transitive two-dimensional group of isometries, in which case the crossed
term vanishes from the beginning, and thus it suffices to take $\bbeta=0$.

\begin{proposition}\label{prop:K_spher}
  Let $n \geq 1$ and $(\mmm,g)$ be a spherically symmetric $C^{n+2}$ background
  with a $C^{n+1}$ metric  satisfying assumption {\AsSone}.
  Let  $\fpt$ be a $C^n(\UinM)$ first order perturbation tensor satisfying
  $\lie_{\axial} \fpt=0$. Then
  there exists  a $C^n (\UinM\setminus \centresph_0)$
  first order gauge vector $\sper=-\talpha\oaxial$,
  that extends continuously to zero on $\centresph_0$,
  such that the corresponding gauge transformed tensor $\Kper^g$ (which is automatically
  $C^{n-1}(\UinM\setminus \centresph_0)$) satisfies
  \begin{equation}
    \fpt^g(\axial,\axial)=\fpt^g(\oaxial,\oaxial),
    \label{proport_cov_g}
  \end{equation}
  where
  \begin{equation}
    \fpt^g(\axial,\axial)=\fpt(\axial,\axial)+2\frac{z}{\normr}\talpha \trho^2.\label{eq:K1gxx_res}
  \end{equation}
  In addition
  \begin{align}
    &\fpt^g(\oaxial,\axial)=\fpt(\oaxial,\axial),\label{eq:K1gix}\\
    &\fpt^g(\oaxial,\vort_1)=\fpt(\oaxial,\vort_1)-\vort_1(\talpha)\trho^2,
      \label{eq:K1gi1}\\
    &\fpt^g(\axial,\vort_1)=\fpt(\axial,\vort_1),\label{eq:K1gx1}\\
    &\fpt^g(\vort_1,\vort_2)=\fpt(\vort_1,\vort_2),\label{eq:K1g12}
  \end{align}
  for any vectors $\vort_1$ and $\vort_2$ orthogonal to the spheres $S_r$.
  The function $\talpha$ is axially symmetric,
  of class $C^n(\UinM\setminus \centresph_0)$ and bounded near $\centresph_0$,
  and $\normr\partialr(\talpha)\in C^{n-1}(\UinM\setminus \centresph_0)$ is also bounded near $\centresph_0$.
  
  Moreover,
  the function $\Kper^g(\axial,\axial)/\trho^2$ is $C^n(\UinM\setminus \centresph_0)$,
  axially symmetric and bounded near $\centresph_0$,
  and the function
  $$\qtheta^g_1:=\normr^2\fpt^g(\oaxial,\partialr) \quad \mbox{ outside } \centresph_0, \quad
  \qtheta^g_1(\centresph_0)=0$$ is $C^{n-1}(\UinM\setminus\centresph_0)\cap C^0(\UinM)$, ${\bf o}(\normx)$
  and takes the  form
  \begin{equation}
    \label{eq:q1theta}
    \qtheta^g_1=
    z\normx^2 P^{*}_{11}+\sum_{k=2}^{[n/2]+1}\normx^{2k}P^{*}_{0k}
    +\normx^2\Phi^{*(n)}_{01}
    +z^2\Phi^{* (n)}_{20}+\normx^2\Phizero_1^{(n-1)}\,
  \end{equation}
  where the $P^{*}$ functions are $C^n$ functions of $\{z,t\}$, $\Phi^{*(n)}\in C^n(\UinM)$ are
  radially symmetric in $\{x,y\}$ 
  and ${\bf o}(\normx^n)$ and $\Phizero_1^{(n-1)}\in C^{n-1}(\UinM\setminus\centresph_0)\cap  C^0(\UinM)$
  is radially symmetric in $\{x,y\}$ 
  and $O(\normr)$.
\end{proposition}

\begin{proof}
  Let $\talpha$ be the function $\balpha$ of Theorem \ref{res:decomp_theorem} applied to $K = \fpt$ and $n=m$, and define the vector $\sper=-\talpha\oaxial$.
  From  that theorem we know that
  $\talpha$ is $C^n(\UinM \setminus \centresph_0)$ and bounded near the origin.
  By Remark \ref{res:decomp_remark} we have that
  $\fpt^g= \fpt +\lie_{\sper}g$ satisfies \eqref{proport_cov_g} and \eqref{eq:K1gix}.
  
  For later use we note that
  \eqref{eq:K1gXY_b} with $H=\fpt^g$, $K=\fpt$, $\balpha=\talpha$ and $\bbeta=0$ yields, for any any pair of vectors $X,Y$,
  \begin{align}
    \fpt^g(X,Y)
    &=\fpt(X,Y) +2\fy \talpha\la\vect{X},\vect{Y}\ra
      -X(\talpha) \bm\oaxial(\vect{Y}) -\bm\oaxial(\vect{X}) Y(\talpha).
      \label{eq:K1gXY}
  \end{align}
  The equation satisfied by $\talpha$ is
  (from (\ref{eq:alpha_first_b}))
  \begin{align}
    \oaxial(\Q^2\talpha) & =q_{1-}, \qquad \quad q_{1-}\defi \frac{\Q^2}{2\trho^2}\left\{\Kper(\oaxial,\oaxial)-\Kper(\axial,\axial)\right\}. 
                           \label{eq:alpha_first} \\
                         & \Longleftrightarrow \qquad
                           0=K_1(\axial,\axial)-K_1(\oaxial,\oaxial)+2\oaxial(\tilde\alpha)\trho^2.
                           \label{eq:for_alpha_prop}
  \end{align}
  
  As discussed in
  the proof of Theorem \ref{res:decomp_theorem} (for $K=\Kper$, $m=n$
  and $\qtheta_{-}=\qtheta_{1-}$),
  \eqref{eq:alpha_first}
  satisfies the requirements of
  Lemma \ref{res:S_c4_lemma} with $\tlemma=\Q^2\talpha$, $\mathcal{Q}=q_{1-}$ and $m=n\geq 1$,
  with $b=2$ and $c=1$.

  Equations (\ref{eq:K1gxx_res}) and 
  \eqref{eq:K1gi1}-\eqref{eq:K1g12} follow immediately from \eqref{eq:K1gXY} after using 
  $\axial(\talpha)=0$, $\bm\oaxial(\oaxial)=\la\axial,\axial\ra=:\trho^2$.
  The claim that  $\Kper^g(\axial,\axial)/\trho^2$
  is $C^n (\UinM \setminus \centresph_0)$ and bounded near the centre follows
  from equation (\ref{eq:K1gxx_res}),
  since $\talpha$ (by Theorem \ref{res:decomp_theorem}) and
  $\Kper(\axial,\axial)/\trho^2$ (by point \emph{2.} in Lemma \ref{res:q-q+}
  applied to $K= \fpt$ and $m=n$) are both
  $C^n (\UinM \setminus \centresph_0)$ and bounded near the centre.

  Concerning the properties of $\qtheta^g_1$, we use 
  (\ref{eq:K1gi1}) with  $\vort_1=\partialr$, to get (recall that $\trho^2 = \UR \normx^2$)    
  \begin{equation}
    \qtheta^g_1 :=
    \normr^2\fpt^g(\oaxial,\partialr)=\normr^2\fpt(\oaxial,\partialr)
    + \normx^2 \Phizero^{(n-1)}_1, \qquad  \Phizero^{(n-1)}_1  :=-\UR\normr^2\partialr(\talpha).
    \label{eq:r2K1on}
  \end{equation}
  By point \emph{3.} in Lemma \ref{res:q-q+} (with $K=\Kper$, $m=n$),
  the first term extends to a $C^n(\UinM)$ function
  admitting an expression of the form \eqref{eq:r2Kon}. For the
  second term, we compute
  (using  $\Q^2=\UR \normr^2$),
  \begin{equation}
    \normr  \partialr (\talpha)
    =  \normr \partialr(\tlemma/\Q^2)
    = \frac{1}{\UR} \left (
      \frac{\partialr(\tlemma)}{\normr}
      -\frac{\tlemma}{\normr^2}\left(\normr\partialr(\UR)\frac{1}{\UR}+2 \right) \right ).
    \label{eq:ntalpha}
  \end{equation}
  Since $\tlemma\in O(\normr^2)$ and $\partialr(\tlemma)\in O(\normr^{1})$
  by virtue of Lemma \ref{res:S_c4_lemma} (since $c+m \geq b = 2$, see above), and
  the term $\normr\partialr(\UR)/\UR=\frac{1}{\UR}x^i\partial_i \UR$ is
  $C^n(U)$ (because $\UR\in C^{n+1}(U)$ and does not vanish on $\UinM$), we
  conclude that
  $\normr\partialr(\talpha)$ is 
  bounded near $\centresph_0$, as claimed. This implies in particular
  that $\Phizero^{(n-1)}_1$ is $O(\normr)$. Given that
  $\talpha\in C^n(\UinM\setminus \centresph_0)$ it follows immediately that
  $\Phizero^{(n-1)}_1$ (extended at the centre with the value zero) is
  $C^{n-1}(\UinM\setminus \centresph_0)\cap C^0(\UinM)$
  and $O(\normr)$.  We conclude that
  the function $\qtheta^g_1$ in \eqref{eq:r2K1on}
  is $C^{n-1}(\UinM\setminus \centresph_0)\cap C^0(\UinM)$,
  and takes the form \eqref{eq:q1theta}, from where it is direct to check
  that, since $n\geq 1$,
  $\qtheta^g_1$ is also ${\bf o}(\normx)$.
  \fin
\end{proof}

\subsubsection{Choice of gauge at second order}
We now perform an analogous, but more involved, procedure for the second order perturbation.

\begin{proposition}\label{prop:K2_spher}
  Assume the setting of Proposition \ref{prop:K_spher} and restrict
  $n\geq 2$.
  Let $\spt$ be a
  $C^{n-1}(\UinM)$ second order perturbation tensor satisfying
  $\lie_{\axial} \spt =0$. 
  Then,
  there exists a $C^{n-1}(\UinM\setminus\centresph_0)$
  second order gauge vector
  $\sperper=-\tX\oaxial$, that extends continuously to zero on $\centresph_0$,
  such that the corresponding gauge transformed tensor
  $\Kperper^g$ (which is immediately $C^{n-2}(\UinM\setminus\centresph_0)$) satisfies
  \begin{equation}
    \spt^g(\axial,\axial)=\spt^g(\oaxial,\oaxial).
    \label{proport_cov_g_second}
  \end{equation}
  In addition,
  \begin{align}
    \Kperper^g(\oaxial,\axial)=&\spt(\oaxial,\axial)-2\oaxial\left(\talpha\Kper(\oaxial,\axial)\right),\label{eq:K2gix}\\
    \Kperper^g(\oaxial,\vort_1)=&\Kperper(\oaxial,\vort_1)-\vort_1(\tX)\trho^2-2\oaxial\left(\talpha\Kper(\oaxial,\vort_1)\right)+\talpha\trho^2\oaxial\left(\vort_1(\talpha)\right)-2\vort_1(\talpha)\Kper(\oaxial,\oaxial)\nonumber\\
                           &+\talpha (\Kper^g+\Kper)(\oaxial,[\oaxial,\vort_1])+3\oaxial(\talpha)\vort_1(\talpha)\trho^2
                             -4\talpha \frac{z}{\normr}\trho^2\vort_1(\talpha),\label{eq:K2gi1}\\
    \Kperper^g(\axial,\vort_1)=&\Kperper(\axial,\vort_1)-2\talpha\oaxial\left(\Kper(\axial,\vort_1)\right)
                                 -2\vort_1(\talpha)\Kper(\axial,\oaxial)+\talpha(\Kper^g+\Kper)(\axial,[\oaxial,\vort_1]),\label{eq:K2gx1}\\
    \Kperper^g(\vort_1,\vort_2)=&\spt(\vort_1,\vort_2)-2\talpha\oaxial\left(\Kper(\vort_1,\vort_2)\right)
                              +2\vort_1(\talpha)\vort_2(\talpha)\trho^2
                                              +\talpha (\Kper^g+\Kper)([\oaxial,\vort_1],\vort_2) \nonumber \\
                           & +\talpha (\Kper^g+\Kper)(\vort_1,[\oaxial,\vort_2])
                    -2\vort_1(\talpha)\Kper(\oaxial,\vort_2)-2\vort_2(\talpha)\Kper(\oaxial,\vort_1).
                    \label{eq:K2g12}
  \end{align}
  for any vectors $\vort_1$ and $\vort_2$ orthogonal to the spheres $S_r$. The function $\tX$ is
  axially symmetric, $C^{n-1}(\UinM\setminus \centresph_0)$ and bounded near $\centresph_0$.
  
  Moreover, the function $\Kperper^g(\axial,\axial)/\trho^2$  is $C^{n-1}(\UinM\setminus \centresph_0)$
  and bounded near $\centresph_0$, and the function
  $$\qtheta^g_{2}:=\normr^2\spt^g(\oaxial,\partialr) \quad \mbox{ outside } \centresph_0, \quad
  \qtheta^g_{2}(\centresph_0)=0$$ 
  is  $C^{n-2}(\UinM\setminus\centresph_0)\cap C^0(\UinM)$  and ${\bf o}(\normx)$
  and takes the form
  \begin{align}
    \label{eq:q2theta}
    \qtheta^g_2
    &=z\normx^2 P^{**}_{11}+\sum_{k=2}^{[(n-1)/2]+1}\normx^{2k}P^{**}_{0k}
    +\normx^2\Phi^{**(n-1)}_{01}
    +z^2\Phi^{** (n-1)}_{20}+\normx^2 \Phizero_2^{(n-2)}
    -\oaxial\left(2\talpha \qtheta_1\right)
  \end{align}
  where $\qtheta_1 := \normr^2 \Kper(\oaxial,\partialr)$, all the $P^{**}$ functions are $C^{n-1}$ functions of $\{z,t\}$, $\Phi^{**(n-1)}\in C^{n-1}(\UinM)$ are
  radially symmetric in $\{x,y\}$
  and ${\bf o}(\normx^{n-1})$ and $\Phizero_2^{(n-2)}\in C^{n-2}(\UinM\setminus\centresph_0)\cap  C^0(\UinM)$
  is radially symmetric in $\{x,y\}$ and $O(\normr)$. Moreover, the function
  $\oaxial(\talpha\qtheta_1)$ is $C^{n-1}(\UinM\setminus\centresph_0)$,
  ${\bf o}(\normx)$ (in particular, it admits a continuous extension to $\centresph_0$ with value zero).
\end{proposition}

\begin{proof}
By Proposition \ref{prop:K_spher} there exists 
a first order gauge vector $\sper=-\talpha\oaxial$,
where $\talpha\in C^n(\UinM\setminus\centresph_0)$ satisfies \eqref{eq:for_alpha_prop},
so that (\ref{proport_cov_g}) holds. The second order gauge transformation
  (\ref{gaugeperper_bis}) can be rewritten
(upon combining with (\ref{gaugeper})) as
\begin{equation}
  \label{gaugeperper_bona}
  {{\spt}^{g}} = {\spt} + \lie_{\sperper} 
  g + \lie_{\sper} M, \qquad
M\defi {\fpt}^g+{\fpt}.  
\end{equation}
Consider the second order gauge vector $\sperper=-\tX\oaxial$. 
Immediate consequences of \eqref{eq:M_S1_to_oaxial_b} (with $\bbeta=0$)
  and \eqref{expr:lie_oxaial_g} are
\begin{align*}
  \lie_{\sperper}g (X,Y)  = &
  2\fy\tX \la\vect{X},\vect{Y}\ra -X(\tX)\bm\oaxial(\vect{Y})-Y(\tX)\bm\oaxial(\vect{X}), \\
   \lie_{\sper}M(X,Y)  =&-\talpha\oaxial\left(M(X,Y)\right)+\talpha M([\oaxial,X],Y)+\talpha M(X,[\oaxial,Y]) \\ &- X(\alpha)M(\oaxial,Y)-M(X,\oaxial)Y(\talpha),
\end{align*}
where to get the second we also used te Leibniz rule for the Lie derivative. Thus, \eqref{gaugeperper_bona} is 
\begin{align}
  \Kperper^g(X,Y)=&\Kperper(X,Y)
  +2\fy\tX \la\vect{X},\vect{Y}\ra -X(\tX)\bm\oaxial(\vect{Y})-Y(\tX)\bm\oaxial(\vect{X})
  -\talpha\oaxial\left(M(X,Y)\right) \nonumber \\
  & +\talpha M([\oaxial,X],Y)+\talpha M(X,[\oaxial,Y])
  - X(\talpha)M(\oaxial,Y)-M(X,\oaxial)Y(\talpha).
\label{eq:K2gXY}
\end{align}
The explicit form of $M$ follows from Proposition \ref{prop:K_spher}, which gives
\begin{align}
&M(\axial,\axial)=2\Kper(\axial,\axial)+2\fy\talpha  \trho^2,\label{eq:Mxx}\\
&M(\oaxial,\oaxial)=2\Kper(\oaxial,\oaxial)+2\fy\talpha  \trho^2-2\trho^2\oaxial(\talpha),\label{eq:Mii}\\
  &M(\oaxial,\axial)=2\Kper(\oaxial,\axial), \qquad  \qquad
M(\oaxial,\vort_1)=2\Kper(\oaxial,\vort_1)-\vort_1(\talpha)\trho^2,
  \label{eq:Mi1} \\
&  M(\axial,\vort_1)=2\Kper(\axial,\vort_1), \qquad \qquad M(\vort_1,\vort_2)=2\Kper(\vort_1,\vort_2), \label{eq:Mi1:rest}
\end{align}
where, to get the second, we have used \eqref{eq:for_alpha_prop}.
Inserting these into (\ref{eq:K2gXY}) and using
  $[\oaxial,\axial]=0$, $\eta(\talpha)=0$, a straightforward computation yields
\eqref{eq:K2gix}-\eqref{eq:K2g12}, as well as
\begin{align}
\Kperper^g(\oaxial,\oaxial)=&\Kperper(\oaxial,\oaxial)+2\fy\tX \trho^2-2\trho^2\oaxial(\tX)
-2\talpha\oaxial\left(\Kper(\oaxial,\oaxial)\right)+\talpha\oaxial\left(-2\fy\talpha \trho^2+2\trho^2\oaxial(\talpha)\right)\nonumber\\
&-4\oaxial(\talpha)\Kper(\oaxial,\oaxial)-4\fy\talpha\oaxial(\talpha)  \trho^2+4\trho^2\oaxial(\talpha)^2, \label{eq_K2gii} \\
\Kperper^g(\axial,\axial)=&\Kperper(\axial,\axial)+2\fy\tX  \trho^2-2\talpha \oaxial\left(\Kper(\axial,\axial)\right)-2\talpha \oaxial\left(\talpha \fy \trho^2\right).\label{eq:K2gxx}
\end{align}
We may now find the explicit form of equation (\ref{proport_cov_g_second}).
Substracting \eqref{eq_K2gii} and \eqref{eq:K2gxx}  one finds, after trivial rearrangements,
that   (\ref{proport_cov_g_second})  becomes
\begin{align}
0=&\Kperper(\axial,\axial)-\Kperper(\oaxial,\oaxial)
-2\talpha \oaxial\left(\Kper(\axial,\axial)-\Kper(\oaxial,\oaxial)\right)+2\trho^2\oaxial(\tX)
-\talpha\oaxial\left(2\trho^2\oaxial(\talpha)\right)\nonumber\\
&+4\oaxial(\talpha)\Kper(\oaxial,\oaxial)+4\fy\talpha \oaxial(\talpha)\trho^2-4\trho^2\oaxial(\talpha)^2. \label{eq:difKperper}
\end{align}
In order to analyse this equation, it turns out to be convenient to
introduce an the auxiliary $\tXX\defi \Q^2\tX+\talpha q_{1-}$, where
$q_{1-}$ was introduced in \eqref{eq:alpha_first}
and it is $C^n(\UinM)$ and $O(\normr^2)$ by Lemma \ref{res:q-q+} applied to
  $K = \Kper$,  $m=n$, $q_-=q_{1-}$. Using
the first equation in (\ref{eq:alpha_first}), a straightforward computation shows that \eqref{eq:difKperper} can be rewritten as
\begin{equation}
  \oaxial(\tXX)=q_{2-}+\frac{1}{\Q^2} q_{1-}(q_{1-}-2q_{1+}),
  \label{eq:tX_final}
\end{equation}
where
\begin{align*}
&q_{1+}\defi \frac{\Q^2}{2\trho^2}\left\{\Kper(\oaxial,\oaxial)+\Kper(\axial,\axial)\right\}, \quad 
&q_{2-}\defi \frac{\Q^2}{2\trho^2}\left\{\Kperper(\oaxial,\oaxial)-\Kperper(\axial,\axial)\right\}.
\end{align*}
The right-hand side of \eqref{eq:tX_final} is invariant under $\axial$, so it
suffices to look for $\tXX$ satisfying $\axial(\tXX)=0$.
It is convenient to decompose $\tXX=\tXX_{(0)}+\tXX_{(1)}$ and split (\ref{eq:tX_final}) into
the two equations
\begin{align}
&\oaxial(\tXX_{(0)})=q_{2-},\label{eq:tXX_0}\\
&\oaxial(\Q^2\tXX_{(1)})=q_{1-}(q_{1-}-2q_{1+}).\label{eq:tXX_1}
\end{align}
The splitting is such that the right-hand side is $C^{n-1} (\UinM)$ in the first equation and $C^n(\UinM)$ in the second. Concerning the first equation
Lemma \ref{res:q-q+}, applied to $K=\Kperper$ and $m= n-1\geq 1$,
tells us that $q_{2-}$ has the form (\ref{eq:q-}) 
and thus satisfies the requirements of Lemma \ref{res:S_c4_lemma} with
$b=\displaystyle{\min_{\V}\{2k+l\}}=2$ and  $c=\displaystyle{\min_{\V'}\{2k'+l'\}}=1$,
and thence $b\leq c+m$.  Applying the lemma, there exists a solution
  $\tXX_{(0)}$ satisfying
  \begin{align}
    \tXX_{(0)} \in C^{n-1}(\UinM\setminus \centresph_0)\cap C^0(\UinM), \qquad
    \tXX_{(0)} \in O(\normr^2), \qquad  \partialr(\tXX_{(0)}) \in O(\normr).
    \label{prop:TXX0}
  \end{align}
Regarding equation (\ref{eq:tXX_1}) we first  obtain the structure of its right-hand side. From Lemma \ref{res:q-q+} applied to $K= \fpt$ (specifically
  from expressions (\ref{eq:q-}) and (\ref{eq:q+}) for $q_{1-}$ and $q_{1+}$) it follows
\[
q_{1-}(q_{1-}-2q_{1+})=z^2\sum_{k=1}^{[n/2]+1}\normx^{2k}\breve{P}_{2k} +
\sum_{k=2}^{2[n/2]+2}\normx^{2k}\breve{P}_{0k}
+\normx^4\breve\Phi^{(n)}_{02}+z^2 \breve\Phi^{(n)}_{20}+z\normx^2\breve\Phi^{(n)}_{11}
\]
where the functions $\breve P$  are $C^n$ in $\{z,t\}$ and $\breve \Phi^{(n)}\in C^n(\UinM)$
radially symmetric in $\{x,y\}$ and ${\bf o}(\normx^n)$. Thus, the right-hand side of (\ref{eq:tXX_1}) satisfies the hypotheses of Lemma \ref{res:S_c4_lemma}
with $m=n\geq 2$, 
$b=\displaystyle{\min_{\V}\{2k+l\}}=4$ and
$c=\displaystyle{\min_{\V'}\{2k'+l'\}}=2$, so that $b\leq c+m$.
By this lemma, there exists a solution $\Q^2 \tXX_{(1}$ of \eqref{eq:tXX_1}
  satisfying
  \begin{align}
    &\Q^2\tXX_{(1)}\in C^{n}(\UinM\setminus \centresph_0)\cap C^0(\UinM),
    \qquad
    \Q^2\tXX_{(1)}\in   O(\normr^4), \qquad
    \partialr(\Q^2\tXX_{(1)}) \in O(\normr^3) \nonumber \\
    & \Longrightarrow \qquad
    \tXX_{(1)}\in C^{n}(\UinM\setminus \centresph_0)\cap C^0(\UinM), \qquad
    \tXX_{(1)}\in   O(\normr^2), \qquad
    \partialr(\tXX_{(1)})\in O(\normr). \label{prop:TXX1}
  \end{align}
Combining \eqref{prop:TXX0} and \eqref{prop:TXX1}, there is a solution
  $\tXX$ of (\ref{eq:tX_final}) satisfying
\begin{align}
  \label{prop:tXX}
  \tXX \in C^{n-1}(\UinM\setminus\centresph_0)\cap C^0(\UinM), \qquad
  \tXX \in O(\normr^2), \qquad
  \partialr(\tXX)\in O(\normr).
\end{align}
This, together with the fact that
$\talpha\in C^n(\UinM\setminus\centresph_0)$ is bounded (by Proposition \ref{prop:K_spher}),
and $q_{1-}$ is $C^n(\UinM)$ and $O(\normr^2)$, imply that
\begin{equation}
\tX=\Q^{-2}(\tXX-\talpha q_{1-}) \quad \quad \mbox{
  is} \quad  C^{n-1}(\UinM\setminus\centresph_0) \quad \quad \mbox{and bounded near}
\quad \centresph_0
\label{warning}
\end{equation}
(note that $\tX$ is not necessarily defined at the origin).
Since $\oaxial$ extends continuously to the centre as the zero vector,
  boundedness of $\tX$ implies that $\sperper=-\tX\oaxial$ extends continuously to the centre with the value zero.

We deal now with the properties of
$\Kperper^g(\axial,\axial)/\trho^2$.  First, we rewrite equation (\ref{eq:K2gxx}) as
\begin{equation}
\frac{1}{\normx^2}\Kperper^g(\axial,\axial)=\frac{1}{\normx^2}\Kperper(\axial,\axial)
+2\frac{z}{\normr}\UR\tX
-2\talpha \frac{1}{\normx^2}\oaxial\left(\Kper(\axial,\axial)\right)
-2\frac{1}{\normx^2}\frac{1}{\normr}\oaxial\left(\talpha z \trho^2\right).
\label{eq:K2aa_g}
\end{equation}
We want to check that all terms
in the right-hand side are $C^{n-1}(\UinM\setminus\centresph_0)$ and bounded near $\centresph_0$.
For the first term this is immediate from
item \emph{2.} of  Lemma \ref{res:q-q+} applied to $\Kperper$ and $m=n-1 \geq 1$.
For the second, it follows from \eqref{warning}  and the properties of $\UR$.
Concerning the third term, point \emph{2.} in Lemma \ref{res:q-q+} implies that
$\Kper(\axial,\axial)$ admits a decomposition of the form
  \begin{align*}
   2 \Kper(\axial,\axial) = \normx^2 \left ( P_{1\,0} (z,t) + \Phi_{\axial}\right ),
  \end{align*}
  with $P_{1\,0}(x,t)$ is a $C^n$ function of its variables and $\Phi_{\axial}$ is $C^n(\UinM)$
  and ${\bf o}(\normx)$. Computing the derivative we get
\[
  2\oaxial(\Kper(\axial,\axial))
  =\normx^2\left\{-2\frac{z}{\normr}\left(P_{1\,0}+\Phi_\axial\right)+
    \frac{\normx^2}{\normr}\partial_zP_{1\,0} +\oaxial(\Phi_\axial)\right\}
\]
after taking into account that $\oaxial(\normx)=- \frac{z}{\normr} \normx$,
$\oaxial(z)=\normx^2/\normr$ and $\oaxial(t)=0$.
By Lemma \ref{res:lemma_oaxial} the function $\Phi_{\axial}$
satisfies $\oaxial(\Phi_\axial)\in {\bf o}(\normx)$. It is clear that
all terms in brackets are
$C^{n-1}(\UinM\setminus \centresph_0)$  
and bounded near $\centresph_0$, so the same holds for
$\normx^{-2} \oaxial(\Kper(\eta,\eta))$ and we conclude that the third term in (\ref{eq:K2aa_g})
is $C^{n-1} (\UinM \setminus \centresph_0)$ and bounded.
Finally, the last term
reads
\[
\frac{1}{\normx^2\normr}\oaxial(\talpha z\trho^2)
=\frac{\qtheta_{1-}}{\normr^2} \frac{z}{\normr}
+\talpha \frac{\trho^2}{\normr^2}-2\UR\talpha \frac{z^2}{\normr^2},
\]
after using (\ref{eq:alpha_first}).
All three terms are $C^n(\UinM\setminus\centresph_0)$,
and since $q_{1-}\in O(\normr^2)$ the first term is, like  the rest, bounded near $\centresph_0$.
Summarizing, \eqref{eq:K2aa_g} implies that
$\Kperper^g(\axial,\axial)/\normx^2$ is $C^{n-1}(\UinM\setminus\centresph_0)$ and bounded near
the origin. The same holds for  $\Kperper^g(\axial,\axial)/\trho^2$ given the properties of $\UR$.

We consider now $\normr^2\Kperper^g(\partialr,\oaxial)$, for which we need
to analyse (\ref{eq:K2gi1}) for $\vort_1=\partialr$. Using $[\oaxial,\partialr]=0$
we obtain,
after simple rearranging,
\begin{align}
  \normr^2\Kperper^g(\oaxial, \partialr)
  =&\normr^2\Kperper(\oaxial,\partialr)
  + \normx^2 \left ( \Phizero_A + \Phizero_B \right )
  -2\oaxial\left(\talpha\normr^2\Kper(\oaxial,\partialr)\right),
  \label{eq:K2g_on}
\end{align}
where we have defined
\begin{align}
  &\Phizero_A \defi\UR\normr\left\{-\normr\partialr(\tX)
+\normr\talpha\oaxial(\partialr(\talpha))
  +\normr\partialr(\talpha)(3\oaxial(\talpha)-4\talpha Y^3)
   \right \}, \label{GammaA}\\
  &\Phizero_B \defi - 2 \frac{\normr^2}{\normx^2} \partialr (\talpha) \Kper(\oaxial,\oaxial). \label{GammaB}
\end{align}
The first term in \eqref{eq:K2g_on}
  is a $C^{n-1}(\UinM)$ function by virtue of Lemma \ref{res:q-q+}
  with $K=\Kperper$, $m=n-1$. In fact, it corresponds to the function $q$ in that lemma, so it  admits an expansion of the form
\begin{align}
\normr^2\Kperper(\oaxial,\partialr) =
  z\normx^2 P^{**}_{11}+\sum_{k=2}^{[(n-1)/2]+1}\normx^{2k}P^{**}_{0k}
+\normx^2\Phi^{**(n-1)}_{01}
+z^2\Phi^{** (n-1)}_{20}
\label{eq:expr_for_q2}
\end{align}
with all properties stated in the Proposition. For the second term we use an analogous procedure
as in Proposition \ref{prop:K_spher}. We may use
(\ref{eq:ntalpha}) replacing
 $\talpha$ by $\tX$ and using the corresponding
 $\tlemma=\Q^2\tX=\tXX-\talpha \qtheta_{1-}$, so that
\begin{align}
\normr\partialr(\tX)=&\frac{1}{\UR}\left(\frac{\partialr(\tXX-\talpha \qtheta_{1-})}{\normr}
-\frac{\tXX-\talpha \qtheta_{1-}}{\normr^2}\left(\normr\partialr(\UR)\frac{1}{\UR}+2 \right)\right)\nonumber\\
=&\frac{1}{\UR}\left(\frac{\partialr(\tXX)}{\normr}-\normr\partialr(\talpha)\frac{ \qtheta_{1-}}{\normr^2}
-\frac{\talpha\partialr( \qtheta_{1-})}{\normr}
   -\frac{\tXX-\talpha \qtheta_{1-}}{\normr^2}\left(\normr\partialr(\UR)\frac{1}{\UR}+2 \right)\right). \label{eq:nv}
   \end{align} 
On the other hand, we compute
\begin{align}
\normr\oaxial(\partialr(\talpha))=\normr\partialr(\oaxial(\talpha))=
\normr\partialr \left (\frac{q_{1-}}{\Q^2} \right )
  =\frac{1}{\UR}\left(\frac{\partialr(q_{1-})}{\normr}-\frac{q_{1-}}{\normr^2}\left(\normr\partialr(\UR)\frac{1}{\UR}+2\right)\right).
  \label{eq:iotanalpha}
\end{align}
All terms in \eqref{eq:nv} and \eqref{eq:iotanalpha} are bounded near the origin as a consequence  of \eqref{prop:tXX}, together with the facts
(we also use that ${\bf o} (\normx) \Longrightarrow o(\normr)$,  see Lemma
\ref{res:rhon-rn})
\begin{align*}
  &\qtheta_{1-} \in O(\normr^2), \qquad \partialr(\qtheta_{1-})\in {\bf o}(\normx) 
  &(\mbox{Lemma  \ref{res:q-q+} for } K = \Kper, m=n \geq 2, q_-=q_{1-}), \\
  &\talpha 
  \quad  \mbox{and} \quad \normr \partialr (\talpha)
    \mbox{ bounded }
  &\mbox{(Proposition \ref{prop:K_spher})}.
\end{align*}
Boundedness of the last term in brackets in \eqref{GammaA} is immediate. The property that all terms  in $\Phizero_A$ are $C^{n-2} (\UinM \setminus \centresph_0)$ is obvious.
Thus, we conclude that 
\begin{align*}
  \Phizero_A \in C^{n-2}(\UinM\setminus \centresph_0), \qquad
  \Phizero_A \in  O(\normr)
\end{align*}
so that, in particular this function extends continuosly to the centre
with the value zero.
Concerning $\Phizero_B$, point \emph{2.} of Lemma 5.2 applied to $K=\Kper$ anb $m=n\geq 2$,
together with the above properties of $\normr \partialr(\talpha)$, ensure $\Phizero_B$
is $C^{n-1}(U\setminus \centresph_0)$ and $O(\normr)$.
The function $\Phizero_2^{(n-2)}$ in the statement of the Proposition is simply
$\Phizero_2^{(n-2)} \defi \Phizero_A + \Phizero_B$ and obviously also satisfies 
$\Phizero_2^{(n-2)} \in C^{n-2} (\UinM \setminus \centresph_0) \cap C^0 (\UinM)$
and $O(\normr)$.

It only remains to show that the last term in
\eqref{eq:K2g_on}, which is clearly $C^{n-1}(\UinM\setminus\centresph_0)$, can
also be extended continuously to $\centresph_0$.
Recall first that the function
$q_1 := \normr^2 \Kper(\iota,\partialr)$ is (by Lemma \ref{res:q-q+}
applied to $K = \fpt$, $m=n$) $C^n(\UinM)$ and ${\bf o}(\normx)$.
Computing the derivative, and using $\oaxial(\talpha)=q_{1-}/\Q^2$, one finds
\begin{equation}
  \oaxial\left(\talpha\normr^2\Kper(\oaxial,\partialr)\right)=
  \UR^{-1} \frac{q_{1-}}{\normr^2} \qtheta_1
  +\talpha
  \oaxial( \qtheta_1)
  ,
  \label{eq:third}
\end{equation}
Since $q_{1-}\in O(\normr^2)$
the first term extends continuously to $\centresph_0$ where it vanishes.
For the second, we apply Lemma \ref{res:lemma_oaxial} to $q_1$ for $l=1$ 
(recall that $n\geq 2$) to conclude that $\oaxial(q_1)\in {\bf o}(\normx)$.
Since $\talpha$ is bounded near the origin, it follows that
$\talpha \oaxial(q_1)$ is ${\bf o}(\normx)$ and hence
extends continuously to $\centresph_0$ with the value
zero. Note that the form
\eqref{eq:q2theta}, given that $n\geq 2$, implies that $\qtheta^g_2$
is ${\bf o}(\normx)$. This completes the proof.
\fin
\end{proof}

\section{General stationary and axisymmetric perturbation scheme
on spherically symmetric backgrounds}
\label{sec:main}

In this section we combine the results in Section \ref{sec:OTactions} for
orthogonally transitive actions
and those in Section \ref{sec:axis_on_sph} involving spherically symmetric backgrounds
to construct a stationary and
axisymmetric perturbation scheme on spherically symmetric backgrounds.
We first show the existence of gauge vectors that render 
the first and second order perturbation tensors 
in the
standard forms assumed in the literature. Subsection
\ref{onUniqGauge}
  is devoted to
  discussing  uniqueness properties of these gauge vectors and the last subsection to
  studying 
   the gauge freedom left
    in those forms.
 
While in the previous section the spherical background was arbitrary, here we restrict ourselves to the static case, since this is what we shall need
  in \cite{PaperII}.  We start by making explicit the definition of static and spherically symmetric spacetime that we use and then impose the global assumptions on the background that will be needed.

\begin{definition}
\label{def:staticspher}
A (four-dimensional) spacetime $(\mmm,g)$ is static and spherically symmetric
if it admits an $SO(3)$ group of isometries
acting transitively on spacelike surfaces (which may degenerate to points), and 
a Killing vector $\stat$ which is timelike everywhere, commutes with the generators of $SO(3)$ and is orthogonal to the $SO(3)$ orbits.
\end{definition}

It is easy to check that such $\xi$ is necessarily hypersurface orthogonal, which justifies the name ``static'' in the definition.

Our global assumption is the following.

\vspace{3mm}

\noindent {\bf Assumption {\AsHone}:} $\mmm$ is diffeomorphic to
$\BB\times I$ where $I\subset\mathbb{R}$ is an open interval and
$\BB$ is a radially symmetric domain of $\mathbb{R}^3$,
which may or may not contain the origin, 
with the orbits of the Killing $\xi$ along the 
$I$ factor and $SO(3)$ acting in the standard way on $\BB$.
In addition, in the cartesian coordinates $\{x,y,z,t\}$ of
$\BB \times I$, the metric $g$ takes the form
\begin{align*}
g= - e^{\nu} dt^2 + \VR (x_i dx^i)^2 + \UR \delta_{ij} dx^i dx^j
\end{align*}
with $\nu,  \VR, \UR$ are $C^{n+1}$ functions of the coordinates
$x,y,z$ and radially symmetric.

\vspace{3mm}

Note that assumption {\AsHone} implies assumption {\AsSone}, so all the results in the
  previous section hold.
Observe also that the set
of fixed points of the $SO(3)$ action (the centre of symmetry) is
either empty or  $\centresph_0 := \{ 0_3\} \times I$.

The Lorentzian signature of $g$ implies that both $\UR$
and $\UR + \VR |x|^2$ are positive everywhere, so we may define a $C^{n+1}$ function on $M$ by $e^{\lambda} = \UR + \VR |x|^2$. We also introduce the non-negative function
$\Q \in C^{n+1}(M \setminus \centresph_0)$ defined by $\Q^2 = \UR |x|^2$. It is clear that
this function can be extended continuously to $\centresph_0$, where it vanishes.

From $\{\xc,\yc,\zc,t\}$ we may define
standard spherical coordinates $\{r,\theta,\phi\}$, see (\ref{def:cartesian_coord}), so that
$r=\normr$. The set $\{r,\theta,\phi,t\}$ is a coordinate system in
$M \setminus \axis$ with $r$ taking values in $(a,b)$ 
with $0 \leq a< b \leq + \infty$. Note that $\BB$ is a ball if and only if $\centresph_0 \neq \varnothing$, and  if and only if $a=0$.

The functions $\nu, \lambda, \Q$ are radially symmetric so, when
expressed in the spherical coordinates depend only on $r$. We write
$\nu(r), \lambda(r), \Q(r)$ (i.e. making explicit the argument $r$)
when we refer to this representation of the functions. We finally note that the metric $g$ on $M \setminus \axis$ in spherical coordinates takes the form
\begin{equation}
g   = - e^{\nu(r)} dt^2 + e^{\lam(r)} dr^2 + \Q^2(r) \left ( d\theta^2
+ \sin^2 \theta d \phi^2 \right ), \quad \quad 
\quad \xi = \partial_t . \label{metric}
\end{equation}

Consider the smooth vector fields $\{\partial_\xc,\partial_\yc,\partial_\zc,\partial_t\}$
and the vector field $\partialr=
\frac{1}{\normr} (\xc \partial_{\xc} + \yc \partial_{\yc} + \zc \partial_{\zc})=\partial_r$,
which is smooth outside $\centresph_0$.
The functions $\nu,\lambda$ satisfy
\[
e^\nu =  g(\partial_t,\partial_t) \quad \mbox{in} \quad \mmm, \quad \quad \quad
e^\lambda =  g(\partialr,\partialr) \quad \mbox{in} \quad \mmm\setminus \centresph_0, \quad \quad \quad
\]
and in fact can also be defined by these expressions in the respective domains.

Any of the Killing vectors  $\axial_a$, $a=1,2,3$,
of the $so(3)$ algebra together with the static Killing $\stat$
generate an Abelian $G_2$ group of isometries
which acts orthogonally transitively on timelike surfaces
outside the axis of rotation $\axis$.
Following Section \ref{sec:axis_on_sph} we choose
$\axial = \partial_{\phi}$, without loss of generality.
Observe that, under
our assumptions $\mmm \setminus \axis$ is a simply connected manifold.

Consider a perturbation scheme 
$(\mmm_\pertp, \gfam_\pertp,\{\psi_\pertp\})$ of class $C^{n+1}$ around this background
that inherits the 
orthogonally transitive stationary and axisymmetric 
action
generated by $\{ \xi, \eta \}$. By Proposition \ref{Block} there is a choice of gauge that preserves the differentiability (i.e. $C^n$ for $\fpt$ and
$C^{n-1}$ for $\spt$) such that
\begin{align}
&\fpt = {\fpt}_{\kvi\kvj}(x^\otp)dx^\kvi dx^\kvj+ {\fpt}_{\otm\otn}(x^\otp)dx^\otm dx^\otn,\label{prop:k1_before}\\
&\spt = {\spt}_{\kvi\kvj}(x^\otp)dx^\kvi dx^\kvj+ {\spt}_{\otm\otn}(x^\otp)dx^\otm dx^\otn,\label{prop:k2_before}
\end{align}
in the coordinates $\{ t,r, \theta,\phi\}$ where
the metric
is given by (\ref{metric}) and $\stat=\partial_t$, $\axial=\partial_\phi$.
The metric has the form  (\ref{form}) with 
$\{x^\kvi\}=\{t,\phi\}$, $\{x^\otm\}=\{r,\theta\}$. Observe
that the coordinates used in Section \ref{sec:spher_symm} correspond now
to $\{x^I\}=\{t,r\}$ and $\{x^A\}=\{\theta,\phi\}$.
We define for convenience the unit vector $n\defi -e^{-\lambda/2}\partialr$
outside the origin,
where it is of class $C^{n+1}$ by construction. In spherical coordinates we have
 $n= -e^{-\lambda/2}\partial_r$.

From now one we let
$(\mmm_\pertp, \gfam_\pertp,\{\psi_\pertp\})$ to denote the 
{\bf maximal} perturbation
scheme (of the given perturbation) where this holds.
Our aim is to show that the perturbation tensors $\Kper$ and $\Kperper$,
assumed to be $C^n$ and $C^{n-1}$ tensors respectively,
can be rendered in the forms found in the literature,
at the cost of (i) loosing their differentiability by one outside the origin (we refer to the discussion of this point in the Introduction) but
  \emph{keeping a crucial property of boundedness  at the origin.}
and (ii)
restricting the gauge freedom, of course.

\begin{proposition}\label{prop:pre_main}
Let $(\mmm,g)$ be a static and spherically symmetric background, with $g$ of class $C^{n+1}$,
with $n\geq 2$, satisfying assumption {\AsHone}. 
Let $\Kper$ and $\Kperper$ be first and second order perturbation tensors
of class $C^n$ and $C^{n+1}$ respectively,
satisfying \eqref{prop:k1_before} and \eqref{prop:k2_before}.
Then, there exists  gauge vectors $\sper$ and $\sperper$, that
  extend continuously to zero at $\centresph_0$, such that
    the gauge transformed tensors
$\fpt^\Psi$ and $\spt^\Psi$  are respectively $C^{n-1}$
and $C^{n-2}$ outside the origin and satisfy
\begin{align}
  &\fpt^\Psi = {\fpt^\Psi}_{\kvi\kvj}(x^\otp)dx^\kvi dx^\kvj+ {\fpt^\Psi}_{\otm\otn}(x^\otp)dx^\otm dx^\otn,\label{prop:k1_OT}\\
  &\spt^\Psi = {\spt^\Psi}_{\kvi\kvj}(x^\otp)dx^\kvi dx^\kvj+ {\spt^\Psi}_{\otm\otn}(x^\otp)dx^\otm dx^\otn,\label{prop:k2_OT}
\end{align}
on $\mmm\setminus\axis$ together with
\begin{align}
  {\fpt^\Psi}_{AB} dx^A dx^B = 4 \fok \Q^2 \gsph, 
  \quad \quad 
  {\spt^\Psi}_{AB} dx^A dx^B  = 4 \sok \Q^2 \gsph, 
  \label{proport}
\end{align}
where the functions $\fok$ and $\sok$, invariant under $\axial$ and $\stat$,
are defined $\mmm\setminus\centresph_0$ by
\begin{equation}
  4\fok=\frac{1}{\trho^2}\Kper^\Psi(\axial,\axial),\quad
  4 \sok=\frac{1}{\trho^2}\Kperper^\Psi(\axial,\axial)
  \label{def:k1ik}
\end{equation}
and satisfy $\fok\in C^{n}(\mmm \setminus \centresph_0)$ and
$\sok\in C^{n-1}(\mmm \setminus \centresph_0)$,
and are bounded near $\centresph_0$.
Moreover,
\begin{enumerate}
\item the 1-form $\fpt^\Psi(\stat,\cdot)$ equals $\fpt(\stat,\cdot)$ on $M \setminus \centresph_0$ and  thus extends to a $1$-form of class $C^n(\mmm)$.
\item $\fpt^\Psi(n,n)$ is $ C^n(\mmm\setminus\centresph_0)$ and bounded near $\centresph_0$.
\item the 1-form $\spt^\Psi(\stat,\cdot)$ is of class $C^{n-1}(\mmm\setminus\centresph_0)$ and bounded near $\centresph_0$. It is given by
  \begin{align}
    \spt^\Psi(\stat,\cdot)= \spt(\stat,\cdot)-2\talpha \lie_{\oaxial} \Kper(\stat,\cdot).
    \label{eq:sptpsi}
  \end{align}
  In particular
  \begin{align}
    \spt^\Psi(\stat,\axial)=\spt(\stat,\axial)-2\talpha \oaxial\left(\Kper(\stat,\axial)\right)=\spt(\stat,\axial)
    -2\talpha\left\{\normx^2\mathcal{P}_n+\oaxial(\Phi^{(n)})\right\},
    \label{eq:K2taxial}
  \end{align}
  where $\mathcal{P}_n\in C^{n-1}(\mmm\setminus\centresph_0)$ and bounded near $\centresph_0$
  and $\Phi^{(n)}$ is $C^{n}(\mmm)$ and $\Phi^{(n)}\in {\bf o}(\normx^n)$.
  
\item  $\spt^\Psi(n,n)$ is $ C^{n-1}(\mmm\setminus\centresph_0)$ and bounded near $\centresph_0$.
\item the functions 
  \begin{align}
    \qfonetheta=\normr^2\fpt^\Psi(\oaxial,\partialr) \quad \mbox{ outside } \centresph_0, \quad
    \qfonetheta=0 \mbox{ on } \centresph_0\label{def:qf1}\\
    \qftwotheta=\normr^2\spt^\Psi(\oaxial,\partialr) \quad \mbox{ outside } \centresph_0, \quad
    \qftwotheta=0 \mbox{ on } \centresph_0\label{def:qf2}
  \end{align}
  are  $C^{n-1}(\mmm\setminus\centresph_0)$ and $C^{n-2}(\mmm\setminus\centresph_0)$ respectively, and both
  $C^0(\mmm)$ and  ${\bf o} (\normx)$,
  and take the forms
  given by \eqref{eq:q1theta} and \eqref{eq:q2theta} respectively.
\end{enumerate}
\end{proposition}
\begin{proof}
  The proof is based on Propositions \ref{prop:K_spher} and
  \ref{prop:K2_spher}. The vectors $\stat, \axial, \oaxial$ are smooth
  on $\mmm\setminus\centresph_0$ (the first two actually everywhere). From the coordinate
  expressions $\stat = \partial_t, \axial = \partial_{\phi}$,
  $\oaxial=-\sin\theta\partial_\theta$, valid on a dense set of
  $\mmm \setminus \centresph_0$, it follows immediately that
  $\{\stat, \axial, \oaxial\}$ commute with each other on $\mmm \setminus \centresph_0$. It is also clear that   $\stat(\trho^2)=0$. Since, by
  assumption $\Kper$, $\Kperper$ satisfy
  $\lie_{\stat} \Kper = \lie_{\stat} \Kperper=0$, the Lie derivative along
  $\stat$ of equation (\ref{eq:for_alpha_prop}) yields
  $\trho^2\stat(\oaxial(\talpha))=0$ and therefore $\stat(\oaxial(\talpha))=0$
  outside the axis $\axis$. Using $[\stat,\oaxial]=0$,
  this is equivalent to $\oaxial(\stat(\talpha))=0$ on $\mmm\setminus\axis$.
  This ensures that the solution $\talpha$ of (\ref{eq:for_alpha_prop})  
  can be constructed such that it  satisfies 
  $\stat(\talpha)=0$ outside the axis. Since
  $\talpha \in C^n (\mmm\setminus\centresph_0)$, we actually have
  $\stat(\talpha)=0$ everywhere on $\mmm \setminus \centresph_0$. Similarly,
  the Lie derivative along $\stat$ of \eqref{eq:difKperper} shows that $\tX$ can be constructed so that it satisfies $\stat (\tX)=0$
  on $\mmm \setminus \centresph_0$. We assume these choices from now on.
  
  The block diagonal form
  of  (\ref{prop:k1_before})-(\ref{prop:k2_before}) is equivalent to 
  \begin{align}
    &\Kper(\oaxial,\axial)=0,\quad
      \Kper(\oaxial,\stat)=0,\quad
      \Kper(\axial,n)=0,\quad
      \Kper(n,\stat)=0
      \label{eq:K10000}, \\
    &\Kperper(\oaxial,\axial)=0,\quad
      \Kperper(\oaxial,\stat)=0,\quad
      \Kperper(\axial,n)=0,\quad
      \Kperper(n,\stat)=0
      \label{eq:K20000}
  \end{align}
  outside the axis. However, since $\Kper$ and $\Kperper$ are $C^n(\mmm)$
  and $C^{n-1}(\mmm)$ respectively,
  (\ref{eq:K10000}) and (\ref{eq:K20000}) also hold on
  $\mmm\setminus \centresph_0$. Applying  the
  gauge transformation discussed in
  Propositions \ref{prop:K_spher}, \ref{prop:K2_spher} and denoting the corresponding
  gauge transformed tensors as $\fpt^\Psi=\Kper^g$, $\spt^\Psi=\Kperper^g$, it follows directly from
  (\ref{eq:K1gix})-(\ref{eq:K1g12}) (applied to $\vort_{1}=\stat$, $\vort_{2}=n$) that
  \begin{align}
    \fpt^{\Psi}(\oaxial,\axial)=0,\quad
    \fpt^{\Psi}(\oaxial,\stat)=0,\quad
    \fpt^{\Psi}(\axial,n)=0,\quad
    \fpt^{\psi}(n,\stat)=0
    \label{eq:K10000_bis}.
  \end{align}
  Similarly, \eqref{eq:K2gix}-\eqref{eq:K2g12}
  yield, after using $[\iota,n]=0$,
  \begin{align}
    \spt^{\Psi}(\oaxial,\axial)=0,\quad
    \spt^{\Psi}(\oaxial,\stat)=0,\quad
    \spt^{\Psi}(\axial,n)=0,\quad
    \spt^{\Psi}(n,\stat)=0.
    \label{eq:K20000_bis}
  \end{align}
  This proves in particular that the block diagonal form claimed in
  (\ref{prop:k1_OT})-(\ref{prop:k2_OT}) holds.

  Define $\fok$, $\sok$ as in \eqref{def:k1ik}. By virtue of Propositions
  \ref{prop:K_spher} and \ref{prop:K2_spher} these functions
  are bounded near the center and, respectively, $C^{n}(\mmm\setminus \centresph_0)$ and
  $C^{n-1}(\mmm\setminus \centresph_0)$. Using spherical coordinates $x^A=\{\theta,\phi\}$ and noting that
  $\oaxial=-\sin\theta\partial_\theta=-\frac{\normx}{\normr}\partial_\theta$,
  on $\mmm \setminus \axis$ we have
  \begin{align*}
    {\fpt^\Psi}_{AB}dx^A dx^B=&\fpt^\Psi(\partial_\theta,\partial_\theta)d\theta^2
                              +\fpt^\Psi(\partial_\phi,\partial_\phi) d\phi^2=
                              \frac{\normr^2}{\normx^2}\fpt^\Psi(\oaxial,\oaxial)d\theta^2+\fpt^\Psi(\axial,\axial)d\phi^2\\
    =&
       \frac{\normr^2}{\normx^2}\fpt^\Psi(\axial,\axial)\left(d\theta^2+\sin^2\theta d\phi^2\right)=
       \frac{\Q^2}{\trho^2}\fpt^\Psi(\axial,\axial)\gsph,
  \end{align*}
  where the first equality is a consequence of $\fpt^\Psi(\oaxial,\axial)=0$, and the third
  follows from  (\ref{proport_cov_g})
  (and that $\normr^2/\normx^2=\Q^2/\trho^2$).
  A similar calculation is valid for $\spt^{\Psi}$. This establishes \eqref{proport}.
  
  We focus now on $\Kper^\Psi(\stat,\cdot)$ and
  $\Kperper^\Psi(\stat,\cdot)$. For the first we apply
  \eqref{eq:K1gXY} to $\stat=\partial_t$ and any  $Y$ to get
  \begin{align}
    \fpt^\Psi(\stat,Y)=\Kper(\stat,Y)-\stat(\talpha)\bm\oaxial(\vect{Y})
    =\Kper(\stat,Y), \label{fpt:statY}
  \end{align}
  after using  $\vect{\stat}=0$ and $\stat(\talpha)=0$.
  Thus, $\Kper^\Psi(\stat,\cdot)=\Kper(\stat,\cdot)$, and claim \emph{1.} follows.
  At second order, \eqref{eq:K2gXY} applied to $\stat$ and
  $\partial_{x^{\mu}}$ gives, after  
  taking into account
  $[\stat,\oaxial]=0$, $\stat(\tX)=0$, $\stat(\talpha)=0$,
  $\Kper(\oaxial,\stat)=0$ 
  and $M(\stat,\cdot)=2\Kper(\stat,\cdot)$ (by \eqref{fpt:statY}),
  \begin{align}
    \spt^\Psi(\stat,\partial_{x^{\mu}})
    =&\spt(\stat,\partial_{x^{\mu}})-2\talpha (\lie_{\oaxial} \Kper)(\stat,\partial_{x^{\mu}}).\label{eq:K2gtx}
  \end{align}
  The first term at the right is clearly $C^{n-1}(M)$. By Proposition
  \ref{prop:K_spher}, $\talpha$ is $C^{n-1}(\mmm\setminus\centresph_0)$ and bounded near $\centresph_0$,
  and thence also is the second term.
  In particular we have $\spt^\Psi(\stat,\axial)=\spt(\stat,\axial)-2\talpha \oaxial\left(\Kper(\stat,\axial)\right)$. By
  expression \eqref{eq:Katz}  in Lemma \ref{res:q-q+} for $K=\Kper$, $m=n$ and $u=t$ one has (we drop a superindex $t$ to simplify the notation)
  \begin{align*}
    \Kper(\stat,\axial) = \normx^2 \sum_{k=0}^{[n/2]-1} \normx^{2k} P_k(t,z)
    + \Phi^{(n)} := \normx^2 \Polz  + \Phi^{(n)}
  \end{align*}
  and hence
  \begin{align*}
    \spt^\Psi(\stat,\axial)
    =&\spt(\stat,\axial)-2\talpha\left\{-2Y^3\normx^2 \Polz +\normx^2\oaxial(\Polz)+\oaxial(\Phi^{(n)})\right\},
  \end{align*}
  where
  \begin{align*}
    \oaxial(\Polz)
        &=\sum_{k=0}^{[n/2]-1} \left\{-2k Y^3\normx^{2k}P\right\}+\frac{\normx^2}{\normr}\partial_z\Polz
  \end{align*}
  is clearly $C^{n-1}(\mmm\setminus\centresph_0)$ and bounded near $\centresph_0$.
  With the appropriate definition of ${\mathcal P}_n$,
  expression \eqref{eq:K2taxial} and all the properties listed in item \emph{3.} follow.
  For $\Kper^\Psi(n,n)$ and $\Kperper^{\Psi}(n,n)$, 
  equations \eqref{eq:K1g12} and \eqref{eq:K2g12} with $\vort_1=\vort_2=n$ lead to
  \begin{align}
    &\Kper^{\Psi}(n,n) =\Kper(n,n), \nonumber \\
    & \spt^\Psi(n,n)
    =\spt(n,n)-2\talpha(\lie_\oaxial\Kper)(n,n)+2(\normr n(\talpha))^2\frac{\trho^2}{\normr^2}-4 \normr n(\talpha)\frac{1}{\normr}\Kper(\oaxial,n)
    \label{equ:sptPSinn}
  \end{align}
  where for the second we used $[n,\oaxial]=0$.
  Lemma \ref{res:q-q+} ensures that $\Kper(n,n)$ is $C^{n} (\mmm\setminus \centresph_0)$ and bounded near
  $\centresph_0$, so point \emph{2.} holds. Concerning the second order,
  all terms in \eqref{equ:sptPSinn},
  are clearly $C^{n-1} (\mmm\setminus \centresph_0)$.
  Lemma \ref{res:q-q+} applied to $\Kper$, $\lie_\oaxial\Kper$, and $\Kperper$,
  and recalling $e^{-\lambda}$ is bounded near $\centresph_0$, ensures
  $\Kperper(n,n)$, $\lie_\oaxial\Kper(n,n)$ are bounded near $\centresph_0$,
  and that $\Kper(\oaxial,n)$ is ${\bf o}(\normx)$. Therefore $\normr^{-1}\Kper(\oaxial,n)$
  is bounded by Lemma \ref{res:rhon-rn}. Since  
  $\normr n(\talpha)$ (by Proposition \ref{prop:K_spher})
  and $\trho^2/\normr^2$ are also bounded, point \emph{4.} is proved.

  Finally, the functions $\normr^2\fpt^\Psi(\partialr,\oaxial)$ and $\normr^2\spt^\Psi(\partialr,\oaxial)$
  are  $C^{n-1}(\mmm\setminus\centresph_0)$ and $C^{n-2}(\mmm\setminus\centresph_0)$
  by virtue of Propositions \ref{prop:K_spher} and \ref{prop:K2_spher}
  respectively. Those propositions also
  provide the explicit forms of $\qfonetheta=\qtheta_1^g$ and $\qftwotheta=\qtheta_2^g$,
  and their behaviour near the axis and the origin as indicated.
  \fin
\end{proof}

\subsection{Canonical form of the perturbations}

By (\ref{prop:k1_OT})-(\ref{prop:k2_OT}) and \eqref{proport}, the perturbation tensors $\fpt^{\Psi}$, $\spt^{\Psi}$ are determined by five functions each. Two of them, $\fok$ and $\sok$ have already been defined in the Proposition.
Before entering into the main result of the paper we introduce the
functions that will  determine the remaining parts (four components respectively)
of $\fpt^\Psi$ and $\spt^\Psi$.

Let us stress first the fact that given any 1-form $\bm{X}$ satisfying
$\bm{X}(\partialr)=\bm{X}(\oaxial)=0$ we have
$\bm{X}(\ngamma)=0$  (see \eqref{eq:ngamma2}) and, as a consequence, the set of equations
$\{-x F=\bm{X}(\partial_y),y F=\bm{X}(\partial_x)\}$ is compatible and
defines a unique function  $F$  outside the axis $x=y=0$.
The 1-forms of the form $\bm{X}=K(\stat,\cdot)$ with $K$ any of the tensors $\Kper$, $\Kperper$, $\fpt^\Psi$, $\spt^\Psi$
satisfy that property.
In the following we use $\stat=\partial_t$.

Regarding $\fpt^\Psi$, and recalling point \emph{1.} of  Proposition \ref{prop:pre_main}, the expressions
\begin{align}
  \foh & \defi -\frac{1}{4} e^{-\nu} \fpt^\Psi(\partial_t, \partial_t),\label{def:foh}\\
  -  x \UR\opert & =\fpt^\Psi(\partial_t, \partial_{y} ) , \qquad
                   y \UR\opert=\fpt^\Psi(\partial_t, \partial_{x} ),\label{def:opert}
\end{align}
define, respectively, $\foh \in C^n(M)$ and $\opert : M \setminus \axis \longrightarrow \mathbb{R}$
Combining \eqref{def:opert} with 
$\axial=x\partial_y-y\partial_x$ yields
\begin{equation*}
  -\trho^2\opert  = \fpt^\Psi(\partial_t,\axial)
\end{equation*}
so that, in particular, $\opert$ is axially symmetric. Since the right-hand sides of \eqref{def:opert} are $C^{n}(M)$
we may apply Lemma \ref{Analy} in Appendix \ref{app:diff_origin} to conclude that $\omega$ extends to a  $C^{n-1}(M)$ function.

The third function is $\fom \in C^{n}(\mmm \setminus \centresph_0)$ defined by
\begin{align}
&\fom \defi\frac{1}{4}\left\{\fpt^\Psi{}^\alpha_\alpha+e^{-\nu}\fpt^\Psi(\partial_t, \partial_t) - 8\fok\right\}=\frac{1}{4}\left\{\fpt^\Psi{}^\alpha_\alpha-4\foh - 8\fok\right\}.\label{def:fom_0}
\end{align}
Finally, we define $\fof : \mmm\setminus\centresph_0 \longrightarrow \mathbb{R}$ as
\begin{equation}
\fof\defi e^{-\lambda}\frac{1}{2\sqrt\UR}\frac{1}{\normr^3}\Upsilon_1 +\beta_1,\label{def:fof_a}
\end{equation}
where $\beta_1$ is any radially symmetric $C^{n-1}(\mmm\setminus\centresph_0)$
function bounded near $\centresph_0$ 
and $\Upsilon_1$ any axially symmetric solution to the equation
\begin{equation}
\oaxial(\Upsilon_1)=\qfonetheta
\label{def:fof_b}
\end{equation}
satisfying the outcome of Corollary \ref{res:main_coro_0} for $m=n$.
We will show below that this Corollary does in fact apply.

As regards to $\spt^\Psi$ we analogously define $\{\soh,\som\}$ on $M \setminus \centresph_0$ and $\sow$ on $M \setminus \axis$ by
\begin{align}
\soh & \defi -\frac{1}{4} e^{-\nu} \left(\spt^\Psi(\partial_t, \partial_t)
-2\trho^2\opert^2\right),\label{def:h}\\
-  x \UR \sow & =\spt^\Psi(\partial_t,\partial_y), \qquad
y \UR \sow =\spt^\Psi(\partial_t, \partial_{x} ),\label{def:sow}\\
\som & \defi\frac{1}{4}\left\{\spt^\Psi{}^\alpha_\alpha+ e^{-\nu}\spt^\Psi(\partial_t,\partial_t) - 8 \sok\right\},\label{def:m_0}
\end{align}
while we define $\sof$ on $\mmm\setminus\centresph_0$ as
\begin{equation}
  \sof\defi e^{-\lambda}\frac{1}{2\sqrt\UR}\frac{1}{\normr^3}\Upsilon_2 +\beta_2,\label{def:f_a}
\end{equation}
where $\beta_2$ is any radially symmetric $C^{n-2}(\mmm\setminus\centresph_0)$
function bounded near $\centresph_0$ and $\Upsilon_2$ is any axially symmetric solution to the equation
\begin{equation}
  \oaxial(\Upsilon_2)=\qftwotheta
  \label{def:f_b}
\end{equation}
satisfying the outcome of Corollary \ref{res:main_coro_0} with $m=n-1$.
As before, we will show below that this Corollary may be applied.
We emphasize that, unlike $\foh$ and $\opert$, we cannot
guarantee that $\soh$ and $\sow$ can be extended differentiably to $M$.
This is because  $\spt^{\Psi}(\stat,\cdot)$ is not known to be
$C^{n-1}(M)$ due to the presence of the last term in \eqref{eq:sptpsi}.

For later use we observe that \eqref{def:sow} implies 
\begin{equation}
-\trho^2\sow=\spt^\Psi(\partial_t,\axial),
\label{eq:K2sow}
\end{equation}
and that 
the functions $\fom$ and $\som$  satisfy (on $\mmm\setminus \centresph_0$)
\begin{align}
&\fom = \frac{1}{4}\fpt^\Psi (n,n)=\frac{1}{4} e^{-\lambda} \fpt^\Psi(\partialr, \partialr),\label{def:fom}\\
&\som = \frac{1}{4}\spt^\Psi (n,n)=\frac{1}{4} e^{-\lambda} \spt^\Psi(\partialr, \partialr).\label{def:m}
\end{align}

Proposition  \ref{prop:pre_main} already determines the regularity restrictions on the  functions $\fok$ and $\sok$ implied by the differentiability of $\Kper$ and $\Kperper$.  The next theorem, 
which is the main result in this work,
combines Propositions \ref{Block} and \ref{prop:pre_main}
in order to, firstly, establish rigorously that first and second order perturbation tensors of finite differentiability and preserving the axial symmetry 
admit a gauge transformation that puts them in a canonical form and, secondly, to  provide detailed information on the differentiability and regularity properties of suitably defined function components.

\begin{theorem}[\textbf{Main Theorem}]
\label{theo:main}
Let $(\mmm,g)$ be a static and spherically symmetric background
satisfying assumption {\AsHone},
with $g$ of class $C^{n+1}$ with $n\geq 2$, given in spherical coordinates by \eqref{metric}. 
Let us be given
first and second order perturbation tensors
$\Kper$ and $\Kperper$ of class $C^n(\mmm)$ and $C^{n-1}(\mmm)$ respectively
satisfying \eqref{prop:k1_before} and \eqref{prop:k2_before},
where $\{x^i\} = \{ t, \phi\}$ and $\{x^{\otm}\} = \{r,\theta\}$ and
$\xi= \partial_t$, $\eta =\partial_\phi$.

Then, there exists gauge vectors $\sper$ and $\sperper$, that extend continuously to zero
  at $\centresph_0$, such that the gauge transformed tensors
$\fpt^\Psi$ and $\spt^\Psi$  are of class
$C^{n-1}(M \setminus \centresph_0)$ and $C^{n-2}(\mmm\setminus \centresph_0)$
respectively, and 
define the functions $\{\foh,\fom,\fok,\opert,\fof, \soh,\som,\sok,\sow,\sof\}$ as above
that satisfy the following properties:
\begin{itemize}
\item[(a.1)]  $\foh$ is $C^{n}(\mmm)$,
\item[(a.2)] $\opert$ is $C^{n-1}(\mmm)$, 
\item[(a.3)] the vector field $\opert \eta$ is $C^n(\mmm\setminus\centresph_0)$,
\item[(a.4)] $\fom$ and  $\fok$ are
  $C^{n}(\mmm\setminus\centresph_0)$ and bounded near $\centresph_0$,
\item[(a.5)] $\fof$  is $C^{n-1}(\mmm\setminus\centresph_0)$,
  bounded near $\centresph_0$, $C^{n}(S_r)$ on all spheres $S_r$, and moreover    $\partial_\theta\fof$ is
    $C^{n-1}$ outside the axis and extends continuously to $\axis\setminus\centresph_0$, where it vanishes, and $\partialr(\fof)$ and $\partial_t \fof$ are $C^{n-1}(S_r)$ on all spheres $S_r$,
\item[(b.1)] $\soh$ is $C^{n-1}(\mmm\setminus\centresph_0)$ and bounded near $\centresph_0$,
\item[(b.2)] $\sow$ is $C^{n-2}(\mmm\setminus\centresph_0)$ 
  and bounded near $\centresph_0$,
\item[(b.3)] The vector field $\sow \eta$ is $C^{n-1}(\mmm\setminus\centresph_0)$, 
\item[(b.4)] $\som$ and $\sok$ are $C^{n-1}(\mmm\setminus\centresph_0)$
  and bounded near $\centresph_0$,
\item[(b.5)] $\sof$  is $C^{n-2}(\mmm\setminus\centresph_0)$,
  bounded near $\centresph_0$,
  $C^{n-1}(S_r)$ on all spheres $S_r$
  and moreover $\partial_\theta \sof$ is $C^{n-2}$ outside the axis and extends continuously
  to $\axis\setminus\centresph_0$, where it vanishes,
  and $\partialr(\sof)$ and $\partial_t \sof$ are $C^{n-2}(S_r)$ on all spheres $S_r$.
\end{itemize}
In terms of these functions
$\fpt^\Psi$ and $\spt^\Psi$ take the following form on $\mmm\setminus \axis$ 
\begin{align}
  \fpt^\Psi =& -4 e^{\nu(r)} \foh(r, \theta) dt^2
               -2  \opert(r,\theta) \Q^2(r)\sin^2\theta dtd\phi
               + 4 e^{\lambda(r)} \fom(r, \theta) dr^2 \nonumber\\
             &+4   \fok(r, \theta)\Q^2(r)(d\theta^2+ \sin ^2 \theta  d\phi^2)
               + 4e^{\lambda(r)}\partial_\theta \fof(r,\theta)\Q(r) dr d\theta ,
               \label{res:fopert_tensor}
  \\
  \spt^\Psi =& \left(-4 e^{\nu(r)} \soh (r, \theta) + 2{\opert}^2(r, \theta) \Q^2(r) \sin ^2 \theta \right)dt^2\nonumber\\
             &-2 
               \sow(r,\theta) \Q^2(r) \sin^2\theta dt d\phi
               + 4 e^{\lambda(r)} \som(r, \theta) dr^2 \nonumber\\
             &+4  \sok(r, \theta) \Q^2(r) (d\theta^2+\sin ^2 \theta  d\phi^2)
               + 4e^{\lambda(r)}\partial_\theta \sof(r,\theta)\Q(r)dr d\theta.\label{res:sopert_tensor}
\end{align}
\end{theorem}
\begin{proof}
Expressions (\ref{res:fopert_tensor}) and  (\ref{res:sopert_tensor})
follow directly from Proposition \ref{prop:pre_main} and the definitions
\eqref{def:foh}-\eqref{def:sow} after taking into account
that $\normx^2 d\phi=\normr^2\sin^2\theta d\phi=  x dy - y dx$ and 
$\Q^2=\UR \normr^2$ and $\trho^2=\UR\normx^2$.

It  remains to analyse the differentiability and boundedness
near $\centresph_0$ of the various functions.
Points \emph{(a.1)} and \emph{(a.2)} have already been proved when $\foh$ and
$\opert$ were introduced in (\ref{def:foh}) and (\ref{def:opert}).
Point \emph{(a.3)} follows immediately from
\begin{align*}
\omega \bm{\eta} = - \left ( \fpt^\Psi(\partial_t, \cdot) + 4 e^{\nu} \foh dt \right )
\end{align*}
since the right-hand side is a $C^n(M)$ 1-form (by
point \emph{1.} in Proposition \ref{prop:pre_main}).
For point \emph{(a.4)}, the statement on $\fok$ has been established in Proposition \ref{prop:pre_main}.
By point \emph{2.} of that proposition the right-hand side of \eqref{def:fom},
and thus $\fom$, also is.

Establishing \emph{(a.5)} needs some additional work.
$\fof$ being defined  by \eqref{def:fof_a},
we directly have from \eqref{def:fof_b} that
\[
\partial_\theta \fof =-\frac{1}{\normx}e^{-\lambda}\frac{1}{2\sqrt{\UR}}\frac{1}{\normr^2} \qfonetheta
\]
outside the axis. Given that
$\qfonetheta$ is $C^{n-1}(\UinM\setminus\centresph_0)\cap C^0(\UinM)$ and ${\bf o}(\normx)$ (see  Proposition \ref{prop:pre_main}),
$\partial_\theta\fof \in C^{n-1}(\UinM\setminus\axis)$
and can be extended continuously to $\axis\setminus\centresph_0$, where it vanishes.
Regarding $\fof$ itself, we must analyse the solutions $\Upsilon_1$ of equation \eqref{def:fof_b}. By 
Proposition \ref{prop:pre_main}, $\qfonetheta$
is given by \eqref{eq:q1theta},  so the right-hand side of
\eqref{def:fof_b} matches the right-hand side
of \eqref{eq:for_gamma_general_zero}
with $b=\displaystyle{\min_{\V}\{2k+l\}=3}$, $c= \displaystyle{\min_{\V'}\{2k'+l'\}=2}$, and $m= n\geq 2$, so that
$$
3\leq b\leq c+n,
$$
and Corollary \ref{res:main_coro_0} ensures there exists a solution $\tlemma_0=
\Upsilon_1$
which is $C^{n-1}(\mmm\setminus\centresph_0)\cap C^0(\mmm)$, $C^n(S_r)$ on all $S_r$, and $O(\normr^3)$, and moreover $\partialr(\Upsilon_1)$ and
$\partial_t\Upsilon_1$ are $C^{n-1}(S_r)$ on all $S_r$.
Since $\UR\in C^{n+1}(\mmm)$ and nowhere zero, it follows that
  $\fof$ defined by \eqref{def:fof_a} has the properties listed in item \emph{(a.5)}.
 
Next, we consider the second order quantities. The definitions of 
$\soh$ and $\sow$ imply
\begin{align}
\Kperper^{\Psi}(\partial_t, \cdot ) = (- 4 e^{\nu} \soh+2\trho^2\opert) dt 
  - \sow\trho^2 d \phi.
  \label{Kperperppartialt}
\end{align}
Point \emph{3.} in Proposition \ref{prop:pre_main}
ensures that the left-hand side of \eqref{Kperperppartialt}
is $C^{n-1}(\mmm\setminus \centresph_0)$ and bounded near $\centresph_0$.
Using that $\opert$ is $C^{n-1}(\mmm)$, \emph{(b.1)} follows after contraction with $\partial_t$.

Expressions  \eqref{def:sow}
define an axially symmetric
function $\sow : \mmm\setminus \axis \longrightarrow \mathbb{R}$. Since the
right-hand sides of \eqref{def:sow}
are $C^{n-1}(\mmm\setminus\centresph_0)$, Lemma \ref{Analy}
(with $n$ replaced by $n-1$) implies that
$\sow \in C^{n-2}(M \setminus \centresph_0$). This proves the
differentiability claim in \emph{(b.2)}.  To show boundedness near the origin
we use an auxiliary function $\sowf_0$ defined in terms of the original second
order perturbation
$\Kperper$ (\emph{before} it has been gauge transformed into $\spt^{\Psi}$).
Specifically, $\sowf_0$ is defined by
\begin{equation}
  -  x \UR\sowf_0 =\Kperper(\partial_t,\partial_y), \quad
  y \UR\sowf_0 =\Kperper(\partial_t, \partial_{x} ).\label{def:sow_new}
  \end{equation}
Since $\Kperper$ is of class $C^{n-1}(\mmm)$, it follows from
Lemma \ref{Analy} that $\sowf_0$ is $C^{n-2}(\mmm)$.  It is also clear that
\eqref{def:sow_new}
imply
\begin{align}
  -\trho^2\sowf_0=\spt(\partial_t,\axial).
  \label{eta2swof}
\end{align}
The function $\sow$ can be written in terms of $\sowf_0$. 
Inserting \eqref{eta2swof} and \eqref{eq:K2taxial} into \eqref{eq:K2sow} yields
\[
  -\trho^2\sow= -\trho^2\sowf_0
  -2\talpha\left\{\normx^2\mathcal{P}_n+\oaxial(\Phi^{(n)})\right\}
\]
with $\mathcal{P}_n\in C^{n-1}(\mmm\setminus\centresph_0)$ and bounded near $\centresph_0$,
and $\Phi^{(n)}\in C^{n}(\mmm)$ and ${\bf o}(\normx^n)$.
Therefore, outside the axis we have
\[
  \sow= \sowf_0
  +2\talpha\left\{\frac{\normx^2}{\trho^2}\mathcal{P}_n+\frac{1}{\trho^2}\oaxial(\Phi^{(n)})\right\}.
\]
This proves that $\sow$ is bounded near the origin because
$\sowf_0\in C^{n-2}(\mmm)$, the functions
$\talpha$
and $\normx^2\mathcal{P}_n/\trho^2=\mathcal{P}_n/\UR$ are both 
bounded near $\centresph_0$ and, finally,
$\oaxial(\Phi^{(n)})/\trho^2$
vanishes at $\normx= 0$ by virtue of Lemma \ref{res:lemma_oaxial},
which ensures $\oaxial(\Phi^{(n)})\in {\bf o}(\normx^2)$ given that
$n\geq 2$. With this we have established point \emph{(b.2)}.

For \emph{(b.3)} we simply note that \eqref{Kperperppartialt}
  can be rewritten as
\begin{equation*}
\sow \bm{\eta} = - \left ( 
\spt^\Psi(\partial_t, \cdot) + (4 e^{\nu} \soh -2\trho^2\opert) dt\right )
\end{equation*}
and we have already seen that
the right-hand side is a $C^{n-1}$ one-form outside $\centresph_0$.

As for \emph{(b.4)}, Proposition \ref{prop:pre_main} already ensures
$\sok$ is $C^{n-1}(\mmm\setminus \centresph_0)$ and  bounded near $\centresph_0$,
and point \emph{4} of the same proposition states that the right-hand side of \eqref{def:m},
and thus $\som$, also is.

We finally focus on $\sof$, defined by \eqref{def:f_a}. Directly from 
  \eqref{def:f_b} we can write, outside the axis,
\[
\partial_\theta \sof =-\frac{1}{\normx}e^{-\lambda}\frac{1}{2\sqrt{\UR}}\frac{1}{\normr^2} \qftwotheta.
\]
Given that 
$\qftwotheta$ is $C^{n-2}(\UinM\setminus\centresph_0)\cap C^0(\UinM)$ and ${\bf o}(\normx)$, as stated
by Proposition \ref{prop:pre_main},
then $\partial_\theta\sof \in C^{n-2}(\UinM\setminus\axis)$
and can be extended continuously to $\axis\setminus\centresph_0$, where it vanishes.
Regarding $\sof$, analogously as for $\fof$,
we must analyse the solutions $\Upsilon_2$ of equation \eqref{def:f_b}.
This time
the form of the inhomogeneous term of the equation $\qftwotheta$,
given by \eqref{eq:q2theta}, renders \eqref{def:f_b} explicitly as
\[
\oaxial\left(\Upsilon_2+2\talpha\qtheta_1\right)=\qtheta_2+\normx^2\Phizero^{(n-2)}_2.
\]
Given the form of $\qtheta_2$ (c.f. the right hand side
of \eqref{eq:expr_for_q2}) and the properties
of $\Phizero^{(n-2)}_2$ in Proposition \ref{prop:K2_spher}, the right-hand side
of the equation  matches the right-hand side
of \eqref{eq:for_gamma_general_zero}
with $b=\displaystyle{\min_{\V}\{2k+l\}=3}$, $c=\displaystyle{\min_{\V'}\{2k'+l'\}=2}$, and $m= n-1\geq 1$, so that
$$
3\leq b \leq c +n-1
$$
and Corollary \ref{res:main_coro_0} shows that there exists a solution
$\tlemma_0=\Upsilon_2+2\talpha\qtheta_1$
which is $C^{n-2}(\mmm\setminus\centresph_0)\cap C^0(\mmm)$, $C^{n-1}(S_r)$ on all $S_r$, and $O(\normr^3)$, and moreover $\partialr(\Upsilon_2+2\talpha\qtheta_1)$ and
$\partial_t(\Upsilon_2+2\talpha\qtheta_1)$ are $C^{n-2}(S_r)$ on all $S_r$.
Since $\talpha \qtheta_1$ is $C^{n-1}(\mmm\setminus\centresph_0)\cap C^0(\mmm)$
and $\qtheta_1\in O(\normr^3)$ (the latter because of Lemma \ref{res:q-q+} with $K=\Kper$ and $\qtheta=\qtheta_1$),
we obtain $\Upsilon_2\in C^{n-2}(\mmm\setminus\centresph_0)\cap C^0(\mmm)$
and $O(\normr^3)$,
and $\partialr(\Upsilon_2)$ and $\partial_t\Upsilon_2$ are $C^{n-2}(S_r)$ on all $S_r$. Therefore
$\sof$ satisfies the properties listed in point \emph{(b.5)}.
\fin
\end{proof}

The previous theorem makes the hypotheses that $\Kper$ and $\Kperper$
are of the form \eqref{prop:k1_before} and \eqref{prop:k2_before}.
This assumption is well-justified in the present setup
because it holds automatically for any  perturbation scheme that inherits an orthogonally transitive, stationary and axially symmetric action, as proved in Proposition \ref{Block}. We state the corresponding result as a
corollary, where we also add the property that the gauge vectors
  $\sper$ and $\sperper$  whose existence has been proved are
  axially symmetric, tangent to $S_r$ and orthogonal to $\axial$,
  and extend continuously to zero at the centre (see  Propositions \ref{prop:K_spher} and \ref{prop:K2_spher}).

\begin{corollary}
\label{res:corollary_main}
  Let $(\mmm,g)$ be a static and spherically symmetric background
  satisfying assumption {\AsHone},
  with $g$ of class $C^{n+1}$, with $n\geq 2$, given in spherical coordinates by \eqref{metric}. Let us be given a $C^{n+1}$ maximal
  perturbation scheme $(\mmm_\pertp, \gfam_\pertp,\{\psi_\pertp\})$
  inheriting the orthogonally transitive stationary and axisymmetric action
  generated by $\{ \xi=\partial_t, \eta=\partial_\phi\}$.
  Then, the outcome of Theorem \ref{theo:main} holds.

  Moreover, the gauge vectors  $\sper$ and $\sperper$ whose existence has been granted and that transform
  any perturbation tensors $\Kper$, $\Kperper$ in this perturbation scheme
  into the form given by Theorem \ref{theo:main},  commute with $\axial$,
  are tangent to $S_r$ and orthogonal to $\axial$,
  and extend to zero at $\centresph_0$.\fin
\end{corollary}

This corollary ensures that there is a change of gauge
that takes any (orthogonally transitive) stationary and asymmetric perturbation (to second oder)
of a static and spherically symmetric configuration
to the usual forms (\ref{res:fopert_tensor})-(\ref{res:sopert_tensor}),
at the expense of loosing one derivative and deteriorating the regularity at the origin.
In the following we discuss 
the gauge properties compatible
with (\ref{res:fopert_tensor})-(\ref{res:sopert_tensor}).

\subsection{On uniqueness of gauges}
\label{onUniqGauge}
So far all the results have focused on the existence of the
gauge vectors $\sper$ and $\sperper$
that take from \eqref{prop:k1_before} and \eqref{prop:k2_before} to the canonical
form of Theorem \ref{theo:main}.
In this subsection we discuss their uniqueness properties for fixed transformed tensors $\Kper^\Psi$ and $\Kperper^\Psi$. We start with a general result  based on the symmetries of the background and of the
perturbation tensors.

\begin{lemma}
Consider a static and spherically symmetric background admitting no further local isometries and denote by $\stat$ its static
   and $\axial$ one of its axial Killing vectors. Let $\Kper$ and $\Kperper$ be $C^{2}$  perturbation tensors invariant under $\xi$ and $\eta$ on this background. If $\Kper'$ and $\Kperper'$ are obtained from $\Kper$ and $\Kperper$ by a gauge transformation defined by axially symmetric gauge vectors
 $\sper$ and $\sperper$ (i.e. that commute with $\axial$), then the most general such gauge vectors are given by 
\[
\tilde \sper =\sper + A_1\stat+B_2\axial, \quad \tilde \sperper =\sperper + A_2\stat + B_2\axial, \qquad A_1, A_2, B_1, B_2 \in \mathbb{R}.
\]
\end{lemma}
   
\begin{proof}
We start with the first order. From \eqref{gaugeper} and, by assumption,
\[
\Kper'=\Kper +\lie_{\sper}g, \qquad \Kper'=\Kper +\lie_{\tilde \sper}g.
\]
This leads to $\lie_{\sper-\tilde \sper}g=0$
and therefore $\tilde \sper-\sper$ is any Killing $\zeta_1$
of the background. By assumption $\zeta_1$ commutes with $\axial$
and given that the Killing algebra of the backgroud is
  $\mathbb{R} \oplus so(3)$ it must be $\zeta_1= A_1\stat+B_1\axial$
with $A_1,B_1\in\mathbb{R}$.
At second order we have, from \eqref{gaugeperper_bis}
using \eqref{gaugeper}, and by assumption
\begin{align*}
\Kperper'=\Kperper +\lie_{\sperper}g+\lie_{\sper}(\Kper^g+\Kper),\quad
\Kperper'=\Kperper +\lie_{\tilde \sperper}g+\lie_{\tilde \sper}(\Kper^g+\Kper).
\end{align*}
Therefore
$
0=\lie_{\sperper-\tilde \sperper}g- \lie_{\zeta_1}(\Kper'+\Kper)=\lie_{\sperper-\tilde \sperper}g,
$
where in the second equality we use that $\Kper'$ and $\Kper$ are invariant
under $\stat$ and $\axial$. The claim follows.\fin

\end{proof}

  This result combined with Lemma \ref{lemma:s_axial_commute} allows us to complement Corollary \ref{res:corollary_main} with the following uniqueness result.

\begin{remark}
  \label{rem:cor_main}In the setup of Corollary \ref{res:corollary_main}, if the
background admits no further local isometries and the perturbation scheme
$(\mmm_\pertp, \gfam_\pertp,\{\psi_\pertp\})$ is restricted so that
the inherited axial Killing vector  $\hat \axial_\pertp=d\psi_\pertp(\axial)$
is independent of the choice $\psi_{\pertp} \in
  \{ \psi_{\pertp} \}$, then
the gauge vectors $\sper$ and $\sperper$ are both unique up to
the addition of a Killing vector of the background that commutes with $\axial$.
We emphasize that this condition on the perturbation scheme is no restriction at all if $\gfam_\pertp$, $\pertp \neq 0$, admits only one axial symmetry.
\end{remark}

\subsection{Gauge freedom}
\label{sec:gauges}
In this subsection we investigate the gauge freedom to second order
that respects the \emph{form} of the first and second order perturbation tensors, as given in the main theorem. Our aim is to find the most general gauge transformation respecting this form under the additional condition
that the first order perturbation tensor takes a very special simple form (corresponding to a pure rotation). This is interesting for two reasons. Firstly,
because in the setup of \cite{PaperII}, this form of the first order perturbation tensor will in fact be a consequence of the gravitational field equations for a rotating fluid ball. And secondly because, as we will discuss later in more detail, our result will include as a particular case the most general gauge transformation that respects the structure of the canonical form of a general 
perturbation tensor \emph{at first order}.

We start with a general lemma concerning gauge vectors and symmetries. 

\begin{lemma}
\label{Lem1}
Let $(M,g)$ be a spacetime with Killing algebra $\Al$.
Let $\K$ be a symmetric two-covariant
tensor invariant under a subalgebra $\Al_0 \subset \Al$. Then the 
tensor $\K' := \K + \lie_{\s} g$, with $s \in \X(M)$ is 
also invariant under $\Al_0$
if and only if
\begin{align}
[ \zeta, \s ] \in \Al \quad \quad \forall \zeta \in \Al_0.
\label{commu1}
\end{align}
Assume now  $\Al = \mbox{span}\{\xi\} \oplus \mbox{span} \{ \eta_a \}$ with 
$\{ \eta_a \}$ a basis of an    $so(3)$ Killing algebra
and  $[\xi,\eta_a]=0$, $a=1,2,3$.
Let $\Al_0 = \mbox{span}\{ \xi, \eta\}$
with $\eta := \eta_1$ then there exists $\widehat{\eta} \in so(3)$ such that
$\s = \widehat{\s} + \widehat{\eta}$  with $\widehat{s}$ satisfying
\begin{align}
[\xi, \widehat{\s} \, ] & = C \eta + D \xi, 
\quad \quad \quad \quad  C,D \in \mathbb{R},
\label{comm1} \\
  [\eta, \widehat{\s} \, ] & =  A \eta + B  \xi, 
\quad \quad \quad \quad A,B \in \mathbb{R}.
\label{comm2}
\end{align}
If, in addition, the orbits of $\eta$ are closed, then $A = B =0$.
\end{lemma}

\begin{proof}
Taking the Lie derivative of $\K'$ along $\zeta \in \Al_0$ and imposing
invariance
\[
0 = \lie_{\zeta} \K' = \lie_{\zeta} \K + 
\lie_{\zeta} \lie_{\s} g =
\lie_{[\zeta,\s]} g + \lie_{\s} \lie_{\zeta} g =
\lie_{[\zeta,\s]} g 
\]
so $[\zeta,\s]$ is a Killing vector field and (\ref{commu1})
is established. 
For the second statement recall that the structure constants of 
$so(3)$ are 
\[
[\eta_a, \eta_b] = \epsilon_{ab}{}^c \eta_c \quad \quad
a,b,c = 1,2,3
\]
where $\epsilon_{ab}{}^{c}$ is the antisymmetric symbol. From (\ref{commu1})
we have $[\eta,\s] = 
A^a \eta_a + B \xi$
and $[\xi,\s] = C^a \eta_a + D \xi
$, $A^a, B,C^a, D \in \mathbb{R}$.
Commuting $\xi$ with $[\eta,\s]$ and using the Jacobi identity
\[
0 = [\xi, A^a \eta_a + B \xi]  =
[ \xi, [ \eta,\s]]  = [[\xi,\eta], \s] + [\eta,[\xi,\s]]  = [\eta,  C^a \eta_a + D \xi ]  = C^a \epsilon_{1a}{}^b \eta_b,
\]
which is equivalent to $C^2 = C^3 =0$.  Define
$\widehat{\eta} = \epsilon_{1}^{\phantom{1}c}{}_{d} A^d \eta_{c}$
and
\[
\widehat{\s} := \s - \widehat{\eta}.
\]
Setting $C:= C^1$, (\ref{comm1}) follows. Moreover,
\begin{align*}
[\eta, \widehat{\s} \,  ]  & =
[\eta, \s - \epsilon_{1}^{\phantom{1}c}{}_{d} A^d \eta_{c}] 
 = B \xi + A^a \eta_a  - 
\epsilon_{1}^{\phantom{1}c}{}_{d} A^d \epsilon_{1c}{}^a \eta_a 
 = B \xi + A^a \eta_a - \left ( \delta^a_d - \delta^a_1 \delta_{d1} \right ) A^d
\eta_a  \\
& =
B \xi + A^a \eta_a - A^a \eta_a + \eta_1 A^1  = B \xi + A^1 \eta_1,
\end{align*}
which is (\ref{comm2}) after setting $A:=A^1$.

Assume now that the orbits of $\eta$ are closed\footnote{Note that this condition is automatic when suitable global conditions, such as e.g. assumption {\AsHone}, are imposed.} (we stay away from the axis).
Since $\eta$, $\xi$ and $\widehat \s$
form a 3-dimensional algebra $\Al_3$ determined by (\ref{comm1})-(\ref{comm2})
plus an Abelian subalgebra for $\{\eta,\xi\}$, the argument
in Section 3 in \cite{alanaxialcomment}\footnote{As noted in the reference,
although the text refers to Killing vectors, the fact that the corresponding
local group of transformations are isometries is not used in the proof.}
shows that $\widehat\s$ must also commute with the cyclic $\eta$. Therefore (\ref{comm2})
must hold with $A=B=0$.\fin
\end{proof}

\begin{remarklem} The results (\ref{commu1}), (\ref{comm1}) 
and (\ref{comm2}) are purely local, so that they apply also
to spacetimes admitting no isometry group action (e.g. because
only a portions of the spacetimes is considered)
\end{remarklem}

Lemma \ref{Lem1} and the last argument of its proof
also implies the following 
Corollary.
\begin{corolemm}
  Let $(M,g)$ be a spacetime admitting a
  two-dimensional Abelian Killing algebra  $\Al$. Write
$\Al =
\mbox{span} \{ \xi,\eta \}$ and assume that $\eta$ has closed
orbits. Then $\lie_{s} g$ is invariant under $\Al$
if and only if
\begin{align}
[\eta, \s ] & =  0 , \label{comm1bis} \\
[\xi, \s \, ] & = C \eta + D \xi, 
\quad \quad \quad \quad  C,D \in \mathbb{R}.
\label{comm2bis}
\end{align}
\fin
\end{corolemm}
This result may have implications for studying gauge invariance
properties of perturbations in general stationary and axially symmetric
spacetimes.

From now on, however, we restrict ourselves to static
and spherically symmetric spacetimes as defined in Definition \ref{def:staticspher} and Assumption {\AsHone} with $n\geq 0$,
and assume we have selected a Killing vector $\eta$
in the $so(3)$ Killing algebra and have chosen 
spherical coordinates adapted to $\eta$, so that, away from the 
axis of $\eta$, the metric takes
the form (\ref{metric}). In particular, for the rest of the section
the functions $\nu,\lambda, \Q$ will be functions of only one variable $r$.
For simplicity we relax our convention of writing the argument explicitly, i.e. $\nu(r)$, etc. since  no confusion will arise.

Define $\widetilde{\s}$ by
\begin{equation}
\widehat{s} = C t \eta + D t \xi + \widetilde{\s},
\label{defstilde}
\end{equation}
and note that, from (\ref{comm1})-(\ref{comm2}) with $A=B=0$,
\begin{equation}
[\eta, \widetilde{\s}] = [\xi,\widetilde{\s}] = 0. \label{commutation}
\end{equation}
The difference tensor $\K'- \K = \lie_{\s} g$ has the form
\begin{align}
\K' - \K & = \lie_{\s} g = \lie_{\widehat{\s}} g 
= \lie_{C t \eta + D t \xi + \widetilde{\s}} g  =
C \left ( dt \otimes \bm{\eta} +  \bm{\eta} \otimes dt \right ) +  
D \left ( dt \otimes \bm{\xi} +  \bm{\xi} \otimes dt \right ) 
+ \lie_{\widetilde{\s}} g \nonumber \\
& = 
- 2 D e^{\nu} dt^2
+ 2 C \Q^2 \sin^2 \theta dt d \phi + \lie_{\widetilde{\s}} g.
\label{poundsstilde}
\end{align}
The form of $\widetilde{\s}$ can be restricted further
under additional assumptions on the form of $\K - \K'$, i.e.
$\lie_{\s} g$.

\begin{proposition}
\label{trans}
Let $(M,g)$ be a static and spherically symmetric spacetime
satisfying assumption {\AsHone} for $n\geq 0$,
with a selection of axial Killing vector field $\eta$ and 
$\s$ a vector field on $M$. Assume 
\begin{itemize}
\item[(i)] $\lie_{\s} g$ 
is invariant under $\mbox{span} \{ \xi,\eta\}$.
\item[(ii)] $\lie_{\s}g$  satisfies, in the coordinates 
$\{t,\phi,r,\theta \}$,
\[
\left ( \lie_{\s} g \right )_{\kvi \otm} =0
\]
where $\{ x^\kvi \} = \{ t,\phi\}$ and $\{ x^\otm \} = \{ r, \theta\}$.
\end{itemize}
(A) Then
\begin{align}
  \s = C t \eta + D t \xi + \overline{\s} + \zeta \quad \quad \quad
  C, D  \in \mathbb{R}
\label{sper2}
\end{align}
where $\overline{\s} = \overline{\s}^\otm (x^\otn) \partial_{x^\otm}$ and
$\zeta$ is any Killing vector field of $g$.

\noindent
(B) If, in addition to (i) and (ii), 
$\lie_{\s} g$ has $x^A = \{ \theta,\phi\}$ components of the form
\begin{align}
(\lie_{\s} g )_{AB} dx^A dx^B =  W(x^\otm) \left (  d\theta^2 + \sin^2 \theta d\phi^2 \right )
\label{thirdcase}
\end{align}
for some function $W(x^\otm)$, then there exists $\Y(r,\theta)$
and $\alpha(r)$ such that 
\begin{align}
 \overline{\s}  =&  2 \Y(r,\theta) \partial_r + 2 \alpha(r) \sin \theta
\partial_{\theta}, \quad \quad \mbox{and} \label{vector} \\
 \lie_{\s} g =&   
- 2 e^{\nu} \left ( \Y \nu_{,r}  + D \right ) dt^2 
+ 2 C \Q^2 \sin^2 \theta dt d \phi + 
4 e^{\lam} \left ( \Y_{,r} + \frac{1}{2} \Y \lam_{,r}\right ) dr^2 \nonumber \\
& + 4 \Q^2 \left ( \Y \frac{\Q_{,r}}{\Q} + \alpha(r) \cos \theta \right )
\left ( d \theta^2 + \sin^2\theta d\phi^2 \right )
+ 4 \Q e^{\lam} \left (\frac{\Y_{,\theta}}{\Q} + \Q e^{-\lam} \alpha_{,r}
\sin \theta \right ) d r d \theta.
\label{var3}
\end{align}
\noindent (C) If, in addition to (i) and (ii), $\Q_{,r}$ and $\nu_{,r}$
are not zero
on dense subsets  and $\lie_{\s} g$ has only component
in $\{t, \phi\}$, i.e. exists $Z(x^\otm)$ such that
\begin{align*}
  \lie_{\s} g = 2 Z(x^\otm) \Q^2 \sin^2  \theta dt d \phi,
\end{align*}
then
\begin{align}
  s &= C t \eta + \zeta, 
    \quad \quad \quad \mbox{and} \quad \quad
  \lie_{\s} g = 2 C \Q^2 \sin^2 \theta dt d \phi, \quad \quad i.e. \quad
  Z(x^\otm) = C. \label{var4}
  \end{align}
\end{proposition}
\begin{proof}
We have already shown, from Lemma \ref{Lem1} and (\ref{defstilde}), that
$\s = \widehat{\eta} + C t \eta + D t \xi + \widetilde{\s}$
for some $\widetilde{\s}$ satisfying (\ref{commutation}).
This imposes that the components
$\tilde{\s}^{\,\alpha}$ of $\tilde{\s}$ only depend on $x^\otn$.
Moreover, (\ref{poundsstilde}) implies $\lie_{\tilde \s} g (\partial_{x^\kvi},
\partial_{x^\otm}) =0$ due to assumption (ii). Explicitly
\[
0 = \tilde{\s}^\mu \partial_{\mu} g_{\kvi \otm} + g_{\kvi\kvj} \partial_\otm \tilde{\s}^\kvj
+ g_{\otm\otn} \partial_{\kvi} \tilde{\s}^\otn =  
g_{\kvi\kvj} \partial_\otm \tilde{\s}^\kvj
\quad \quad \Longrightarrow  \quad \quad \partial_\otm \tilde{\s}^\kvj =0,
\]
and therefore $\tilde{\s}^{\,\kvj}$ are in fact constants. Thus, there
exist constants $a, b$ such that
$\tilde{\s} = a \eta + b \xi + \overline{\s}$ with
$\overline{\s} = \overline{\s}^\otm (x^\otn) \partial_{x^\otm}$.
Defining $\zeta = \widehat{\eta} + a \eta + b \xi$, expression
(\ref{sper2}) follows.

In case (B), i.e. under the assumption (\ref{thirdcase}), we need to 
impose that  
$(\lie_{\overline{\s}} g )_{\phi\phi} -
\sin^2 \theta (\lie_{\overline{\s}} g )_{\theta\theta} =0$. Explicitly
\[
0 = \overline{\s}^\otm \partial_{\otm} g_{\phi\phi}
- \sin^2 \theta \left (
\overline{\s}^\otm \partial_\otm g_{\theta\theta}
+ 2 g_{\theta\theta} \partial_{\theta} \overline{\s}^{\theta} \right ) 
=
2 \Q^2 \sin^2 \theta \left ( \frac{\cos \theta }{\sin \theta}
\overline{\s}^{\theta}
- \partial_{\theta} \overline{\s}^{\theta} \right ).
\]
Integrating, there exists $\alpha(r)$ such that $\overline{\s}^{\theta}
=  2 \alpha(r) \sin \theta $. Letting $2 \Y(r,\theta):=\overline{\s}^{r}$
we obtain (\ref{vector}). Moreover, a direct computation gives
\begin{align*}
\lie_{\overline{\s}} g =
& - 2 \Y e^{\nu} \nu_{,r}  dt^2 + 4 e^{\lam} \left ( \Y_{,r} + \frac{1}{2} \Y \lam_{,r}\right ) dr^2
+ 4 \Q^2 \left ( \Y \frac{\Q_{,r}}{\Q} +  \alpha(r) \cos \theta \right )
\left ( d \theta^2 + \sin^2 \theta d\phi^2 \right ) \\
& + 4 \Q e^{\lam} \left ( \frac{\Y_{,\theta}}{\Q} + \Q e^{-\lam} \alpha_{,r}(r)
\sin \theta \right ) d r d \theta, 
\end{align*}
which inserted into (\ref{poundsstilde}) yields (\ref{var3}).

Case (C) is obviously a particular case of (B), so (\ref{vector})
and (\ref{var3}) hold. This already implies that $Z = C$. To show the rest
of (\ref{var4}) we use the fact that the
coefficients in (\ref{var3}) all vanish, except the $\{ t, \phi\}$. The
$dr^2$ component restricts $\Y(r,\theta)$ to be 
\begin{align*}
  \Y = e^{-\lambda/2} \Theta(\theta).
\end{align*}
for some function $\Theta(\theta)$.
The $d\theta^2$ component forces, after using that $\Q_{,r}$ is non-zero
on a dense set, that
\begin{align*}
  \alpha(r) = \frac{\alpha_0 \Q_{,r} e^{-\lambda/2}}{\Q}, \quad \quad
  \Y = -\alpha_0 e^{-\lambda/2} \cos \theta, \quad \quad \alpha_0 \in \mathbb{R}.
\end{align*}
The $dt^2$ then gives, using that $\nu_{,r}$ is non-zero almost everywhere,
$\alpha_0=0$ and $D=0$, and (\ref{var4}) follows.
\fin
\end{proof}

We have now the ingredients to analyse in detail the complete gauge freedom
that respects the form of the perturbation tensors achieved in our main
Theorem \ref{theo:main} assuming that the first order perturbation tensor
has the specific form given by \eqref{Kper} below.
As already said, this assumption turns out not to be restrictive
at all in the setting of \cite{PaperII} where
slowly rotating fluid balls with fixed central pressure are considered.

Even more,
the next proposition determines the most general gauge transformation that leaves invariant the form of a \emph{general} first order perturbation tensor in canonical form. Indeed, the gauge transformation law of second order perturbation tensors reduces to the first order one when $\Kper$ is identically zero  (see  (\ref{gaugeper})-(\ref{gaugeperper_bis})). Since the case $\Kper \equiv 0$ is covered by the proposition (by simply setting $\opert=0$)  it follows that the most general gauge transformation that leaves invariant the form of 
a first order perturbation $\Kper$ defined by the right-hand side of
(\ref{Kperper}) (with $\opert=0$) is given by (\ref{hg})-(\ref{Wtwog}) with corresponding gauge vector (\ref{gauge_second}).

\begin{proposition}[\textbf{Gauge freedom}]
  \label{prop:full_class_gauges}
  In the  setup of Theorem \ref{theo:main}
  we assume further that $\Q_{,r}(r)$ and $\nu_{,r}(r)$
  do not vanish identically on open sets and that the first
  and second order perturbation tensors take the form
\begin{align}
\Kper = &  -2  \opertbase(r,\theta)\Q^2(r) \sin^2 \theta dt d \phi.
\label{Kper} \\
\Kperper = &\left ( -4 e^{\nu(r)} h(r,\theta) +
2 \opertbase^2(r,\theta) \Q^2(r) \sin^2 \theta  \right ) dt^2
+ 4 e^{\lam(r)} m(r,\theta) dr^2 \nonumber \\
&+ 4  k(r,\theta)\Q^2(r) \left ( d \theta^2 + \sin^2 \theta
d \phi^2 \right )
+ 4  e^{\lam(r)} \partial_{\theta} f (r,\theta) \Q(r) dr d \theta\nonumber\\
&-2  \Wtwo (r,\theta)\Q^2(r)\sin^2\theta  dt d\phi.\label{Kperper}
\end{align}
Then a first order gauge vector $\sper$ 
preserves the form of $\Kper$ (i.e. there is $\opertbase^{g}$
such  that $\Kper^{g} := \Kper + \lie_{\sper} g$ is given by (\ref{Kper})
with $\opertbase \longrightarrow \opertbase^{g}$) if and only if, up to the
addition of a Killing vector
of the background,
\begin{align}
  \sper = C t \partial_{\phi}, \quad \quad C \in \mathbb{R},
  \quad \quad  \mbox{and then} \quad \quad
\opertbase^{g} = \opertbase - C.
\label{gauge_first}
\end{align}
For $\sper$ as in (\ref{gauge_first}), the second order
gauge vector $\sperper$ preserves the form of (\ref{Kperper})
if and only if 
\begin{align}
\sperper = A t \partial_t + B t \partial_{\phi}
+ 2 \Y(r,\theta) \partial_r + 2 \alpha(r) \sin \theta \partial_{\theta}
+ \zeta,
\quad A,B \in \mathbb{R}, \quad 
\zeta \mbox{ Killing vector of } g,
\label{gauge_second}
\end{align}
and $\Kperper^{g}$ takes the form 
(\ref{Kperper}) with the coefficients $h,m,k,f$ transformed to
\begin{align}
h^{g} &= h + \frac{1}{2} A  + \frac{1}{2} \Y  \nu_{,r}, \label{hg} \\
m^{g} &= m + \Y_{,r} + \frac{1}{2} \Y \lam_{,r}, \label{mg} \\
k^{g} &= k + \Y \frac{\Q_{,r}}{\Q} + \alpha(r) \cos \theta, \label{kg} \\
f^{g} &= f + \frac{\Y}{\Q} - \Q e^{-\lam} \alpha_{,r} \cos \theta + \beta(r), \label{fg}\\
\Wtwo {}^{g} &= \Wtwo -B, \label{Wtwog}
\end{align}
where the arbitrary function $\beta(r)$ arises because
$\Kperper^{g}$ only involves $\partial_{\theta} f^{g}$.
\end{proposition}

\begin{remarkpro}
  This proposition determines the most general gauge transformation that leaves the \emph{form} of the perturbation tensors (\ref{Kper}) and (\ref{Kperper}). One could analyse also which restrictions on the gauge vectors are required to ensure that the gauge transformed tensors  also satisfy the regularity and boundedness properties  of Theorem \ref{theo:main}. This is however not needed for the applications we have in mind.
\end{remarkpro}

\begin{proof}
The first order gauge
transformation (\ref{gaugeper}) imposes
 $\lie_{\sper} g = \Kper^g - \Kper$, so we may apply 
part (C) of Proposition \ref{trans} to $\s = \sper$ and (\ref{gauge_first})
follows immediately.
For the second part, we let $\sper :=  C t \partial_{\phi}$. The second order
  gauge transformation (\ref{gaugeperper_bis}) takes the form
  \begin{align*}
    \Kperper^{g} = \Kperper + \lie_{\sperper} g
    + \lie_{\sper} \left ( 2 \Kper^g - \lie_{\sper} g \right ).
  \end{align*}
  The tensor in parenthesis is $
    2 \Kper^g - \lie_{\sper} g
    = - 2 \Q^2 \left  ( 2 \opertbase - C \right ) \sin^2 \theta dt d\phi$
and its  Lie derivative along $Ct \partial_{\phi}$ is immediately computed to
  be
  \begin{align*}
    \lie_{C t \partial_{\phi}} \left (      2 \Kper^g - \lie_{\sper} g
\right ) = -2 \Q^2 C \left ( 2 \opertbase - C \right ) \sin^2 \theta dt^2.
\end{align*}
Thus the equation that $\sperper$ must satisfy is
\begin{align}
\lie_{\sperper} g  = & \Kperper^g - \Kperper + 2 \Q^2 C \left (
2 \opertbase - C \right ) \sin^2 \theta dt^2  \nonumber \\
= & \left ( -4 e^{\nu(r)} (h^g - h)  +
2 \Q^2 \sin^2 \theta \left ( 
\opertbase^g{}^2 - \opertbase^2 + C ( 2 \opertbase -C)
\right ) \right ) dt^2
+ 4 e^{\lam(r)} (m^g- m)  dr^2  \nonumber \\
&+ 4 \Q^2(r) (k^g- k) \left ( d \theta^2 + \sin^2 \theta
d \phi^2 \right )
+ 4 \Q(r) e^{\lam(r)} \partial_{\theta} (f^g - f) dr d \theta 
\nonumber \\
&-2 \Q^2(r) (\Wtwo^g - \Wtwo) \sin^2\theta  dt d\phi  \nonumber \\
= & -4 e^{\nu(r)} (h^g - h)  dt^2 + 4 e^{\lam(r)} (m^g- m)  dr^2  + 4 \Q^2(r) (k^g- k) \left ( d \theta^2 + \sin^2 \theta
d \phi^2 \right ) \nonumber \\
& + 4 \Q(r) e^{\lam(r)} \partial_{\theta} (f^g - f) dr d \theta -2 \Q^2(r) (\Wtwo^g - \Wtwo) \sin^2\theta  dt d\phi, \label{Diff}
\end{align} 
where in the third equality we inserted (\ref{Kperper})
and the corresponding expression for $\Kperper^g$
and the cancellations in the last equality follow from
$\opertbase^g{} = \opertbase - C$.
Thus, we may apply part (B) of Proposition (\ref{trans})
with $\s = \sperper$.
Expression (\ref{gauge_second}) follows directly from (\ref{sper2}) and 
(\ref{vector}), after renaming  $C \rightarrow A$ and $D \rightarrow B$
while
(\ref{hg})-(\ref{Wtwog}) are obtained by comparing (\ref{var3}) with
(\ref{Diff})  after the same redefinition for $C$ and $D$.\fin
\end{proof}

\section*{Acknowledgements}

We thank David Brizuela for valuable discussions.
M.M. acknowledges financial support under the projects
PGC2018-096038-B-I00
(Spanish Ministerio de Ciencia, Innovaci\'on y Universidades and FEDER)
and SA083P17 (JCyL).
B.R and R.V. acknowledge financial support under the projects
FIS2017-85076-P
(Spanish Ministerio de Ciencia, Innovaci\'on y Universidades and FEDER)
and IT-956-16 (Basque Government). B.R. was supported by the post-doctoral grant POS-2016-1-0075 (Basque Government). Many calculations have been performed using the free PSL version of the REDUCE computer algebra system.

\appendix

\section{Differentiability of radially symmetric functions}
\label{app:diff_origin}
In this Appendix we analyse the relationship between
the differentiability properties of a radially symmetric
function and of its trace (see below). We expect several of
these results to be known, but they do not seem to be
easily accessible in the literature, at least in the specific
form that we need. The starting point is, however, well-known (see Lemma 3.1 in \cite{Andreasson}). We include a proof for completeness.

We set $\Bp \subset \mathbb{R}^p$ to be the open ball of radius $\rad>0$ centered
at the origin.

\begin{lemma}
Let $q : \Bp \rightarrow \mathbb{R}$ be radially symmetric,
i.e. such that there exists $\trace{q}:  [0,\rad) \rightarrow \mathbb{R}$
(the trace of $q$) with $q(x)= \trace{q}(|x|) $. 
Then $q$ is a $C^n(\Bp)$ ($n\geq 0$) function
if and only if $\trace{q}$ is $C^n([0,\rad))$ (i.e.  up to the inner
    boundary) and all its odd derivatives up to order $n$ vanish at zero. Equivalently, if and only if 
    \begin{align}
      \trace{q}(r) = \Qpol_n(r^2)  + \Phi^{(n)}(r), \label{polinom}
      \end{align}
 where $\Qpol_n$ is a polynomial
of degree 
$[\frac{n}{2}]$ 
and $\Phi^{(n)}$ is $C^n([0,\rad))$ and satisfies
$\Phi^{(n)}(r) = o (r^n)$.
\label{origin}
\end{lemma}

\begin{proof}
Evaluating on the line $x = (\x,0,\cdots)$ we have $q(x_1,0,\cdots) = \trace{q} (\mbox{abs}(x_1))$. The left-hand side is a $C^n$ even function on $(-\rad,\rad)$, so all its odd derivatives (up to order $n$) vanish at zero. It follows from the equality that $\trace{q}$ is $C^n([0,\rad))$ up to boundary and that all odd derivatives (up to order $n$) vanish at zero. A Taylor expansion then gives (\ref{polinom}). The converse follows by a simple computation.
\end{proof}

\begin{remarklem}
This result can be applied to functions $q(x^i)$ which are $C^m$ in 
all variables and radially symmetric only in a subset of the coordinates.
A similar remark will apply for the remaining results in this Appendix.
\end{remarklem}

Lemma  \ref{origin} will be used in several ways. Our first application 
is the following statement.
\begin{lemma}
\label{Analy}
Consider the space $\mathbb{R}^2 \times \mathbb{R}^q$ ($q \geq 0$) coordinated
by $\{ x_1,x_2,\z\}$ and let $W$ be an open and connected neighbourhood of the axis $\axis := \{ x_1=x_2=0\}$,
radially symmetric in $\{x_1,x_2\}$.
Let $q : W \setminus \axis \rightarrow \mathbb{R}$ be radially symmetric
in $\{ x_1,x_2\}$ and assume that 
$x_1 q$ and $x_2 q$ extend to $C^n(W)$ functions, $n \geq 1$. Then $q$ extends to a $C^{n-1}(W)$ function.
\end{lemma}

\begin{remarklem}
The result on the differentiability of $q$ is sharp. Consider the function $|x|^{\alpha}$
with $1< \alpha< 2$. If is easy to check that 
$x_1 |x|^{\alpha}$ and $x_2 |x|^{\alpha}$ are $C^2(\mathbb{R}^2)$ while $|x|^{\alpha}$
  (which is $C^1(\mathbb{R}^2)$ in agreement with the lemma) is not
  $C^2(\mathbb{R}^2)$.
\end{remarklem}

\begin{proof}
Set $x= (x_1,x_2)$ and define $\trace{q}(r,\z)$ by $q(x,\z)= \trace{q}(|x|,\z)$.
Let $q_1 := x_1 q$, $q_2:= x_2 q$. By assumption $q_1$ is $C^n(W)$,
in particular when restricted to the line $\{ x_1 \geq 0, x_2=0\}$. Moreover,
\begin{align*}
q_1 (x_1,x_2=0,\z)= x_1 \trace{q}(x_1,\z).
\end{align*}
Since the left-hand side is $C^n$ up to and including the boundary
$x_1=0$ the same holds for $r \trace{q}(r,\z)$. 

On the other hand, the function $Q:= x_1 q_1 + x_2 q_2 = |x|^2 q$
is radially symmetric (in $\{x_1,x_2\}$) and $C^n(W)$. Moreover
$Q(x,\z)= r^2 \trace{q}(r,\z)|_{r=|x|}$. Given that $r^2\trace{q}(r,\z)$ is
$C^{n}$ up to boundary, Lemma \ref{origin} implies that
(we also use the fact that, by construction, $Q$ vanishes at the origin)
\begin{align}
r^2 \trace{q}  (r,\z) = r^2 \Qpol_n(r^2) +   \Phi^{(n)} (r,\z)
\label{eq:r2Phi}
\end{align}
where $\Qpol_n(u)$ is identically zero for $n=1$ and a polynomial in 
$u$ of degree $[\frac{n}{2}-1]$ when $n \geq 2$
(with coefficients  that are $C^n$ functions of $\z$)
and $\Phi^{(n)}(r,\z)$ is $C^n$ up to boundary 
and $o (r^n)$.
Now, $\frac{\Phi^{(n)}}{r} = r \trace{q} (r,\z) - r\Qpol_n$.
Since $r \trace{q}(r,\z)$ is $C^{n}$  it follows that
 $\frac{\Phi^{(n)}}{r}$ is $C^{n}$ and $o(r^{n-1})$
and therefore admits a Taylor expansion of the form
\[
\frac{\Phi^{(n)}(r,\z)}{r} =\alpha(\z)r^n+ \widehat\Phi^{(n)}(r,\z),
\]
where the remainder $\widehat\Phi^{(n)}(r,\z)$ is $C^n$ and $o(r^n)$.
Inserting this back into (\ref{eq:r2Phi}) yields
$\trace{q} (r,\z) = \Qpol_n(r^2)+\alpha(\z) r^{n-1} + \widehat\Phi^{(n)}/r$.
The last term is $C^{n-1}$ and $o(r^{n-1})$ by
item \emph{(iii)} of Lemma \ref{res:lemma_nkf} applied to the one-dimensional case,
\footnote{We stress that all the results in 
    Appendix \ref{app:iota} up to and including Lemma  \ref{res:lemma_nkf}
    are independent of Appendix \ref{app:diff_origin}. Proving the one-dimensional result here
    to avoid quoting a result from Appendix \ref{app:iota}
  would be redundant.}
and we
have found, after renaming $\widehat{\Phi}^{(n-1)}(r,\z)\defi \widehat\Phi^{(n)}/r$,
\begin{align*}
\trace{q} (r,\z) =  \Qpol_n(r^2)+\alpha(\z) r^{n-1} +\widehat{\Phi}^{(n-1)}(r,\z).
\end{align*}
If $n$ is odd, we can apply Lemma \ref{origin} 
to conclude that $q(x,\z)= \trace{q}(|x|,\z)$ is $C^{n-1}(W)$.
If $n$ is even, we still need to show that $\alpha(\z)=0$ before we can apply
that lemma. Let $k$ be defined by $n=2k$ (note that  $k \geq 1$).
The polynomial $\Qpol_n(r^2)$ plays no role
in the argument, so it can be set to zero without loss of generality
(we redefine $q(x) - \Qpol_n(|x|^2)$ as $q(x)$). 
Let us compute
\begin{align}
& \left ( \partial_\x \right )^{2k-1} (q_2) 
= \left ( \partial_\x \right )^{2k-1} (x_2 q) 
= x_2 \alpha(\z) \left ( \partial_\x \right )^{2k-1} |x|^{2k -1}
+ x_2 (\partial_\x)^{2k-1} (\widehat{\Phi}^{(2k-1)}(|x|,\z)).
\label{eq:q_2}
\end{align}
The term
$(\partial_\x)^{2k-1} (\widehat{\Phi}^{(2k-1)})$ is convergent as $|x|
\rightarrow 0$ (this claim is a particular case of Lemma \ref{res:lemma_di} in
  Appendix \ref{app:iota}). However, the first term does not converge at zero unless $\alpha(\z)=0$. This can be shown explicitly as follows. A simple induction argument based on $\partial_\x |x| = \x/|x|$ shows that, for all $l \leq  2k-1$, $k \geq 1$,  
  \begin{align*}
    \left ( \partial_\x \right )^l |x|^{2k-1} = \sum_{i=0}^l |x|^{2k-1-l-i} \x^i a_i,
    \qquad a_i\in\mathbb{R},\quad \quad a_{l} \neq 0.
  \end{align*}
  In particular, $(\partial_\x)^{2k-1} |x|^{2k-1}$ evaluated on the path
$x(r) = \beta r$ with
  $\beta= (\beta_1, \cdots, \beta_p)$ a constant unit vector yields
  $\lim_{r \rightarrow 0} \left ( \partial_\x \right )^{2k-1} |x|^{2k-1} = \sum_{i=0}^{2k-1} \beta_1^i a_i$, and the limit depends on the path.
Since, by assumption $q_2$ is $C^{2k}(W)$ (in particular $C^{2k-1}(W)$ as well) the only possibility
that the right hand side in \eqref{eq:q_2} has
a limit  as $|x| \rightarrow 0$ is  $\alpha(\z)=0$ and we may apply Lemma \ref{origin}
to $q(x,\z)$ to conclude the  proof. \fin
\end{proof}

\begin{lemma}
Let $f : \Bp \rightarrow \mathbb{R}$
be radially symmetric.
Assume that $f$ is $C^n(\Bp)$ and $o(|x|^n)$ with $n \geq 0$, then
the function
$\frac{x_i}{|x|} f(x)$ ($i=1,\cdots,p$)  
is also $C^n(\Bp)$ and $o(|x|^n)$.
\label{res:cart_product}
\end{lemma}

\begin{proof} 
  That $\frac{x_i}{|x|} f$ is $o(|x|^n)$ follows immediately
from the boundedness of $x_i/|x|$ near the origin.
Let $\trace{f}: [0,\rad) \rightarrow \mathbb{R}$ be the trace of $f$.
By construction $\trace{f}(r)$ is $C^n([0,\rad))$ and $o(r^n)$.
 
Let $\alpha = (\alpha_1, \alpha_2, \cdots ,\alpha_p)$ , $\alpha_j \in
\mathbb{N}$ and use the natural notion of (partial) order $\beta \leq \alpha$
iff all $\alpha_i \leq \beta_i$.
We use the multiindex notation
$\partial^{\alpha} := \frac{\partial^{\alpha_1}}{\partial (x_1){}^{\alpha_1}}
\cdots \frac{\partial^{\alpha_p}}{\partial (x_p){}^{\alpha_p}}$,
and $x^{\alpha} := x_1^{\alpha_1} \cdots x_p^{\alpha_p}$.
We write, as usual,  $|\alpha| = \sum_{j=1}^p \alpha_j$. 
Define the $i$-th 1-addition over $\alpha$
by $\alpha+1_i := (\alpha_1,\cdots, \alpha_i+1, \cdots, \alpha_m)$.
A simple
induction argument (with induction parameter $|\alpha|$, note that
$|\alpha+1_i|=|\alpha|+1$ for all $i$)
and the fact that, when acting
on any radially symmetric functions, 
$\partial_j f = \frac{x^j}{r} \left . \frac{d\trace{f}(r)}{dr} \right |_{r=|x|}$ implies that,
as long as $|\alpha| \leq n$,
\begin{align}
\partial^{\alpha} \left ( \frac{x^i}{|x|} f(x) \right )
= \sum_{a=0}^{|\alpha|} 
\frac{\trace{f}^{(a)}(|x|)}{|x|^{|\alpha|-a}}
\left ( \sum_{\beta \leq \alpha+1_i} b^{\alpha}_{a \, \beta} \frac{x^{\beta}}{|x|^{|\beta|}}
\right ),
\label{ders}
\end {align}
where $b^{\alpha}_{a \, \beta}$  are constants (which may vanish) and
$\trace{f}^{(a)}$ denotes the $a$-th derivative of $\trace{f}$. 
The properties of $\trace{f}$ imply that $\trace{f}^{(a)}$ is $C^{n-a}([0,\rad))$ and
$o(r^{n-a})$. Thus 
$\frac{\trace{f}^{(a)}(|x|)}{|x|^{|\alpha|-a}}$ is
$o(|x|^{n-|\alpha|})$, hence $o(1)$ since $|\alpha| \leq n$. Moreover
$x^{\beta}/|x|^{|\beta|}$ is bounded in a neighbourhood of the origin. We conclude
from (\ref{ders})  that 
$\partial^{\alpha} \left ( \frac{x^i}{|x|} f(x) \right )$ converges to 
zero at the origin, which proves
the lemma.\fin
\end{proof}

\section{Existence and regularity of a singular differential equation on spheres}
\label{app:iota}

The fundamental aim of this Appendix is to establish Lemma \ref{res:lemma_oaxial_equation} below.
In arbitrary dimension $p \geq 1$ 
let $\Bp \subset \mathbb{R}^p$ be an open
ball centered at 
the origin $\{ 0_p\}$.
Consider a natural number (possibly zero) $q$ and let $V^q \subset \mathbb{R}^q$
be an open connected neighbourhood 
in $\mathbb{R}^q$. Define also
$V:=\Bp\times V^q\subset \mathbb{R}^{p+q}$, where we make the usual
identification of $\mathbb{R}^p\times \mathbb{R}^q$ and
$\mathbb{R}^{p+q}$. Letting $\pi := \mathbb{R}^{p+q} \longrightarrow \mathbb{R}^p$
be the projection into the first factor, we set
$\normx := | \pi(x)|_{\mathbb{R}^p}$, where
$| \cdot |_{\mathbb{R}^p}$ is the standard euclidean norm in $\mathbb{R}^p$.
We introduce the axis $\axis\defi \{ 0_p \} \times V^q$ and note that, also,
$\axis=\{x\in V;\normx=0\}$.

The full set of coordinates in $V$ will be denoted by $x_\mu$, $\mu=1,\cdots, p+q$,
while  $a=1,\cdots, p$ coordinate the first $\mathbb{R}^{p}$ factor.
We will use a special name $u$ to refer to points in the first factor $\mathbb{R}^{p}$
and $w$ to refer to points in the second factor $\mathbb{R}^{q}$,
so that a point in $V$ reads $x=(u,w)$. In particular $(0_p,w)$ corresponds to a point in $\axis$.
Multiindices  (see proof of Lemma \ref{res:cart_product})
will include the following  special notation to refer to the part related to $\mathbb{R}^{p}$:
any multiindex
$\alpha=(\alpha_1,\cdots,\alpha_{p+q})$ will be separated as
$\alpha=\arho \oplus\alpha_w$ where $\arho=(\alpha_1,\cdots,\alpha_{p})$
is a multiindex in $\mathbb{R}^{p}$ and $\alpha_w$ in $\mathbb{R}^{q}$. Observe that
$|\alpha|=|\arho|+|\alpha_w|$.

Following the splitting
of $V$ we  extend
the usual little-o notation in terms of a limit on the axis $\axis$:
For any positive function $g$ defined on  $V \setminus \axis$ we set
(observe that the notation for ${\bf o}$ is in boldface)
$$f\in {\bf o}(g)\iff \lim_{\normx\to 0}f g^{-1}=0.$$

We need a simple result concerning Taylor expansions of $C^l$
functions with  ${\bf o}(\normrho{x}^l)$ behaviour.  Let us first recall that 
given a $C^l$ ($l \geq 1$) function $f:V\to \mathbb{R}$ the expansion in $u$ in the $\mathbb{R}^p$ factor around $\zerorho$ provides\footnote{As usual we set  $\arho!=\alpha_1!\cdots \alpha_{p}!$.}
\begin{equation}
  f(x)=\sum_{|\arho|\leq l}\frac{1}{\arho!}\partial^{\arho}f(\zerorho,w) x^{\arho}+\sum_{|\arho|= l}Q_{\arho}(x) x^{\arho},
  \label{TaylorExp}
\end{equation}
where the remainder
\[
Q_{\arho}(x)=\frac{|\arho|}{\arho!}\int_0^1(1-t)^{|\arho|-1}\partial^{\arho} f(tu,w) dt
-\frac{1}{\arho!}\partial^{\arho} f(\zerorho,w)
\]
is $C^0(V)$ and satisfies $Q_{\arho}(\zerorho,w)=0$.
\begin{lemma}[Taylor]
\label{lemma:taylor}
Let $f:V\to \mathbb{R}$ be $C^{l}$ and ${\bf o}(\normrho{x}^l)$
with $l \geq 1$. Then $\partial^{\alpha} f(a)=0$
for all $a\in\axis$ and $|\alpha|\leq l$. Consequently, in particular,
\[
f(x)=\sum_{|\arho|= l}R_{\arho}(x) x^{\arho},
\]
where
\[
R_{\arho}(x)=\frac{|\arho|}{\arho!}\int_0^1(1-t)^{|\arho|-1}\partial^{\arho} f(tu,w) dt
\]
is  $C^0(V)$ and $R_{\arho}(a)=0$ for all  $a\in\axis$.
\end{lemma}

\begin{proof}
  We start by proving that $\partial^{\arho} f(a)=0$ for all $a\in\axis$ and $|\arho|\leq l$. 
  Assume this conclusion is not true and let $s \leq l$ be the smallest number
  for which there exists $a_0 \in \axis$ and a multiindex
  $\arho{}_0$ with $|\arho{}_0| = s$ satisfying
  $\partial^{\arho{}_0} f (a_0)
  \neq 0$. Define the (not all zero) constants
  $C_{\arho} := \frac{1}{\alpha_u!} \partial^{\arho} f (a_0)$ for all $|\arho| = s$.
  Then, writing $a_0 = (\zerorho,w_0)$, 
  we have
  from (\ref{TaylorExp}) that at all points   $x = (u,w_0)$ 
    \begin{align*}
    \frac{f(x)}{\normrho{x}^s} =
    \sum_{|\arho| = s} C_{\arho} \frac{x^{\arho}}{\normrho{x}^s} +
    \sum_{s < |\arho|\leq l}\frac{1}{\arho!}\partial^{\arho}f(\zerorho,w_0)
    \frac{x^{\arho}}{\normrho{x}^s} + \sum_{|\arho|= l}Q_{\arho}(x) \frac{x^{\arho}}{\normrho{x}^s}.
  \end{align*}
  The left-hand side tends to zero when $\normrho{x} \rightarrow 0$ by assumption.
  The same is clearly true for the second and third terms of the right-hand side.
  Taking the limit along a path $x = (\beta r,w_0)$, $r \rightarrow 0$ with
  any $\beta \in \mathbb{R} ^{p}$ of unit norm, it must be that
  \begin{align*}
    \sum_{|\arho|=s} C_{\arho} \beta^{\arho} = 0, \quad \quad \forall \beta
    \in \mathbb{R}^{p}.
  \end{align*}
  This polynomial in $\beta$ can vanish identically only if all 
  $C_{\arho}$ vanish, which is a contradiction.
  Therefore, we have that $\partial^{\arho} f(\zerorho,w)=0$ for $0\leq |\arho|\leq l$, and for all $w$.
  This implies that $\partial^{\aw}\partial^{\arho} f(\zerorho,w)\equiv\partial^{\alpha}f(\zerorho,w)=0$ for $0\leq |\alpha|\leq l$
  with $\alpha=\arho\oplus\aw $, as claimed. \fin
\end{proof}

\begin{remarklem}
  \label{remark:taylor}
  Any $f$ as in Lemma \ref{lemma:taylor} is, in particular, $C^k$ and ${\bf o}(\normrho{x}^k)$ for $0\leq k \leq l$, so the following is also true
  \begin{equation}
    f(x)=\sum_{|\arho|= k}R_{\arho}(x) x^{\arho} \quad \mbox{ for any } 0\leq k \leq l.
    \label{res:remainder}
  \end{equation}
  Note also that $R_{\arho}$ with $|\arho|\leq l-1$ can be differentiated. In particular
  \begin{equation}
    \partial_\mu R_{\arho}(x)=\frac{|\arho|}{\arho!}\int_0^1(1-t)^{|\arho|-1}
    t^{\epsilon_\mu}
    \partial_\mu\partial^{\arho} f(tu,w) dt, \quad\epsilon_\mu = 1
    \mbox{ if $\mu \leq p$,} \quad 0 \mbox{ otherwise}, 
    \label{res:dRa}
  \end{equation}
  so that $\partial_\mu R_{\arho}(a)=0$ for all $a\in\axis$.
  Observe that since $f\in {\bf o}(\normrho{x}^l)$,
  (\ref{res:remainder}) implies
  \[
    \lim_{\normrho{x}\to 0}\sum_{|\arho|= k}\frac{R_{\arho}(x)}{\normrho{x}^{l-|\arho|}} \frac{x^{\arho}}{\normrho{x}^{|\arho|}}=0 \quad \mbox{ for any } 0\leq k \leq l.
  \]
  The application of the $\beta$-path method as in the previous proof
  implies that for any $\beta\in \mathbb{R}^{p}$
  \[
    \lim_{\normrho{x}\to 0}\sum_{|\arho|= k}\frac{R_{\arho}(x)}{\normrho{x}^{l-|\arho|}} \beta^{\arho}=0 \quad \mbox{ for any } 0\leq k \leq l,
  \]
  and therefore $R_{\arho}(x)\in {\bf o}(\normrho{x}^{l-|\arho|})$. In summary,
  \begin{align*}
    \partial_\mu  R_{\arho} (a)=0, \quad 
    R_{\arho}(x)\in {\bf o}(\normrho{x}^{l-|\arho|}) \quad \quad
    \forall \arho \mbox{ with } |\arho| \leq l -1.
  \end{align*}
\end{remarklem}

\begin{lemma}
\label{res:lemma_di}
Let $f:V\to \mathbb{R}$ be $C^{l}$  with $l\geq 1$
and let $\alpha=\arho\oplus\aw$ be any multiindex satisfying $0\leq |\alpha|\leq l$.
\begin{itemize}
\item $f\in{\bf o}(\normx^{l})$ $\implies$ 
$
\partial^{\alpha}f \in {\bf o}(\normx^{l-|\arho|}) \quad (\implies \partial^{\alpha}f \in {\bf o}(\normx^{l-|\alpha|}))
$
\end{itemize}
\end{lemma}

\begin{proof}
The proof is based  on the following
two facts, which we establish first. The first one is that 
for any
$f^{(n)}\in C^{n}(V)$ and ${\bf o}(\normx^n)$ with  $n\geq 1$ it is true that
$\partial_\mu f^{(n)} \in {\bf o}(\normx^{n-1})$ and the second is that when
$\mu> p$ the decay gets improved to
$\partial_\mu f^{(n)} \in {\bf o}(\normx^{n})$.
We show the first claim by applying (\ref{res:remainder}) in Remark \ref{remark:taylor} to $f^{(n)}$ and
$k = n-1$
  to write
\[f^{(n)}(x)=\sum_{|\arho|=n-1}R_{\arho}(x) x^{\arho}.\] 
Differentiating and dividing by $\normrho{x}^{n-1}$ yields
\begin{equation}
\frac{\partial_\mu f^{(n)}(x)}{\normrho{x}^{n-1}}=\sum_{|\arho|=n-1}\left(\partial_\mu R_{\arho}(x) \frac{x^{\arho}}{\normrho{x}^{n-1}}+
R_{\arho}(x) \partial_{\mu}x^{\arho}\frac{1}{\normrho{x}^{n-1}}\right).
\label{eq:df_o_p-1}
\end{equation}
We consider now the limit of this equation to any $a\in\axis$.
The limit of the first term in the right hand side vanishes as a consequence
  of $\partial_\mu R_{\arho}(a)=0$ (Remark \ref{remark:taylor})
  given  that $\frac{x^{\arho}}{\normrho{x}^{n-1}}$ are  bounded
  (because $|\arho|=n-1\geq 0$).
The limit of the second term vanishes because $R_{\arho}(x)\in {\bf o}(\normrho{x})$
(again by Remark  \ref{remark:taylor}),
and $\partial_\mu x^{\arho}/\normrho{x}^{n-2}$ are bounded, so that 
$R_{\arho}(x) \frac{\partial_i x^{\arho}}{\normrho{x}^{n-2}}\frac{1}{\normrho{x}}\in {\bf o}(\normrho{x}^0)$. Hence
\[
\lim_{\normrho{x}\to 0}\frac{\partial_\mu  f^{(n)}}{\normrho{x}^{n-1}}=0,
\]
i.e. $\partial_\mu f^{(n)}\in  {\bf o} (\normx^{n-1})$, as claimed. Obviously this function is also $C^{n-1}(V)$, which allows us to
repeat the process as long as $n-1\geq 1$.
We continue with the second claim, i.e.
that 
$\lim_{\normx\to 0}(\partial_\nu f^{(n)})/\normx^n =0$ for any $\nu> p$.
Let such $\nu$ be fixed and introduce the notation
$h_{(\nu)}$ for $ (\ldots, 0_{\nu-1},h,0_{\nu+1},\ldots)$. The
fundamental theorem of calculus gives, at any point 
$w_0\in V^q$,
\begin{equation}
f^{(n)}(u,w_0+h_{(\nu)})=f^{(n)}(u,w_0) +\int^h_0\partial_\nu f^{(n)}(u,w_0+s_{(\nu)})ds.
\label{eq:integral_f_z}
\end{equation}
If we define
\[
\mathcal{F}_\nu(h,u)\defi \frac{1}{h\normx^n}\int^{h}_{0}\partial_\nu f^{(n)}(u,w_0+s_{(\nu)})ds,
\]
equation (\ref{eq:integral_f_z}) divided by $h \normx^n$  taken to the limit $\normx\to 0$
leads to
\begin{align}
  &0=\lim_{\normx\to 0} \mathcal{F}_\nu(h,u),
\quad\mbox{ for each } \nu>p \quad\mbox{ and all }\quad h\neq 0.
\label{cond_A}
\end{align}
  Since $\partial_\nu f^{(n)}$ is continuous
  the limit $\lim_{\normx\to0}\mathcal{F}_\nu(h,u)$ (at fixed $h \neq 0$)
converges uniformly to zero on compact subsets of $V^q$.
Therefore we can take the limit $h\to 0$ of (\ref{cond_A}) and interchange the
limits by the Moore-Osgood theorem \cite{Moore-Osgood} to get
\begin{align}
&0=\lim_{h\to 0}\lim_{\normx\to 0}\mathcal{F}_\nu(h,u)
=\lim_{\normx\to 0}\lim_{h\to 0}\mathcal{F}_\nu(h,u)
   =\lim_{\normx\to 0}\frac{1}{\normx^p}\partial_\nu f^{(n)}(u,w_0),
   \quad \nu > p,
\end{align}
where in the last equality we simply used the definition
of $\mathcal{F_\nu}$.
Since the point $w_0$ in $V^q$ is arbitrary our second claim is established.

We may now return to the proof of lemma, which follows by simply
applying the derivative $\partial^{\alpha} f$
with $\alpha = (\alpha_1, \cdots, \alpha_{p+q})$ and $|\alpha| \leq l$ 
in the reverse order
$\partial_{x_{p+q}}^{\alpha_{p+q}} \cdots
\partial_{x_1}^{\alpha_1} f$. The derivatives in the first $p$ directions decrease the order to ${\bf o}(\normx^{l-|\arho|})$, while the remaining derivatives do not
change this order at all. This proves the first implication of the lemma.  The second implication is trivial because $|\alpha| \geq |\arho|$.\fin
\end{proof}

This lemma has an immediate application for products of functions. We restrict
to the one-dimensional case, since this is all  we need in the main text.

\begin{corollary}
\label{corollary:ufff}
Let $l\geq 1$ and  $f,g \in C^l([0,a)), a>0$. Assume that
$f(x)$ and $g(x)$ are $o(x^l)$.
Then the funcion $h(x)=f(x)g(x)/x$ extends as a $C^l([0,a))$
function to $x=0$.
\end{corollary}
\begin{proof}
It suffices to check that the limits $\lim_{x\to 0}(\partial_x)^{k}h(x)$ for $k=0,\ldots,l$ exist.
We have
\begin{equation*}
(\partial_x)^{k} h(x)= \sum^{k}_{p=0}\sum^{k-p}_{r=0}\frac{1}{x^{p+1}} a_{kp}
(\partial_x)^{k-p-r}f(x)(\partial_x)^{r}g(x)
\end{equation*}
where $a_{kp}$ are constants. Lemma \ref{res:lemma_di} applied 
to the one-dimensional interval $[0,a)$ ensures that $(\partial_x)^k f(x)$ is $C^{l-k}$
and $o(x^{l-k})$, and the same for $g$.
Therefore, each term in the  sum is
$o(x^{l-k+p+r+l-r-p-1})=o(x^{2l-k-1})$ and, in particular, $o(x^{l-1})$
  for $k=0, \cdots, l$.\fin
\end{proof}

We now introduce an additional structure in the $V^q\subset \mathbb{R}^q$ factor.
We split $\mathbb{R}^q = \mathbb{R}^{d_z}
\times \mathbb{R}^{d_t}$ (in particular $q = d_z + d_t$) and
assume that $V^q$ is of the form $V^q=Z^{d_z}\times T^{d_t}$ where
$Z^{d_z}\subset \mathbb{R}^{d_z}$
and $T^{d_t}\subset \mathbb{R}^{d_t}$ are open and connected.
Furthermore, we
require that $Z^{d_z}$ contains the origin of $\mathbb{R}^{d_z}$.
Thus, we have now the splitting $\mathbb{R}^{p+q} = \mathbb{R}^p \times
\mathbb{R}^{d_z} \times \mathbb{R}^{d_t}$ and we want to concentrate on
the first two factors, so we define $\mathbb{R}^d =
\mathbb{R}^p \times \mathbb{R}^{d_z}$ (again with the usual identification).
Cartesian coordinates in $\mathbb{R}^{d}$ are denoted by
$x_i$, $i=1,\cdots, d$. As before we introduce the projection $\pi_2$
of $\mathbb{R}^{p+q} = \mathbb{R}^d \times \mathbb{R}^{d_t}$ into the first
$\mathbb{R}^d$ factor and define the seminorm $\normr = |\pi_2(x)|_{\mathbb{R}^d}$.
The set $\centresph_0:=\{x\in V; \normr=0\}$ is precisely
$\{0_d\} \times T^{d_t}$, where $\{0_d\}$ stands for the origin
of $\mathbb{R}^d$. It is immediate to check that
$\centresph_0\subset \axis$.

We also introduce the usual little-o notation associated to $\normr$
by defining (observe that $o$ is not in boldface now)
$$\qquad f\in o(|x|^l)\iff \lim_{\normr\to 0}\frac{f}{|x|^{l}}=0.$$

When $d_z=0$ we recover the previous setup. Therefore,
Lemma \ref{res:lemma_di} implies in particular that for $l\geq 1$
(irrespectively of the value of $d_z$)
$$f\in  C^{l}(V) \mbox{ and } o(\normr^{l})\implies \partial^{\alpha}f \in o(\normr^{l-|\alpha|}).
$$

Let us continue by providing two easy but convenient auxiliary results.
\begin{lemma}
  \label{res:rhon-rn}
Let  $f: V \rightarrow \mathbb{R}$.
If $f\in {\bf o}(\normx^l)$
then $f\in o(\normr^l)$.
\end{lemma}
\begin{proof}
Since $\normr^l\geq \normx^l$ for $l\geq 0$ then $|f| \normr^{-l}\leq |f| \normx^{-l}$,
so $\lim_{\normr\to 0}|f| \normr^{-l}\leq \lim_{\normr\to 0}|f| \normx^{-l}$.
But $\normr \rightarrow 0$ implies $\normx \rightarrow 0$, so the limit is zero
by the assumption $f\in {\bf o}(\normx^l)$.
\fin
\end{proof}

\begin{lemma}
\label{res:lemmaf8}
Let  $f: V \rightarrow \mathbb{R}$ and $k\geq 0.$ Then
  $f\in {\bf o}(\normx^l) \implies f/\normr^k\in {\bf o}(\normx^{l-k})$.
\end{lemma}
\begin{proof}
  The result follow directly
    from
  $\frac{1}{\normx^{l-k}}\frac{|f|}{\normr^k}=\frac{|f|}{\normx^{l}}\frac{\normx^k}{\normr^k}\leq \frac{|f|}{\normx^l}$, taking the limit
    $\normx\to 0$.
  \fin
\end{proof}

Let us introduce now the vector $
  \partialr\defi \frac{1}{\normr}\left(x^i\partial_{i}\right),
  $
  which is clearly smooth in $V\setminus \centresph_0$ and bounded near $\centresph_0$. Define also
  \begin{align*}
    \partialr^i\defi\partialr(x^i)= \frac{x^i}{\normr}=\partial^i \normr.
  \end{align*}
  We note the  well-known fact that
$\normr^{|\alpha|}\partial^\alpha \partialr^i$ for $|\alpha|\geq 0$,
as well as $\normr^{|\alpha|+1}\partial^\alpha (\frac{1}{\normr})$, are bounded.
We will use $v^{(m)}(f)$ for the derivation of $f$ $m$ times
along a vector field $v$. Since $\partialr$ does not exist at $\centresph_0$ we define 
$\partialr(f): V \rightarrow \mathbb{R}$ by setting
$\partialr(f) (\centresph_0)=0$, and similarly for higher powers $\partialr^{(m)}(f)$.

\begin{lemma}
\label{res:lemma_nkf}
Let $f:V\to \mathbb{R}$ be $C^{l}$, $l\geq 1$,
and assume $f\in {\bf o}(\normx^{l})$. Then, for any $1 \leq k \leq l$
\begin{itemize}
\item [(i)]
$ \partialr^{(k)}(f)\in {\bf o}(\normx^{l-k}) \mbox{ and } C^{l-k}(V).$
\item [(ii)] $\normr \partialr(f)\in {\bf o}(\normx^{l}).$
\item [(iii)] $ \frac{1}{\normr^k}f\in {\bf o}(\normx^{l-k}) \mbox{ and } C^{l-k}(V)$.
\end{itemize}
\end{lemma}
\begin{proof}
We start by establishing the following fact: 
For $n\geq 1$ and  $f^{(n)}\in C^{n}(V)$ satisfying
$f^{(n)} \in {\bf o}(\normx^{n})$, it holds 
\begin{align}
f^{(n-1)}:=\partialr(f^{(n)}) \quad \mbox{ is }  C^{n-1}(V) \mbox{ and }  {\bf o}(\normx^{n-1}).
\label{claim}
\end{align}
By the comment before the lemma, this claim actually involves
the function $g: V \rightarrow \mathbb{R}$ defined by
$g(\centresph_0) =0$ and
$g\defi \partialr(f^{(n)})=\partialr^i\partial_if^{(n)}$ outside $\centresph_0$.
Clearly $g\in C^{n-1}(V\setminus \centresph_0)$. Moreover,  given that $\partial_if^{(n)}\in {\bf o}(\normx^{n-1})$ (by Lemma \ref{res:lemma_di}),
and $\partialr^i$ are bounded everywhere, we have
$\partialr^i\partial_if^{(n)}\in {\bf o}(\normx^{n-1})$, which in particular means that  
$g\in C^0(V)$ and $g\in {\bf o}(\normx^{n-1})$.
Furthermore, the boundedness near the axis of
  $x^a/\normx$ for $1 \leq a \leq p$ and $x^j$ for $p <j\leq d$ combined with
  Lemma \ref{res:lemma_di} imply $ x^a\partial_af^{(n)}\in {\bf o}(\normx^n)$
and $x^j\partial_jf^{(n)}\in {\bf o}(\normx^n)$, and therefore
\[
\normr\partialr(f^{(n)})= \normr \partialr^a\partial_af^{(n)}+\normr \partialr^j\partial_jf^{(n)}=x^a\partial_af^{(n)}+x^j\partial_jf^{(n)}\in {\bf o}(\normx^n).
\]
This already proves  point (ii) in the lemma.

It only remains  to check that $g$ is $C^{n-1}(V)$.
Take $\alpha$ such that  $0\leq |\alpha|\leq n-1$ and compute
\[
\partial^\alpha \left ( \partialr(f^{(n)}) \right )
=\sum_{\beta\leq\alpha}b_\beta \partial^\beta \partialr^i \partial^{\alpha-\beta}\partial_i f^{(n)},
\]
where $b_{\beta}\in \mathbb{R}$.
The terms in the sum can be written as
$\normr^{|\beta|}\partial^\beta\partialr^i\frac{1}{\normr^{|\beta|}}\partial^{\alpha-\beta}\partial_i f^{(n)}$,
and given that $\normr^{|\beta|}\partial^\beta\partialr^i$ are bounded
and that $\partial^{\alpha-\beta}\partial_i f^{(n)}\in {\bf o}(\normx^{n-1-|\alpha|+|\beta|})$
by Lemma \ref{res:lemma_di} (note that $|\alpha|-|\beta|= |\alpha-\beta|$ because $\beta\leq \alpha$),
each term belongs to ${\bf o}(\normx^{n-1-|\alpha|+|\beta|-|\beta|})={\bf o}(\normx^{n-1-|\alpha|})$.
Therefore $\partial^\alpha \left ( \partialr(f^{(n)}) \right ) \in {\bf o}(\normx^{n-1-|\alpha|})$
and thus  $\partial^\alpha \left ( \partialr(f^{(n)}) \right ) \in o(\normr^{n-1-|\alpha|})$
by Lemma \ref{res:rhon-rn}. 
Since $|\alpha|\leq n-1$ the limits
of $\partial^\alpha \left ( \partialr(f^{(n)}) \right )$
vanish at $\normrho{x}=0$, and thus at $\normr=0$ in particular. 
Define the functions
\[
g^\alpha \defi \left\{
\begin{array}{ll}
        \partial^\alpha \left ( \partialr(f^{(n)}) \right ) & x\neq \centresph_0\\
        0& x = \centresph_0
\end{array}
\right.,
\]
with $g^0=g$, which are $C^0(V)$ by construction, and ${\bf o}(\normx^{n-1-|\alpha|})$
(and thus $ o(\normr^{n-1-|\alpha|})$).

It remains to verify
that the differentials
(to order $n-1$) of $g$ on $\centresph_0$
  exist and vanish
  for all $c\in\centresph_0$. We do that at once 
  by showing that $D g^{\beta}_c$ (namely,
  the differential of $g^{\beta}$  at $c$) vanishes 
for any $0\leq |\beta|\leq n-2$.
We compute
\begin{align*}
&\lim_{x\to c}\frac{|g^\beta(x)-g^\beta(c)|}{\normu{x}}
=\lim_{x\to c}\frac{|g^\beta(x)|}{\normu{x}}\leq \lim_{x\to c}\frac{|g^\beta(x)|}{\norm{x}}=0, \mbox{ for } 0\leq |\beta|\leq n-2
\end{align*}
where in the first equality we insert
$g^\beta(c)=0$ (since $c \in \centresph_0$), the inequality
follows from $\normu{x}\geq \norm{x}$,
and in the final equality we use that the limit $x\to c$ is equivalent to $\norm{x}\to 0$,
and that  $g^\beta\in o(\normr^{n-1-|\beta|})$,
so that $g^\beta/\normr\in o(\normr^{n-2-|\beta|})$. By definition of differential, this limit shows that $Dg^\beta_c$ exists and vanishes
for $0\leq |\beta|\leq n-2$ and all $c\in\centresph_0$.
For $|\beta|=0$ this means that $g$ is differentiable at $\centresph_0$ with vanishing differential, which establishes
that $g\in C^1(V)$. Iterating, and noting that $|\beta| \leq n-2$ means
that we may take up to $n-1$ derivatives of $g$, it follows  that
$g \in C^{n-1} (V)$, and the claim (\ref{claim}) is verified.

We now apply this result iteratively to the functions
$f^{(l)} := f$ and
$f^{(l-s)}\defi\partialr^{(s)}(f^{(l)})$ for $s=0,\cdots,l$.
By hypothesis, $f^{(l)}$ is $C^l(V)$ and ${\bf o}(\normrho{x}^{l})$,
so 
\begin{align*}
  f^{(n-1)}=\partialr(f^{(n)}) \quad \mbox{ is } \quad C^{n-1}(V) \quad \mbox{ and } \quad {\bf o}(\normx^{n-1})
\end{align*}
for $1\leq n\leq l $. Given that $1\leq l$ and $1\leq k\leq l$ by assumption,
we have $1\leq l-k+1\leq l$, and we can take $n = l-k+1$ in the preceding 
statement 
to conclude $\partialr^{(k)}(f^{(l)})=f^{(l-k)}=\partialr(f^{(l-k+1)})$ is $C^{l-k}(V)$ and ${\bf o}(\normx^{l-k})$. This is item (i) of the Lemma.

The proof of the third point (iii) follows an analogous procedure. First, for
$f^{(n)}\in C^{n}(V)$ and ${\bf o}(\normx^{n})$ we have
$\frac{1}{\normr}f^{(n)}\in  {\bf o}(\normx^{n-1})$ by virtue of Lemma \ref{res:lemmaf8}.
We define a new $g$, now setting $g(\centresph_0) =0$ and $g\defi  \frac{1}{\normr}f^{(n)}$ outside $\centresph_0$ and show that $g\in C^{n-1}(V)$. Indeed,
for $0\leq |\alpha|\leq n-1$,
\[
\partial^\alpha \left (\frac{1}{\normr}f^{(n)} \right )
=\sum_{\beta\leq\alpha}b_\beta \partial^\beta \frac{1}{\normr} \partial^{\alpha-\beta} f^{(n)}.
\]
The terms in the sum can be written as
$\normr^{|\beta|+1}\partial^\beta\frac{1}{\normr}\frac{1}{\normr^{|\beta|+1}}\partial^{\alpha-\beta}f^{(n)}$,
and given that $\normr^{|\beta|+1}\partial^\beta\frac{1}{\normr}$ are bounded
and that $\partial^{\alpha-\beta}f^{(n)}\in {\bf o}(\normx^{n-|\alpha|+|\beta|})$
by Lemma \ref{res:lemma_di}, we get that 
each term belongs to
${\bf o}(\normx^{n-1-|\alpha|})$.
Therefore $\partial^\alpha \left ( \frac{1}{\normr}f^{(n)} \right ) \in {\bf o}(\normx^{n-1-|\alpha|})$
and thus  $\partial^\alpha \left (  \frac{1}{\normr}f^{(n)} \right ) \in o(\normr^{n-1-|\alpha|})$
by Lemma \ref{res:rhon-rn}. 
Since $|\alpha|\leq n-1$ the limits
of $\partial^\alpha \left (  \frac{1}{\normr}f^{(n)}  \right )$
vanish at $\normrho{x}=0$, and thus $\normr=0$ in particular.
The rest of the proof follow the same steps as for the previous $g$
and the same iteration process as for item (i).
\fin
\end{proof}

From here onwards we particularise to $d=3$, $d_z=1$ ($\Longrightarrow p=2$) and $d_t=1$.
Given  Cartesian coordinates in $B^2\times Z\subset\mathbb{R}^3$, $\{x_i\}=\{x_1,x_2,z\}$, $i=1,2,3$
(indices raised and lowered with $\delta_{ij}$), the seminorms read explicitly
$\normr\defi \sqrt{x^i x_i}$, $\normx=\sqrt{x^a x_a}=\sqrt{(x_1)^2+(x_2)^2}$.
Moreover we consider a ball $B^3\subset \mathbb{R}^3$ centered at the origin $\{0_3\}$
small enough so that $B^3\subset B^2\times Z$, and denote by $U$ the corresponding set in $V$, that is,
$U=B^3\times T\subset V$.

Introduce the additional vectors tangent to the $B^3$ factor of $U$
\begin{align}
&\ngamma\defi x^a\partial_a= x_1\partial_{x_1}+x_2\partial_{x_2},\quad \nzeta\defi \partial_{z},\quad \axial\defi x_1\partial_{x_2}-x_2 \partial_{x_1}, \mbox{ smooth in } U
\nonumber \\
&\oaxial\defi -\frac{1}{\normr}\left(z\ngamma-\normx^2\nzeta\right),  \mbox{ smooth in } U\setminus \centresph_0
\mbox{ and extends continuously } (\oaxial=0) \mbox{ to } \centresph_0.
\label{app:vectors}
    \end{align}
It is straightforward to check that $\partialr$, $\oaxial$ and $\axial$ are mutually orthogonal and commute.
Their explicit expressions in spherical coordinates
in the $B^3\setminus\{0_3\}$ factor of $U \setminus \centresph_0$ are given by  $\partialr=\partial_r$,
$\axial=\partial_\phi$ and $\oaxial=-\sin \theta\partial_\theta$.
Regarding the $T$ factor of $U$, it is now an interval in the real line, and
we will denote by $t$ both the points and coordinate within $T$. Clearly $\partial_t$ commutes with
$\partialr$, $\oaxial$ and $\axial$.

Let us continue with a result we will also need in the main text.
\begin{lemma}
\label{res:lemma_oaxial}
Let $f:V\to \mathbb{R}$ be $C^{l}$, $l\geq 1$,
and assume $f\in {\bf o}(\normx^{l})$. Then $\oaxial(f)\in {\bf o}(\normx^{l})$.
Also, $\normr \oaxial(f)\in C^{l-1}(V)$.
\end{lemma}
\begin{proof}
By virtue of (\ref{app:vectors}) we have
\[
\oaxial(f)=-\frac{1}{\normr}\left(z x^a \partial_af-\normx^2\partial_z f\right)=
-\frac{z}{\normr} x^a \partial_af+\frac{\normx}{\normr} \normx\partial_z f.
\]
As in the proof above, boundedness near the axis of  $x^a/\normx$ 
combined with
Lemma \ref{res:lemma_di} imply $ x^a\partial_af\in {\bf o}(\normx^l)$.
Boundedness of $z/\normr$ near the axis thus establishes the first term is 
${\bf o}(\normx^l)$.
For the second term, Lemma \ref{res:lemma_di}
also ensures that  $\partial_z f\in {\bf o}(\normx^l)$, and hence
boundedness of $\normx/\normr$ 
leads to the first result. The final statement follows from the $C^{\infty}$ smoothness
of both factors $z x^a$ and $\normx^2$.\fin
\end{proof}

The main result of this Appendix is the following lemma, which
concerns radially symmetric functions
in $\{x_1,x_2\}$ (see Appendix \ref{app:diff_origin}), that is,
functions invariant under $\axial$. We refer to these functions also
as ``axially symmetric''. The corresponding trace functions
are  defined in the domain
$U_{\eta} := \{ \rho \in \mathbb{R}_{\geq 0};
(\rho,0,z,t ) \in U\} \subset \mathbb{R}_{\geq 0} \times \mathbb{R}^{2}$.
We will also use the notation introduced in Section \ref{sec:spher_symm}
regarding the spheres $S_r\defi\{x\in U; t=const.,\normr=r>0\}$.
Observe that we suppress the label $t$ in $S_r$.
As in the main text, given any vector field $V$
in $\mathbb{R}^3$, we write $\vect{V}$ for its
tangential projection to the spheres $S_r$. The operator
$\dsph$ acting on a scalar $\beta$ gives the differential of its pull-back
on each of the spheres. We also let $\la,\ra$
denote the euclidean metric on $\mathbb{R}^3$ and observe that
$\la V,\cdot\ra= \norm{x}^2 \gsph( \vect{V},\cdot )$, where
$\gsph$ is the standard round unit metric. Note also that
$\axial$ and $\oaxial$ coincide with $\vect{\axial}$ and $\vect{\oaxial}$
respectively, and
$\la \axial, \axial\ra = \la \oaxial, \oaxial \ra = \normx^2$.

We have the ingredients to  state the key result of the Appendix.
In the first item of the lemma we use the big-O notation with its usual meaning.

\begin{lemma}
\label{res:lemma_oaxial_equation}
\begin{enumerate} Let $k$, $l$ be non-negative integers.
\item Let $P(z,t)$ be  $C^m(U)$ with $m\geq 0$. Then the equation on $U$
\begin{equation}
  \label{eq:for_gamma_P}
  \norm{x} \oaxial(Z_P)=\normrho{x}^{2k} z^lP(z,t)
\end{equation}
with $k\geq 1$ admits an axially symmetric solution
$Z_P\in C^{m}(U)$
and $O(\norm{x}^{2k+l-1})$. Moreover, if $m\geq 1$, $\partialr(Z_P)\in O(\normr^{2k+l-2})$. 
\item Let  $ \Phi^{(m)}$ be radially symmetric in $\{x_1,x_2\}$,  $C^m(U)$
and ${\bf o}(\normrho{x}^{m})$ with $m\geq 1$.
Then the equation on $U$
\begin{equation}
  \label{eq:for_gamma_f}
   \norm{x} \oaxial(\Zphi)=\normrho{x}^{2k} z^l \Phi^{(m)}
\end{equation}
with $2k+l\geq 1$ admits an axially symmetric solution
$\Zphi\in C^{m}(U\setminus \centresph_0)\cap C^0(U)$  and $o(\norm{x}^{2k+l+m-1})$.
Moreover, $\partialr(\Zphi)\in o(\normr^{2k+l+m-2})$.
\item Let  $ \Phizero^{(m)}$ be radially symmetric in $\{x_1,x_2\}$,
$C^m(U\setminus\centresph_0)\cap C^0(U)$ and $O(\normr)$ with $m\geq 0$.
Then the equation on $U$
\begin{equation}
  \label{eq:for_gamma_f_zero}
     \norm{x} \oaxial(\Zzero)= \normx^{2} \Phizero^{(m)}
\end{equation}
admits an axially symmetric solution
$\Zzero\in C^m(U\setminus\centresph_0)\cap C^{0}(U)$  and $O(\norm{x}^{2})$
which is also $C^{m+1}$ on each sphere $S_r$.
Moreover, $\partialr(\Zzero)$ and $\partial_t(\Zzero)$ are
$C^{m}$ on each sphere $S_r$.
\end{enumerate}
\vspace{2mm}
\end{lemma}
\begin{proof}
Since the inhomogeneous (source) terms of all the equations are invariant under $\axial$,
and we only care about existence, it will suffice to consider solutions
which are radially symmetric in $\{x_1,x_2\}$.
We concentrate first on the problem for $Z_P$.
Clearly $\ngamma$ (see (\ref{app:vectors})) acts
on radial functions in $\{x_1,x_2\}$ as $\ngamma(g(\normx))=\rho\partial_\rho g(\rho)|_{\rho=\normx}$. 
Note that the trace of the function $\norm{x}$ is $\sqrt{\rho^2+z^2}$.
Equation (\ref{eq:for_gamma_P}) for $Z_P$ is cast in terms of the traces in $U_{\eta}$ as
\begin{equation}
-\rho(z\partial_\rho-\rho\partial_z) \trace{Z}_P(\rho,z,t)=
\rho^{2k} z^l P(z,t).
\label{eq:for_alpha_2}
\end{equation}
It is direct to check that
\[
\trace{Z}{}_P(\rho,z,t)=\int^z_0(\rho^2+z^2-s^2)^{k-1}s^l P(s,t)ds
\]
satisfies the equation. The assumption $k\geq 1$
ensures that the integrand is a polynomial in even powers of $\rho$,
which by Lemma \ref{origin} ensures, in turn, that $Z_P$ is
$C^m(U)$ (i.e. the differentiability of $P$). 
In terms of functions on $U$
\[
Z_P=\int^z_0(|x|^2-s^2)^{k-1}s^l P(s,t)ds
\]
satisfies (\ref{eq:for_gamma_P}).
In spherical coordinates, and using the change of variable $s=\lambda r$
in the integral, we can reexpress $Z_P$ as
\begin{equation}
Z_P(r,\theta,\phi)=r^{2k+l-1}\int^{\cos\theta}_0(1-\lambda^2)^{k-1}\lambda^l P(\lambda r,t)d\lambda.
\label{eq:Zp_spher}
\end{equation}
Boundedness of the integral for $k\geq 1$ establishes that
$Z_P \in O(|x|^{2k+l-1})$ as claimed.
We compute now a radial derivative using
  (\ref{eq:Zp_spher}), taking into account that 
  in spherical coordinates $\partialr=\partial_r$,   to obtain
  \[
    \partialr(Z_P)(r,\theta,\phi)=(2k+l-1)\frac{Z_P}{r}+r^{2k+l-1}\int^{\cos\theta}_0(1-\lambda^2)^{k-1}\lambda^{l+1} \partial_zP (\lambda r,t)d\lambda.
  \]
  This time boundedness of the integral for $k\geq 1$ implies $\partialr(Z_P)\in O(\normr^{2k+l-2})$.

We proceed with the second point. The proof rests on an iteration
based on the following claim.

\vspace{3mm}

{\bf Claim.}  Fix integers $m\geq 1$ and $0\leq s\leq m$  and let
  $f^{(s)}\in C^{s}(U)$   be a radially symmetric
  function (in $\{x_1,x_2\}$) satisfying $\normr f^{(s)} \in {\bf o}(\normx^{s+1})$ for $0 \leq s \leq m-1$ and $f^{(m)} \in {\bf o}(\normx^{m})$.
  Let $l,k$ be non-negative integers satisfying $2k + l\geq 1$ and define
  $g^{(s)}=\normx^{2k}z^l  f^{(s)}$.
Then the equation
\begin{equation}
  \label{eq:for_beta_s}
   \normr \oaxial(\beta)=g^{(s)} 
\end{equation}
admits a unique axially symmetric solution $\beta=\widehat\beta{}_s$
such that its restriction to $S_r$ is $C^{s}$
and satisfies the boundary condition
$\widehat\beta{}_s|_{\theta=0}=0$. Moreover, this function is 
also $C^0$ with respect to $r$  and $t$ for $r>0$ and satisfies
$\widehat\beta{}_s\in o (\normr^{2k+l+s-1})$, so that in particular it
extends
continuously to $\centresph_0$ as $\widehat\beta{}_s|_{\centresph_0}=0$.

\vspace{3mm}

\emph{Proof of the claim}.
  We  observe that the assumptions imply (since $\normx/\normr$ is bounded)
\begin{align}
&  f^{(s)} \in {\bf o} (\normx^{s}), \quad \quad
  \normr f^{(s)} \in {\bf o}(\normx^{\min\{s+1,m\}}), \quad \quad
  0 \leq s \leq m, \nonumber \\
&  f^{(m)}/\normx \in o (\normr^{m-1}), \quad \quad \quad 
  \normr f^{(s)}/\normx \in {\bf o}(\normr^{s}), \quad \quad
  0 \leq s \leq m-1, \label{condsf}
\end{align}
the second line following form the first combined with Lemma \ref{res:rhon-rn}
(and $m \geq 1$).
Clearly $g^{(s)}\in C^s(U)$ and, since $\normx^{2k} z^l /\normr$ is also bounded,
$g^{(s)} \in {\bf o}(\normx^{\min\{s+1,m\}})$.
We start by writing (\ref{eq:for_beta_s})
in spherical coordinates in $U\setminus\centresph_0$
recalling that $\oaxial=-\sin\theta\partial_\theta$ and $g^{(s)}$ is invariant under
$\axial=\partial_\phi$, as
\begin{equation}
\partial_\theta\beta=-\frac{1}{r\sin\theta}g^{(s)}(\theta,r,t) \quad \mbox{ for }\quad r>0.
\label{eq:beta_spher}
\end{equation}
We may write the solution $\widehat\beta{}_s$ that satisfies
$\widehat\beta{}_s|_{\theta=0}=0$ as 
\begin{equation}
\widehat\beta{}_s(\theta;r,t)=-\int^\theta_0 \frac{g^{(s)}(\lambda,r,t)}{r\sin\lambda}d\lambda.
\label{eq:int_beta}
\end{equation}
Since in particular $g^{(s)}\in {\bf o}(\normx^{1})$ the integrand is bounded in all the domain of integration
and thus the integral exists and $\widehat\beta{}_s(\theta;r,t)$ is continuous in $\theta$
up to the boundary. Clearly, $\widehat\beta{}_s$ thus constructed is $C^0$ on each $S_r$.
On the other hand, since $g^{(s)}\in C^s(U)$ then 
$g^{(s)}$ is $C^{s}$ in $r$ and $t$ for $r>0$, and
therefore $\widehat\beta{}_s(\theta;r,t)$
is also
$C^{0}$ in $r$  and $t$ for $r>0$.
Moreover,  at each fixed $t$ we have
\begin{align}
\frac{\norm{\widehat\beta{}_s}(\theta;r,t)}{r^{2k+l+s-1}}
\leq\frac{1}{r^{2k+l+s-1}}\int^\theta_0  \left|\frac{g^{(s)}}{r\sin\lambda} \right| d\lambda
  \leq \frac{\theta}{r^{2k+l+s-1}} \sup_{\normr=r}\left|\frac{g^{(s)}}{\normx} \right|\nonumber \\
  \leq \frac{\theta}{r^{s-1}} \sup_{\normr=r}\left|\frac{f^{(s)}}{\normx} \right|
  = \frac{\theta}{r^{s}} \sup_{\normr=r}\left|\frac{\normr f^{(s)}}{\normx} \right|,
  \label{ineqs}
\end{align}
where we have used the explicit form
of $g^{(s)}=\normx^{2k}z^l f^{(s)}$ and
$\sup_{\normr=r}(\normx^{2k} |z^l|) \leq r^{2k+l}$ in the last inequality.
Properties (\ref{condsf}) imply in particular that
  $\normr f^{(s)} /\normx \in o (\normr^{s})$ for all $s \in \{0, \cdots, m\}$, so, the last term is $o(1)$ and we conclude
\begin{align}
\label{limitbeta}
  \lim_{\normr\to 0}\frac{\norm{\widehat\beta{}_s}}{\normr^{2k+l+s-1}}=0,
  \quad \quad 0 \leq s \leq m.
\end{align}

The claim for $s=0$ has been shown. In the following we deal with $s\geq 1$.
In order to obtain the differentiability of the solutions $\widehat\beta{}_s$
along the spheres $S_r$ one could differentiate \eqref{eq:int_beta}
repeatedly and take control over the behaviour of the terms around the axis.
That is trivial for the first derivative, but for $s\geq 2$ we follow a more
straightforward strategy. 
Note also that with the integral expression only
we cannot extract any sort of differentiability (nor continuity) of
$\partial_r\widehat\beta{}_s(\theta;r,t)$ on $S_r$ because the fact that
$g^{(s)}$ is ${\bf o}(\normx^s)$ does not translate to $\partial_r g^{(s)}$ in any way in general.

A convenient form of expressing equation (\ref{eq:for_beta_s}) is
\begin{align}
  \normr\d \beta(\oaxial)=\frac{1}{\la \axial,\axial \ra} g^{(s)}\la\oaxial,\oaxial\ra,
  \label{Eqbeta}
\end{align}
where at the right hand side we are just using $\la\oaxial,\oaxial\ra=
\la \axial,\axial \ra$.
Clearly $\d\beta(\oaxial)=\dsph\beta(\vect{\oaxial})$ outside the origin.
On the other hand, since we are constructing solutions invariant under $\axial$,
it is enough to restrict
$\beta$ so that $\d\beta(\axial)=\dsph\beta(\vect{\axial})=\axial(\beta)=0$.
Therefore, since $\la\oaxial,\axial\ra=0$,
 equation (\ref{Eqbeta}) is equivalent to
\begin{equation}
\dsph \beta(\cdot)
=\frac{1}{\normr}\frac{1}{\trho^2} g^{(s)} \la\oaxial,\cdot\ra
= \frac{\norm{x}}{\normx^2}  g^{(s)} \gsph(\vect{\oaxial},\cdot)
=:\gsph(V,\cdot)
\label{eq:for_beta}
\end{equation}
on each sphere $S_r$,  where we have defined the vector 
\[
  V := \frac{\normr}{\normx^2}g^{(s)} \vect{\oaxial} 
  =-\normx^{-2}g^{(s)}(zx_1\partial_{x_1}+zx_2\partial_{x_2}-\normx^2\partial_z).
\]
and (\ref{app:vectors}) has been used in the second equality.
We deal first with
the regularity and behaviour around the axis of the vector $V$.
The applications of $V$ to the Cartesian coordinate functions (observe $V(t)=0$) provide
\[
V(x_1)=-z \frac{x_1}{\normx}g^{(s)} \frac{1}{\normx},\quad
V(x_2)=-z \frac{x_2}{\normx}g^{(s)} \frac{1}{\normx},\quad
V(z)=g^{(s)}.
\]
Since $g^{(s)}$ is ${\bf o}(\normx^s)$ in particular, Lemma \ref{res:cart_product} 
ensures that 
$\frac{x_1}{\normx}g^{(s)} $ and  $\frac{x_2}{\normx}g^{(s)} $
are $C^s$ and ${\bf o}(\normx^s)$.
We apply
next Lemma \ref{Analy} to these two functions to find that
$\frac{1}{\normx}g^{(s)} $ is $C^{s-1}$. Given that, also,
$\frac{1}{\normx}g^{(s)} \in{\bf o}(\normx^{s-1})$
we finally apply  Lemma \ref{res:cart_product} to $\frac{1}{\normx}g^{(s)} $ to conclude that $\frac{x_2}{\normx}\frac{1}{\normx}g^{(s)}$ and
$\frac{x_1}{\normx}\frac{1}{\normx}g^{(s)}$ are $C^{s-1}$ and ${\bf o}(\normx^{s-1})$.
Therefore $V(x_1)$ and $V(x_2)$ are $C^{s-1}$ and ${\bf o}(\normx^{s-1})$ in $U$,
while $V(z)$ is clearly $C^s$.
As a result, $V$ is a $C^{s-1}$ vector field on
$U$, thus on each $S_r$.

For $s=1$ it suffices now to trivially use the equation \eqref{eq:for_beta}, which applied to
$\partial_{x^i}=\{\partial_{x_1}, \partial_{x_2},\partial_{z}\}$ gives
\[
\dsph \beta(\partial_{x^i})=V(x^i),\quad  \dsph \beta(\partial_t)=0.
\]
Given the above, any solution $\beta$ of \eqref{eq:for_beta} invariant under
$\axial$ satisfies $\dsph \beta\in C^0(S_r)$. Therefore
we have that $\widehat\beta{}_1\in C^1(S_r)$ in particular.

We deal now with the case $s\geq 2$.
The application of the divergence ($\starsphere\dsph\starsphere$)
at both sides of \eqref{eq:for_beta} yields
\begin{align}
\Delta_{\mathbb{S}^2}\beta&
=\mbox{div}_{\mathbb{S}^2}V
\label{eq:laplace_beta}
\end{align}
on each $S_r$.
Since we are dealing with a one-dimensional
problem
(because $\axial(\beta)=0$) and $V$ is a $C^{s-1}$ vector field 
on each $S_r$,
the solutions $\beta$ of the equation (\ref{eq:laplace_beta}) for $s\geq 2$
must be $C^{s}$ on each $S_r$.\footnote{Without axial symmetry
we would need that the inhomogeneous term is Hölder $C^{s,\alpha}$.}
Observe that a solution $\beta$ of \eqref{eq:laplace_beta}
is unique up to an additive constant at each $S_r$, hence up to a
radially symmetric (and $t$-dependent) function. 
Therefore, the solution $\widehat{\beta}{}_s(r,t)$ (for $r>0$)  for each $s$
constructed above is the unique solution at each $S_r$ and $t$,
fixed by the condition $\widehat{\beta}{}_s|_{\theta=0}(r,t)=0$, i.e. vanishing at the north
poles of each $S_r$, and 
satisfies $\widehat\beta{}_s\in C^s(S_r)$.
This finishes the proof of the claim.\fin

\vspace{3mm}
We deal now with the original equation (\ref{eq:for_gamma_f}). Applying the claim to $s=m\geq 1$ in (\ref{eq:for_beta_s}) 
and $g^{(m)}=\normx^{2k}z^l\Phi^{(m)}$
and setting $\Zphi(r,t)=\widehat \beta{}_m(r,t)$, we conclude that this function is $C^m$ on each
$S_r$ and $o(\normr^{2k+l+m-1})$. It remains to show that  $\Zphi$
is also  differentiable in $r$ and $t$, which will then imply
$\Zphi\in C^{m}(U\setminus \centresph_0)$.

We first consider the radial derivatives. 
Let us set  $f^{(m)}\defi\Phi^{(m)}$ and define $f^{(s)}$  by the iteration
 \begin{equation}
f^{(s-1)}\defi(2k+l-1)\frac{1}{\normr}f^{(s)}+\partialr(f^{(s)})\quad \mbox{ for } \quad 1\leq s\leq m.
\label{def:f_s-1}
\end{equation}
Points $(i)$ and $(iii)$ in Lemma \ref{res:lemma_nkf} show, after a 
  trivial iteration 
  starting at $f^{(m)} \in C^m(U)$ and $\bm{o} (\normx^m)$,
  that $f^{(s)}$ is $C^{s}(U)$ and ${\bf o}(\normx^{s})$ for $0\leq s \leq m$.
In addition, since
  $\normr \partialr(f^{(s)}) \in
  {\bf o}(\normx^{(s)})$
  (by point $(ii)$ in Lemma \ref{res:lemma_nkf}),
\eqref{def:f_s-1} implies $\normr f^{(s)}  \in {\bf o}(\normx^{s+1\}})$ for all
$0\leq s \leq m-1$.
The point of introducing these functions is that $\partialr^{(m-s)} (\Zphi)$
with $s\in\{0,\ldots,m\}$ satisfies the equation
\begin{align}
  \normr \iota(\partialr^{(m-s)} (\Zphi) ) = \normx^{2k} z^l f^{(s)}.
 \label{iota_hatns}
\end{align}
We show this by iteration. The statement is clearly true for $s=m$ (the original
equation). Assume it is true for a given $s \in \{1, \cdots, m\}$
and apply $\partialr$ to (\ref{iota_hatns}). Using
 $\partialr(\norm{x}) =1$, $\partialr(\normx)= \norm{x}^{-1} \normx$ and
$\partialr(z)= \norm{x}^{-1} z$, together with the fact that $\partialr$ and $\iota$ commute yields 
\begin{align*}
  \iota(\partialr^{(m-s)} (\Zphi) ) +
  \norm{x} \iota
  (\partialr^{(m-s+1)} (\Zphi)) & =
  \normx^{2k} z^l \left ( (2k+l) \frac{f^{(s)}}{\norm{x}} + \partialr (f^{(s)}) \right )
  \quad \quad \Longrightarrow \\
  \norm{x} \iota (\partialr^{(m-(s-1))} (\Zphi)) & = \normx^{2k} z^l \left ( (2k+l-1 ) \frac{f^{(s)}}{\norm{x}} + \partialr (f^{(s)}) \right ) =
\normx^{2k} z^l f^{(s-1)},
\end{align*}
where in the second line we inserted the equation for
$\iota(\partialr^{(m-s)} (\Zphi) )$ and the definition of $f^{(s-1)}$. Thus,
equation (\ref{iota_hatns}) is also true for $s-1$, and hence for all $s \in \{
0,\cdots, m\}$.

Moreover, the function $\partialr^{(m-s)} (\Zphi)$ vanishes  on the line $\theta=0$.
Thus, all the conditions of the claim are satisfied for equation (\ref{iota_hatns})
and we conclude that 
$\partialr^{(m-s)} (\Zphi)$ is $C^{s}$ on each sphere and $o(\norm{x}^{2k+l+s-1})$.
This proves that all radial derivatives
$\partial_r^{(m-s)} \Zphi$ for $s\in\{0,\ldots,m\}$ are $C^{s} (U \setminus \centresph_0)$
and, moreover, extend continuously to $r=0$ with the value zero.

Regarding the  derivatives with respect
to $t$, a similar (in fact much easier) argument applies.
Indeed, $\partial^{(m-s)}_t \Zphi$, $s \in \{0, \cdots, m\}$ satisfies the equation
\begin{align*}
  \norm{x} \iota \left ( \partial^{(m-s)}_t \Zphi \right )
  = \normx^{2k} z^l \partial_t^{(m-s)} \Phi^{(m)}.
\end{align*}
Lemma \ref{res:lemma_di} ensures that
$f^{(s)}\defi \partial_t^{(m-s)} \Phi^{(m)}$ is $C^{s}(U)$ and ${\bf o}(\normx^{m})$,
and this also implies in particular
$\normr f^{(s)}\in {\bf o}(\normx^{s+1})$, $s \in \{0,\cdots, m-1\}$.
We may apply our main claim to
$g^{(s)}= \normx^{2k} z^l  f^{(s)}$
to conclude that $\partial^{(m-s)}_t\Zphi$ is $C^{s}$ on each $S_r$ and
$o(\normr^{2k+l+s-1})$ for $s\in\{0,\ldots,m\}$.

  To sum up,  $\Zphi(r,t)$  is $C^m$ on each $S_r$ and it is also $C^m$ with respect to the parameters $\{r,t\}$ for $r>0$ and $t\in T$. We may coordinate the sphere with  two charts so that, together  with $r$ and $t$, we also cover 
$U\setminus \centresph_0$ with  two coordinate charts. The coordinate change 
to cartesian coordinates is smooth on $U \setminus \centresph_0$, so we have proved that (\ref{eq:for_beta_s}) on $U$ admits an axially symmetric solution
$\Zphi\in C^m(U\setminus\centresph_0)$ (as function of $\{x_1,x_2,z,t\}$),
and $o(\normr^{2k+l+m-1})$.

The proof of point 3 is more direct than the previous one because the right-hand side of the equation has better behaviour near the axis.
We first consider $m\geq 1$, for which
equation \eqref{eq:for_gamma_f_zero} has the form of
  \eqref{eq:for_beta_s} with a $C^1(U)$ right-hand side. We therefore recover \eqref{eq:laplace_beta}  in the form
\[
\Delta_{\mathbb{S}^2}\Zzero
=\mbox{div}_{\mathbb{S}^2}V,
\]
where the vector $V$ is defined as in \eqref{eq:for_beta} with $g^{(s)}$ replaced by  $\normx^2 \Gamma^{(m)}$, i.e.
$V=
\normr\Phizero^{(m)} \vect{\oaxial}=-\Phizero^{(m)}(zx_1\partial_{x_1}+zx_2\partial_{x_2}-\normx^2\partial_z)$.
Clearly, $V$ has the same differentiability of $\Phizero^{(m)}$,
that is, $C^m(U\setminus\centresph_0)$,
and therefore $V$ is $C^m$ on each
$S_r$. As a result (see above), the solution $\Zzero$ is a  $C^{m+1}$ function
on each $S_r$. We denote by $\widehat\Zzero(r,t)$ the solution for $r>0$ that
vanishes at the north poles of each $S_r$. Using spherical coordinates, as done previously,
we can express that solution by
\begin{equation}
\widehat\Zzero(\theta;r,t)=-\int^\theta_0r\sin\lambda \Phizero^{(m)}(\lambda,r,t)d\lambda.
\label{eq:int_zzero}
\end{equation}
This integral expression directly shows that
$\widehat\Zzero(\theta;r,t)$ is $C^m$ with respect to $\{r,t\}$ by construction,
which together with $\widehat\Zzero\in C^{m+1}(S_r)$, implies
$\widehat\Zzero\in C^m(U\setminus\centresph_0)$.
To obtain the differentiability on each $S_r$ of both
$\partialr(\Zzero)$ and $\partial_t(\Zzero)$
we differentiate \eqref{eq:for_gamma_f_zero} accordingly, to get
\[
\norm{x} \oaxial(\partialr(\Zzero))= \normx^{2} \left(\frac{\Phizero^{(m)}}{\normr}+\partialr(\Phizero^{(m)})\right),\quad
\norm{x} \oaxial(\partial_t\Zzero)= \normx^{2} \left(\partial_t\Phizero^{(m)}\right).
\]
The terms within brackets at
right hand sides in both equations are $C^{m-1}(U\setminus\centresph_0)$, and the same argument
as above involving the Laplace equation
shows that $\partialr(\Zzero)$ and $\partial_t(\Zzero)$ are $C^{(m-1)}(S_r)$.

It only remains to  consider the case $m=0$. It suffices
to use the integral expression
\eqref{eq:int_zzero} which is obviously still valid. It is immediate 
that the integral is
continuous in $\{\theta,r,t\}$
for $r >0$, from where it follows easily that
$\widehat\Zzero \in C^0(U\setminus\centresph_0)$. Clearly
this function can be differentiated once with respect to
$\theta$ on $\theta \in [0,\pi]$, with continuous derivative. Again, we easily conclude that $\widehat\Zzero\in C^1(S_r)$ on each $S_r$.

It only remains to show that, irrespectively of the value of $m \geq 0$,
  $\widehat{\Zzero} \in (\normr^2)$, as this already implies that this function extends continuously to $\centresph_0$. We compute
\[
\frac{\norm{\widehat\Zzero}(\theta;r,t)}{r^{2}}
\leq\int^\theta_0  \left|\sin\lambda\frac{\Phizero^{(m)}}{r} \right| d\lambda
  \leq \theta \sup_{\normr=r}\left|\frac{\Phizero^{(m)}}{\normr} \right|.
\]
By the assumption $\Phizero^{(m)}\in O(\normr)$, the last term is bounded
and the property follows.
\fin
\end{proof}

\begin{corollary}
\label{res:main_coro}
Let $\V$, $\V'$ be finite subsets of $\mathbb{N}\times
  \mathbb{N}$ with the restrictions $\V \subset \{k \geq 1\} \times
\{ l \geq 0\}$ and $\V' \subset \{ 2k'+ l' \geq 1 \}$
and define $\displaystyle{b:= \min_{\V}\{2k +l\}, c :=
  \min_{\V'} \{ 2k'+l'\}}$. Consider the equation on $U$
\begin{equation}
  \label{eq:for_gamma_general}
  \oaxial(\tlemma)=
  \sum_{(k,l) \in \V}\normrho{x}^{2k} z^lP_{lk}(z,t)
  +\sum_{(k',l') \in \V'} \normrho{x}^{2k'} z^{l'} \Phi^{(m)}_{l'k'}(x),
\end{equation}
where $P_{lk}$ are $C^m$ in their arguments and $\Phi^{(m)}_{l'k'}$ are $C^m(U)$ and ${\bf o}(\normx^m)$.
If $m\geq 1$ the equation
admits an axially symmetric solution $$\tlemma=\normr(Z_P+\Zphi)$$ where $Z_P+\Zphi\in C^m(U\setminus \centresph_0)$ with $Z_P\in O(\normr^{b-1})$ and
$\Zphi\in o(\normr^{c+m-1})$. Moreover, $\partialr(Z_P)\in O(\normr^{b-2})$ and
$\partialr(\Zphi)\in o(\normr^{c+m-2})$.
In particular, 
\[
\tlemma\in O(\normr^{b}),\quad \partialr(\tlemma)\in O(\normr^{b-1}) \quad \mbox{ provided } b\leq c +m,
\]
\[
\tlemma\in o(\normr^{c+m}),\quad \partialr(\tlemma)\in o(\normr^{c+m-1}) \quad \mbox{ for } b > c +m.
\]
\end{corollary}

\begin{proof}
  Write $\gamma = \norm{x} \widehat{Z}$ so that $\widehat{Z}$ solves
  \begin{align*}
    \norm{x} \iota (\widehat{Z}) =
    \sum_{(k,l) \in \V}\normrho{x}^{2k} z^lP_{lk}(z,t)
+\sum_{(k',l') \in \V'}\normrho{x}^{2k'} z^{l'} \Phi^{(m)}_{l'k'}(x),
  \end{align*}
  Since the equation is linear the solution decomposes into a sum and
  we may apply Lemma \ref{res:lemma_oaxial_equation} to each term.
  The rest 
  is immediate.\fin
\end{proof}

\begin{corollary}
\label{res:main_coro_0}
With the above definitions (in particular
$\Gamma^{(m-1)}$ is axially symmetric, $C^{m-1}(U \setminus \centresph_0)
\cap C^0(U)$ and $O(|x|)$, cf. point \emph{3.} in Lemma \ref{res:lemma_oaxial_equation})
the equation on $U$
\begin{equation}
  \label{eq:for_gamma_general_zero}
  \oaxial(\tlemma_0)=
  \sum_{(k,l) \in \V}\normrho{x}^{2k} z^lP_{lk}(z,t)
+\sum_{(k',l') \in \V'}\normrho{x}^{2k'} z^{l'} \Phi^{(m)}_{l'k'}(x) +\normx^{2}\Phizero^{(m-1)},
\end{equation}
with $m\geq 1$
admits an axially symmetric solution
$\tlemma_0\in C^{m-1}(U\setminus\centresph_0)\cap C^0(U)$,
which is also $C^m$ on each sphere $S_r$,
and moreover $\partialr(\tlemma_0)$ and $\partial_t \tlemma_0$ are
$C^{m-1}$ on each sphere $S_r$.
In addition, $\tlemma_0 \in  O(\normr^{3})$ provided
$3\leq b\leq c +m$.
\end{corollary}

\begin{proof}
We use the previous corollary, which ensures the existence of the solution $\tlemma$,
followed by the addition $\tlemma_0=\tlemma+\normr\Zzero$ of the solution $\Zzero$ for
\eqref{eq:for_gamma_f_zero} given by the third point of the lemma.\fin
\end{proof}

\bibliography{references}{}
\bibliographystyle{review_bib}

\end{document}